\documentclass[12pt,a4paper]{amsart}
\usepackage{amsmath,amssymb,a4wide,amsthm,ifthen,hyperref,xcolor,enumerate,enumitem}
\usepackage[utf8]{inputenc}

\usepackage{graphicx}
\usepackage{url}
\usepackage{ulem}
\usepackage{epstopdf}
\usepackage[section]{placeins}
\usepackage{natbib}
\usepackage{flafter} 
\usepackage{multirow}

\usepackage[linesnumbered,ruled,vlined]{algorithm2e}

\SetKwInput{KwInput}{Input}                
\SetKwInput{KwOutput}{Output}

\newtheorem{theo}{Theorem}
\newtheorem{prop}{Proposition}
\newtheorem{lemma}{Lemma}
\theoremstyle{remark}

\begin{document}

\title[Variable selection in sparse GLARMA models]{Variable selection in sparse GLARMA models}

\date{\today}

\author{M. Gomtsyan}
\address{UMR MIA-Paris, AgroParisTech, INRAE, Universit\'e Paris-Saclay, 75005, Paris, France}
\email{mgomtsian@gmail.com}
\author{C. L\'evy-Leduc}
\address{UMR MIA-Paris, AgroParisTech, INRAE, Universit\'e Paris-Saclay, 75005, Paris, France}
\email{celine.levy-leduc@agroparistech.fr}
\author{S. Ouadah}
\address{UMR MIA-Paris, AgroParisTech, INRAE, Universit\'e  Paris-Saclay, 75005, Paris, France}
\email{sarah.ouadah@agroparistech.fr}
\author{L. Sansonnet}
\address{UMR MIA-Paris, AgroParisTech, INRAE, Universit\'e Paris-Saclay, 75005, Paris, France}
\email{laure.sansonnet@agroparistech.fr}


\maketitle

\begin{abstract}
  In this paper, we propose a novel and efficient two-stage variable selection approach for sparse GLARMA models, which are pervasive for
  modeling discrete-valued time series. Our approach consists in iteratively combining the estimation of the autoregressive moving average (ARMA) coefficients of
  GLARMA models with regularized methods designed for performing variable selection
  in regression coefficients of Generalized Linear Models (GLM). We first establish the consistency of the ARMA part coefficient estimators in a specific case.
Then, we explain how to efficiently implement our approach.  Finally, we assess the performance of
our methodology using synthetic data and compare it with alternative methods. 
Our approach is very attractive since it benefits from a low computational load and is able to outperform the other methods in terms of coefficient estimation,
particularly
in recovering the non null regression coefficients.

  

\end{abstract}



\section{Introduction}

\textcolor{black}{Discrete-valued time series arise in a wide variety of fields ranging from finance to molecular biology and public health.
For instance, we can mention the number of transactions in stocks in the finance field, see \cite{brannas:quoreshi:2010}.
In the field of molecular biology, modeling RNA-Seq kinetics data is a challenging issue, see \cite{Thorne:2018} and in the public health
context, there is an interest in the modeling of daily asthma presentations in a given hospital, see \cite{SOUZA:2014}.}


The literature on modeling discrete-valued time series is becoming increasingly abundant, see \cite{handbook:2016} for a review. 
Different classes of models have been proposed such as the Integer Autoregressive Moving Average 
(INARMA) models and the generalized state space models. 

The Integer Autoregressive process of order 1 (INAR(1)) was first introduced by \cite{McKenzie:1985} and the Integer-valued Moving Average 
(INMA) process is described in \cite{Al-Osh:1988}. One of the attractive features of INARMA processes is that their autocorrelation structure is similar 
to the one of autoregressive moving average (ARMA) models. However, it has to be noticed that statistical inference in these models is generally complicated and requires to develop 
intensive computational approaches such as the 
efficient MCMC algorithm devised by \cite{Neal:rao:2007} for INARMA processes of known AR and MA orders. This strategy was extended 
to unknown AR and MA orders by \cite{enciso:nea:rao:2009}.
 For further references on INARMA models, we refer the reader to \cite{weiss:dts}. 





The other important class of models for discrete-valued time series is the one of generalized state space models which can have a parameter-driven and an observation-driven version, 
see \cite{davis:1999} for a review. 
The main difference between \textcolor{black}{these two versions} is that in parameter-driven models, the state vector evolves independently of the past history
of the observations whereas the state vector depends on the past observations in observation-driven models. More precisely, in parameter-driven models, let $(\nu_t)$ be a stationary process,
the observations $Y_t$ are thus modeled as follows: conditionally on $(\nu_t)$, $Y_t$ has a Poisson distribution of parameter $\exp(\beta_0^\star+\sum_{i=1}^p\beta_i^\star x_{t,i}+\nu_t)$,
where the $x_{t,i}$'s are the $p$ regressor variables (or covariates). Estimating the parameters in such models has a very high computational load, see \cite{jung:2001}.

Observation-driven models initially proposed by \cite{cox:1981} and further studied in \cite{zeger:qaqish:1988} do not have this computational drawback and are thus considered as a promising alternative to parameter-driven models.  Different kinds of observation-driven models can be found in the literature: 
the Generalized Linear Autoregressive Moving Average (GLARMA) models introduced by \cite{davis:1999}
 and further studied in \cite{davis:dunsmuir:streett:2003}, \cite{davis:dunsmuir:street:2005}, \cite{dunsmuir:2015} and the (log-)linear Poisson autoregressive models studied in 
\cite{fokianos:2009}, \cite{fokianos:2011} and \cite{fokianos:2012}. Note that GLARMA models cannot be seen as a particular case of the log-linear Poisson autoregressive models.

\textcolor{black}{In the following, we shall consider the GLARMA model introduced in \cite{davis:dunsmuir:street:2005} with additional covariates. More precisely,} 
given the past history $\mathcal{F}_{t-1}=\sigma(Y_s,s\leq t-1)$, \textcolor{black}{we assume that}
\begin{equation}\label{eq:Yt}
Y_t|\mathcal{F}_{t-1}\sim\mathcal{P}\left(\mu_t^\star\right),
\end{equation}
where $\mathcal{P}(\mu)$ denotes the Poisson distribution with mean $\mu$. In (\ref{eq:Yt}),
\begin{equation}\label{eq:mut_Wt}
\mu_t^\star=\exp(W_t^\star) \textrm{ with } W_t^\star=\beta_0^\star+\sum_{i=1}^p\beta_i^\star x_{t,i}+Z_t^\star,
\end{equation}
where the $x_{t,i}$'s are the $p$ regressor variables ($p\geq 1$),
\begin{equation}\label{eq:Zt}
Z_t^\star=\sum_{j=1}^q \gamma_j^\star E_{t-j}^\star \textrm{ with } E_t^\star=\frac{Y_t-\mu_t^\star}{\mu_t^\star}=Y_t\exp(-W_t^\star)-1,
\end{equation}
with $1\leq q\leq\infty$
and $E_t^\star=0$ for all $t\leq 0$. \textcolor{black}{Here, the $E_t^\star$'s correspond to the working residuals in classical Generalized Linear Models (GLM), which means that we limit ourselves to the case 
$\lambda=1$ in the more general definition:
$
 E_t^\star=(Y_t-\mu_t^\star){\mu_t^{\star}}^{-\lambda}
$. Note that in the case where $q=\infty$, $(Z_t^\star)$ satisfies the ARMA-like recursions given in Equation (4) of \cite{davis:dunsmuir:street:2005}.
The model defined by (\ref{eq:Yt}), (\ref{eq:mut_Wt}) and (\ref{eq:Zt}) is thus referred as a GLARMA model.}

The main goal of this paper is to introduce a novel variable selection approach in the deterministic part \textcolor{black}{(covariates)} of
sparse GLARMA models that is in (\ref{eq:Yt}), (\ref{eq:mut_Wt}) and (\ref{eq:Zt})
where the vector of the $\beta_i^\star$'s is sparse meaning many $\beta_i^\star$'s are null. \textcolor{black}{The novel approach that we propose consists in} combining a procedure 
for estimating the ARMA part coefficients with regularized methods designed for GLM.

The paper is organized as follows. Firstly, in Section \ref{sec:estim}, we describe the classical estimation procedure in GLARMA models
and in Section \ref{sec:consistency}, establish a consistency result in a specific case.
Secondly, we propose a novel two-stage estimation procedure which is described in Section \ref{sec:our_estim}. It consists
in first estimating the ARMA coefficients and then in estimating the regression coefficients 
by using a regularized approach.
\textcolor{black}{The practical implementation of our approach is given in Section \ref{sec:practical}.}
Thirdly, in Section \ref{sec:num}, we provide some numerical experiments to illustrate our method and to compare its performance to alternative approaches
on finite sample size data. Finally, we give the proofs of the theoretical results in Section \ref{sec:proofs}.

\section{Statistical inference}\label{sec:stat_inf}

\subsection{Classical estimation procedure in GLARMA models}\label{sec:estim}


Classically, for estimating the parameter $\boldsymbol{\delta}^\star=(\boldsymbol{\beta}^{\star\prime},\boldsymbol{\gamma}^{\star\prime})$
where $\boldsymbol{\beta}^\star=(\beta_0^\star,\beta_1^\star,\dots,\beta_p^\star)'$ is the vector of regressor coefficients defined in (\ref{eq:mut_Wt})
and $\boldsymbol{\gamma}^\star=(\gamma_1^\star,\dots,\gamma_q^\star)'$ is the vector of the ARMA part coefficients defined in (\ref{eq:Zt}), the following criterion,
based on the conditional log-likelihood, is maximized with respect to $\boldsymbol{\delta}=(\boldsymbol{\beta}',\boldsymbol{\gamma}')$, with
$\boldsymbol{\beta}=(\beta_0,\beta_1,\dots,\beta_p)'$ and $\boldsymbol{\gamma}=(\gamma_1,\dots,\gamma_q)'$:
\begin{equation}\label{eq:likelihood}
L(\boldsymbol{\delta})=\sum_{t=1}^n\left(Y_t W_t(\boldsymbol{\delta})-\exp(W_t(\boldsymbol{\delta}))\right).
\end{equation}
In (\ref{eq:likelihood}),
\begin{equation}\label{eq:Wt}
W_t(\boldsymbol{\delta})=\boldsymbol{\beta}'x_t+Z_t(\boldsymbol{\delta})=\beta_0+\sum_{i=1}^p\beta_i x_{t,i}+\sum_{j=1}^q \gamma_j E_{t-j}(\boldsymbol{\delta}),
\end{equation}
with $x_t=(x_{t,0},x_{t,1},\dots,x_{t,p})'$, $x_{t,0}=1$ for all $t$ and

\begin{eqnarray}
E_t(\boldsymbol{\delta})=Y_t\exp(-W_t(\boldsymbol{\delta}))-1,\mbox{ if }t>0\mbox{ and }E_t(\boldsymbol{\delta})=0\mbox{, if }t\leq 0.
\label{eq:Et}
\end{eqnarray}
%
For further details on the choice of this criterion, we refer the reader to \cite{davis:dunsmuir:street:2005}.

To obtain $\widehat{\boldsymbol{\delta}}$ defined by
\begin{equation*}
\widehat{\boldsymbol{\delta}}=\textrm{Argmax}_{\boldsymbol{\delta}} \; L(\boldsymbol{\delta}),
\end{equation*}
the first derivatives of $L$ are considered:
\begin{equation}\label{eq:def:grad}
\frac{\partial L}{\partial \boldsymbol{\delta}}(\boldsymbol{\delta})=\sum_{t=1}^n(Y_t-\exp(W_t(\boldsymbol{\delta}))\frac{\partial W_t}{\partial \boldsymbol{\delta}}(\boldsymbol{\delta}),
\end{equation}
where 
\begin{equation*}
\frac{\partial W_t}{\partial \boldsymbol{\delta}}(\boldsymbol{\delta})=\frac{\partial\boldsymbol{\beta}' x_t}{\partial \boldsymbol{\delta}}+\frac{\partial Z_t}{\partial \boldsymbol{\delta}}
(\boldsymbol{\delta}),
\end{equation*}
$\boldsymbol{\beta}$, $x_t$ and $Z_t$ being given in (\ref{eq:Wt}). 
The computations of the first derivatives of $W_t$ are detailed in Section \ref{subsub:first_derive}. 

Based on Equation (\ref{eq:def:grad}) which is non linear in $\boldsymbol{\delta}$ and which has to be recursively computed, it is not possible to
obtain a closed-form formula for $\widehat{\boldsymbol{\delta}}$. 
Thus $\widehat{\boldsymbol{\delta}}$ is computed by using the Newton-Raphson algorithm. 
More precisely, starting from an initial value for $\boldsymbol{\delta}$ denoted by
$\boldsymbol{\delta}^{(0)}$, the following recursion for $r\geq 1$ is used: 
\begin{equation}\label{eq:newton_raphson}
\boldsymbol{\delta}^{(r)}=\boldsymbol{\delta}^{(r-1)}-\frac{\partial^2 L}{\partial \boldsymbol{\delta}'\partial \boldsymbol{\delta}}(\boldsymbol{\delta}^{(r-1)})^{-1}\frac{\partial L}{\partial \boldsymbol{\delta}}(\boldsymbol{\delta}^{(r-1)}),
\end{equation}
where $\frac{\partial^2 L}{\partial \boldsymbol{\delta}'\partial \boldsymbol{\delta}}$ corresponds to the Hessian matrix of $L$
and is defined in (\ref{eq:def:hess}) given below.
Hence, it requires the computation of the first and second derivatives of $L$. 
We already explained how to compute the first derivatives of $L$. As for the second derivatives of $L$, it can be obtained as follows:

\begin{equation}\label{eq:def:hess}
\frac{\partial^2 L}{\partial \boldsymbol{\delta}'\partial \boldsymbol{\delta}}(\boldsymbol{\delta})
=\sum_{t=1}^n(Y_t-\exp(W_t(\boldsymbol{\delta}))\frac{\partial^2 W_t}{\partial \boldsymbol{\delta}'\partial\boldsymbol{\delta}}(\boldsymbol{\delta})
-\sum_{t=1}^n\exp(W_t(\boldsymbol{\delta}))\frac{\partial W_t}{\partial \boldsymbol{\delta}'}(\boldsymbol{\delta})\frac{\partial W_t}{\partial \boldsymbol{\delta}}(\boldsymbol{\delta}).
\end{equation}
The computations of the second derivatives of $W_t$ are detailed in Section \ref{subsub:second_derive}. 

However, in our sparse framework where many components of $\boldsymbol{\beta}^\star$ are null,
this procedure provides poor estimation results, see Section \ref{sec:sparse_estim} for numerical illustration.
This is the reason why we devised a novel estimation procedure described in the next section.

\subsection{Our estimation procedure}\label{sec:our_estim}


For selecting the most relevant components of $\boldsymbol{\beta}^\star$, we propose the following two-stage procedure:  Firstly, we
estimate $\boldsymbol{\gamma}^\star$ by using the Newton-Raphson algorithm described in Section \ref{sec:estim_gamma}
and secondly, we estimate  $\boldsymbol{\beta}^\star$
by using the regularized approach detailed in Section \ref{sec:variable}.

\subsubsection{Estimation of $\boldsymbol{\gamma}^\star$}\label{sec:estim_gamma}

To estimate $\boldsymbol{\gamma}^\star$, we propose using
\begin{equation*}
\widehat{\boldsymbol{\gamma}}=\textrm{Argmax}_{\boldsymbol{\gamma}} \; L({\boldsymbol{\beta}^{(0)}}',\boldsymbol{\gamma}'),
\end{equation*}
where $L$ is defined in (\ref{eq:likelihood}), $\boldsymbol{\beta}^{(0)}=(\beta_{0}^{(0)},\dots,\beta_{p}^{(0)})'$ is a given initial value for
$\boldsymbol{\beta}^\star$ and $\boldsymbol{\gamma}=(\gamma_1,\dots,\gamma_q)'$.
Similar to the approach proposed in Section \ref{sec:estim}, we use the Newton-Raphson algorithm
to obtain $\widehat{\boldsymbol{\gamma}}$ based on the following recursion for $r\geq 1$ starting from the initial value
$\boldsymbol{\gamma}^{(0)}=(\gamma_1^{(0)},\dots,\gamma_q^{(0)})'$:
\begin{equation}\label{eq:newton_raphson:gamma}
  \boldsymbol{\gamma}^{(r)}=\boldsymbol{\gamma}^{(r-1)}-\frac{\partial^2 L}{\partial \boldsymbol{\gamma}'\partial
    \boldsymbol{\gamma}}({\boldsymbol{\beta}^{(0)}}',{\boldsymbol{\gamma}^{(r-1)}}')^{-1}
  \frac{\partial L}{\partial \boldsymbol{\gamma}}({\boldsymbol{\beta}^{(0)}}',{\boldsymbol{\gamma}^{(r-1)}}'),
\end{equation}
where the first and second derivatives of $L$ are obtained using the same strategy as the one used
for deriving Equations (\ref{eq:def:grad}) and (\ref{eq:def:hess}) in Section \ref{sec:estim}.

\subsubsection{Variable selection: Estimation of $\boldsymbol{\beta}^\star$}\label{sec:variable}

To perform variable selection in the $\beta_i^\star$ of Model (\ref{eq:mut_Wt}) aimed to obtain a sparse estimator of $\beta_i^\star$, 
we shall use a methodology inspired by \cite{friedman:hastie:tibshirani:2010} 
for fitting generalized linear models with $\ell_1$ penalties. It consists in penalizing a quadratic approximation to the log-likelihood obtained by a Taylor expansion. Using $\boldsymbol{\beta}^{(0)}$ and $\widehat{\boldsymbol{\gamma}}$ defined in Section \ref{sec:estim_gamma}, 
the quadratic approximation is obtained as follows:
\begin{align*}
  \widetilde{L}(\boldsymbol{\beta})&:=L(\beta_0,\dots,\beta_p,\widehat{\gamma})\\
  &=\widetilde{L}(\boldsymbol{\beta}^{(0)})
+\frac{\partial L}{\partial \boldsymbol{\beta}}(\boldsymbol{\beta}^{(0)},\widehat{\boldsymbol{\gamma}})(\boldsymbol{\beta}-\boldsymbol{\beta}^{(0)})
+\frac12 (\boldsymbol{\beta}-\boldsymbol{\beta}^{(0)})'
\frac{\partial^2 L}{\partial \boldsymbol{\beta}\partial \boldsymbol{\beta}'}(\boldsymbol{\beta}^{(0)},\widehat{\boldsymbol{\gamma}})
(\boldsymbol{\beta}-\boldsymbol{\beta}^{(0)}),
\end{align*}
where
$$\frac{\partial L}{\partial \boldsymbol{\beta}}=\left(\frac{\partial L}{\partial \beta_0},\dots,\frac{\partial L}{\partial \beta_p}\right)
\textrm{ and }
\frac{\partial^2 L}{\partial \boldsymbol{\beta}\partial \boldsymbol{\beta}'}=\left(\frac{\partial^2 L}{\partial \beta_j \partial \beta_k}\right)_{0\leq j,k\leq p}.$$
Thus,
\begin{align}\label{eq:Ltilde}
\widetilde{L}(\boldsymbol{\beta})=\widetilde{L}(\boldsymbol{\beta}^{(0)})+\frac{\partial L}{\partial \boldsymbol{\beta}}(\boldsymbol{\beta}^{(0)},\widehat{\boldsymbol{\gamma}})
U(\boldsymbol{\nu}-\boldsymbol{\nu}^{(0)})-\frac12 (\boldsymbol{\nu}-\boldsymbol{\nu}^{(0)})' \Lambda (\boldsymbol{\nu}-\boldsymbol{\nu}^{(0)}),
\end{align}
where $U\Lambda U'$ is the singular value decomposition of the positive semidefinite symmetric matrix 
$-\frac{\partial^2 L}{\partial \boldsymbol{\beta}\partial \boldsymbol{\beta}'}(\boldsymbol{\beta}^{(0)},\widehat{\boldsymbol{\gamma}})$
and $\boldsymbol{\nu}-\boldsymbol{\nu}^{(0)}=U'(\boldsymbol{\beta}-\boldsymbol{\beta}^{(0)})$.

In order to obtain a sparse estimator of $\boldsymbol{\beta}^\star$, we propose using $\widehat{\boldsymbol{\beta}}(\lambda)$ defined by
\begin{equation}\label{eq:beta_hat}
\widehat{\boldsymbol{\beta}}(\lambda)=\textrm{Argmin}_{\boldsymbol{\beta}}\left\{-\widetilde{L}_Q(\boldsymbol{\beta})+\lambda \|\boldsymbol{\beta}\|_1\right\},
\end{equation}
for a positive $\lambda$, where $\|\boldsymbol{\beta}\|_1=\sum_{k=0}^p |\beta_k|$ and $\widetilde{L}_Q(\boldsymbol{\beta})$ denotes the quadratic approximation of the log-likelihood. 
This quadratic approximation is defined by
\begin{equation}\label{eq:LQtilde}
-\widetilde{L}_Q(\boldsymbol{\beta})=\frac12\|\mathcal{Y}-\mathcal{X}\boldsymbol{\beta}\|_2^2,
\end{equation}
with
\begin{equation}\label{eq:def_Y_X}
\mathcal{Y}=\Lambda^{1/2}U'\boldsymbol{\beta}^{(0)}
+\Lambda^{-1/2}U'\left(\frac{\partial L}{\partial \boldsymbol{\beta}}(\boldsymbol{\beta}^{(0)},\widehat{\boldsymbol{\gamma}})\right)' ,\;  \mathcal{X}=\Lambda^{1/2}U'
\end{equation}
and $\|\cdot\|_2$ denoting the $\ell_2$ norm in $\mathbb{R}^{p+1}$.
Computational details for obtaining the expression \eqref{eq:LQtilde} of $\widetilde{L}_Q(\boldsymbol{\beta})$ appearing in
Criterion (\ref{eq:beta_hat}) are provided in Section \ref{sub:var_sec}.

To obtain the final estimator $\widehat{\boldsymbol{\beta}}$ of $\boldsymbol{\beta}^\star$, we shall consider two different approaches:

\begin{itemize}
\item \textsf{Standard stability selection.} It consists in using the stability selection procedure devised by \cite{meinshausen:buhlmann:2010}
  which guarantees the robustness of the selected variables. This  approach can be described as follows.
The vector $\mathcal{Y}$ defined in (\ref{eq:def_Y_X}) is randomly split into several subsamples of size $(p+1)/2$, which corresponds to half of the length of $\mathcal{Y}$.
For each subsample $\mathcal{Y}^{(s)}$ and the corresponding design matrix $\mathcal{X}^{(s)}$,
the LASSO criterion (\ref{eq:beta_hat}) is applied with a given $\lambda$,
where $\mathcal{Y}$ and $\mathcal{X}$ are replaced by $\mathcal{Y}^{(s)}$ and $\mathcal{X}^{(s)}$, respectively.
For each subsampling, the indices $i$ of the non null $\widehat{\beta}_i$ are stored 
and, for a given threshold, we keep in the final
 set of selected variables only the ones appearing a number of times larger than this threshold. 
 Concerning the choice of $\lambda$, we shall consider the one obtained by cross-validation (Chapter 7 of \cite{hastie2009elements})
 and the smallest element of the grid of $\lambda$ provided by the R \texttt{glmnet} package.
\item \textsf{Fast stability selection.}
  It consists in applying the LASSO criterion (\ref{eq:beta_hat}) for several values of $\lambda$.
  For each $\lambda$, the indices $i$ of the non null $\widehat{\beta}_i(\lambda)$ are stored and, for a given threshold, we keep in the final
 set of selected variables only the ones appearing a number of times larger than this threshold.
\end{itemize}
These approaches will be further investigated in Section \ref{sec:num}.

\subsection{\textcolor{black}{Practical implementation}}\label{sec:practical}

In practice, the previous approach can be summarized as follows. 

\begin{itemize}
\item\textsf{Initialization.} We take for $\boldsymbol{\beta}^{(0)}$ 
the estimator of $\boldsymbol{\beta}^\star$ obtained by fitting a GLM to the observations
$Y_1,\dots,Y_n$ thus ignoring the ARMA part of the model in the case where $n>p$.
If $p$ is larger than $n$, then a regularized criterion for GLM models
can be used, see for instance \cite{friedman:hastie:tibshirani:2010}.
For $\boldsymbol{\gamma}^{(0)}$, we take the null vector.
\item\textsf{Newton-Raphson algorithm.} We use the recursion defined in (\ref{eq:newton_raphson:gamma}) with
the initialization $(\boldsymbol{\beta}^{(0)},\boldsymbol{\gamma}^{(0)})$ obtained in the previous step and
we stop at the iteration $R$ such that $\|\boldsymbol{\gamma}^{(R)}-\boldsymbol{\gamma}^{(R-1)}\|_\infty<10^{-6}$.
\item\textsf{Variable selection.} To obtain a sparse estimator of $\boldsymbol{\beta}^\star$, we use the criterion (\ref{eq:beta_hat})
where $\boldsymbol{\beta}^{(0)}$ and $\widehat{\boldsymbol{\gamma}}$ appearing in (\ref{eq:def_Y_X}) are replaced by $\boldsymbol{\beta}^{(0)}$ and $\boldsymbol{\gamma}^{(R)}$ 
obtained in the previous steps. We thus get $\widehat{\boldsymbol{\beta}}$ by using one of the three approaches described at the end of Section
\ref{sec:variable}.
\end{itemize}

This procedure can be improved by iterating the \textsf{Newton-Raphson algorithm} and \textsf{Variable selection} steps.
More precisely, let us denote by $\boldsymbol{\beta}_{1}^{(0)}$, $\gamma_{1}^{(R_1)}$ and $\widehat{\boldsymbol{\beta}}_1$ 
the values of $\boldsymbol{\beta}^{(0)}$, $\gamma^{(R)}$ and $\widehat{\boldsymbol{\beta}}$ obtained
in the three steps described above at the first iteration.
At the second iteration, $(\boldsymbol{\beta}^{(0)},\boldsymbol{\gamma}^{(0)})$ appearing in the \textsf{Newton-Raphson algorithm} step
is replaced by $(\widehat{\boldsymbol{\beta}}_1,\gamma_{1}^{(R_1)})$. At the end of this second iteration, $\widehat{\boldsymbol{\beta}}_2$ and $\gamma_{2}^{(R_2)}$
denote the obtained values of $\widehat{\boldsymbol{\beta}}$ and $\gamma^{(R)}$, respectively.
This approach is iterated until the stabilization of $\gamma_{k}^{(R_k)}$.

\subsection{Consistency results}\label{sec:consistency}

In this section, we shall establish the consistency of the parameter $\gamma_1^\star$ in the case where $q=1$
from $Y_1,\dots,Y_n$ defined in (\ref{eq:Yt}) and (\ref{eq:Zt})
where (\ref{eq:mut_Wt}) is replaced by
\begin{equation}\label{eq:mut_simple}
\mu_t^\star=\exp(W_t^\star) \textrm{ with } W_t^\star=\beta_0^\star+Z_t^\star.
\end{equation}
We limit ourselves to this framework since in the more general one 
the consistency is much more tricky to handle and is beyond the scope of this paper.
Note that some theoretical results have already been obtained in this framework (no covariates and $q=1$)
by \cite{davis:dunsmuir:streett:2003} and \cite{davis:dunsmuir:street:2005}. However, here, we provide, on the one hand, a more detailed version
of the proof of these results and on the other hand, a proof of the consistency of $\gamma_1^\star$ based on a stochastic equicontinuity result.

\begin{theo}\label{theo:MA1}
Assume that $Y_1,\dots,Y_n$ satisfy the model defined by (\ref{eq:Yt}), (\ref{eq:mut_simple}) and (\ref{eq:Zt}) with $q=1$ and $\gamma_1^\star\in\Gamma$ where $\Gamma$ is a compact set
of $\mathbb{R}$ which does not contain 0. Assume also that $(W_t^\star)$ 
starts with its stationary invariant distribution. Let $\widehat{\gamma}_1$ be defined by:
$$
\widehat{\gamma}_1=\textrm{Argmax}_{\gamma_1\in\Gamma}\; L(\beta_0^\star,\gamma_1),
$$
where 
\begin{equation}\label{eq:L:beta_0}
L(\beta_0^\star,\gamma_1)=\sum_{t=1}^n\left(Y_t W_t(\beta_0^\star,\gamma_1)-\exp(W_t(\beta_0^\star,\gamma_1)\right),
\end{equation}
with 
\begin{equation}\label{eq:W_Z}
W_t(\beta_0^\star,\gamma_1)=\beta_0^\star+Z_t(\gamma_1)=\beta_0^\star+\gamma_1 E_{t-1}(\gamma_1), 
\end{equation}

$$
E_{t-1}(\gamma_1)=Y_{t-1}\exp(-W_{t-1}(\beta_0^\star,\gamma_1))-1, \textrm{ if } t>1 \textrm{ and } E_{t-1}(\gamma_1)=0, \textrm{ if } t\leq 1.
$$

Then $\widehat{\gamma}_1\stackrel{p}{\longrightarrow}\gamma_1^\star$, as $n$ tends to infinity, where $\stackrel{p}{\longrightarrow}$ denotes the convergence in probability.
\end{theo}

The proof of Theorem \ref{theo:MA1} is based on the following propositions which are proved in Section \ref{sec:proofs}. These propositions are the classical arguments for establishing consistency
results of maximum likelihood estimators. Note that we shall explain in the proof of Proposition \ref{prop1}
why a stationary invariant distribution for $(W_t^\star)$ does exist. The main tools used for proving Propositions \ref{prop1} and \ref{prop3} are the Markov property and the ergodicity of $(W_t^\star)$.

\begin{prop}\label{prop1}
For all fixed $\gamma_1$, under the assumptions of Theorem \ref{theo:MA1}, 
\begin{equation}\label{eq:conv}
\frac1n L(\beta_0^\star,\gamma_1)\stackrel{p}{\longrightarrow} 
\mathcal{L}(\gamma_1):=\mathbb{E}\left[Y_3 W_3(\beta_0^\star,\gamma_1)-\exp(W_3(\beta_0^\star,\gamma_1)\right], \textrm{ as $n$ tends to infinity.}
\end{equation}
\end{prop}

\begin{prop}\label{prop2}
The function $\mathcal{L}$ defined in (\ref{eq:conv}) has a unique maximum at the true parameter $\gamma_1=\gamma_1^\star$.
\end{prop}

\begin{prop}\label{prop3}
Under the assumptions of Theorem \ref{theo:MA1}
$$\sup_{\gamma_1\in\Gamma}\left|\frac{L(\beta_0^\star,\gamma_1)}{n}-\mathcal{L}(\gamma_1)\right|\stackrel{p}{\longrightarrow}0, \textrm{ as $n$ tends to infinity,}$$
where $\mathcal{L}(\gamma_1)$ is defined in (\ref{eq:conv}).
\end{prop}



\section{Numerical experiments}\label{sec:num}

The goal of this section is to investigate the performance of our method both from a statistical and a numerical points of view, using synthetic data generated by the model defined by (\ref{eq:Yt}), (\ref{eq:mut_Wt}) and (\ref{eq:Zt}).

\subsection{Statistical performance}

\subsubsection{Estimation of the parameters when $p=0$}

In this section, we investigate the statistical performance of our methodology in the model
defined by (\ref{eq:Yt}), (\ref{eq:mut_Wt}) and (\ref{eq:Zt}) for $n$ in $\{50,100,250,500,1000\}$
in the case where $p=0$, namely when there are no covariates and for $q$ in
$\{1,2,3\}$. The performance of our approach for estimating $\beta_0^\star$ and the $\gamma_k^\star$ are displayed in Figures \ref{fig:estim_beta}, \ref{fig:estim:gam1} and \ref{fig:estim:gam2_3}. We can see from these figures that
the accuracy of the parameter estimations is improved when $n$ increases, which corroborates the consistency of $\gamma_1^\star$ 
given in Theorem \ref{theo:MA1} in the case $q=1$.

\begin{figure}[!htbp]
  \includegraphics[scale=0.28]{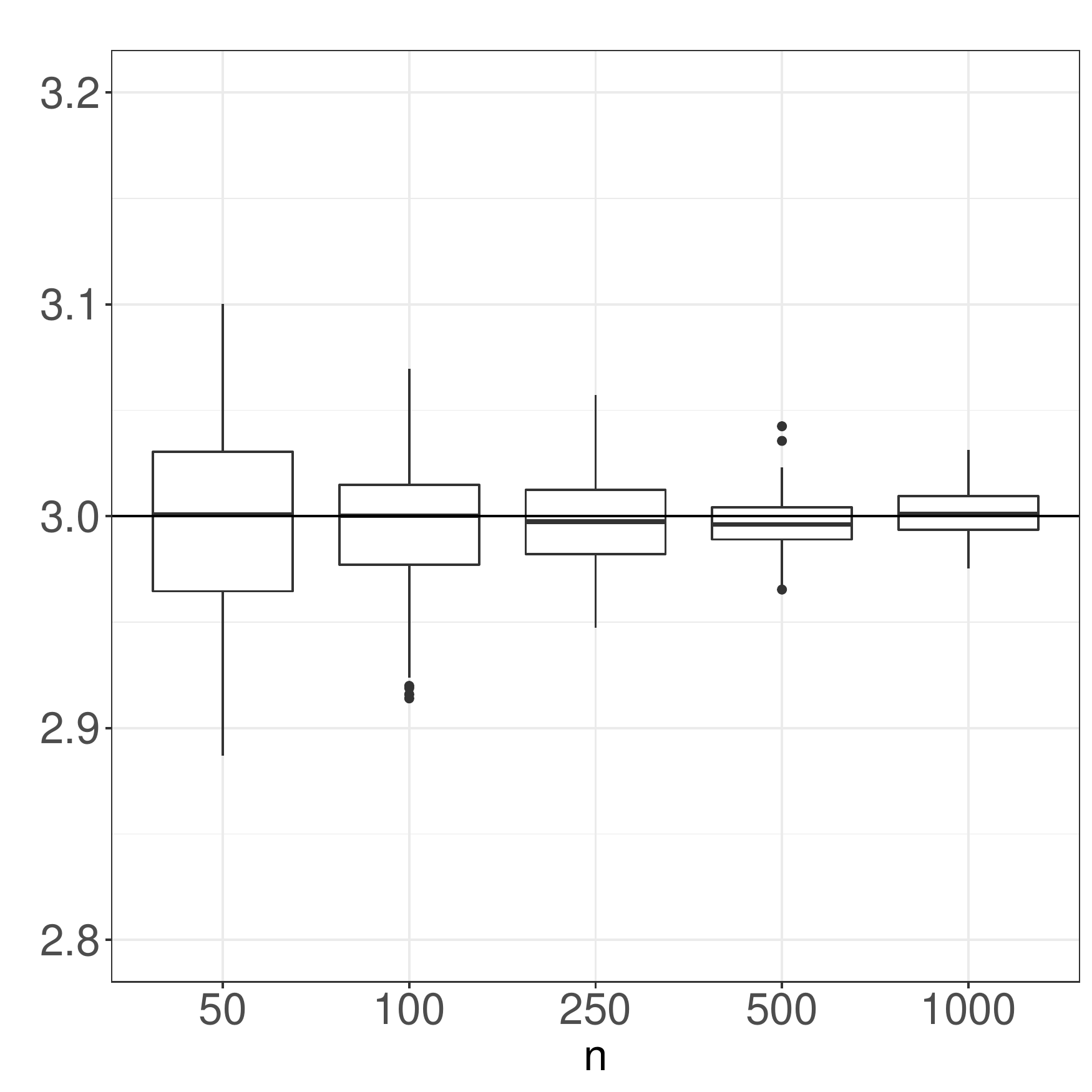}
  \includegraphics[scale=0.28]{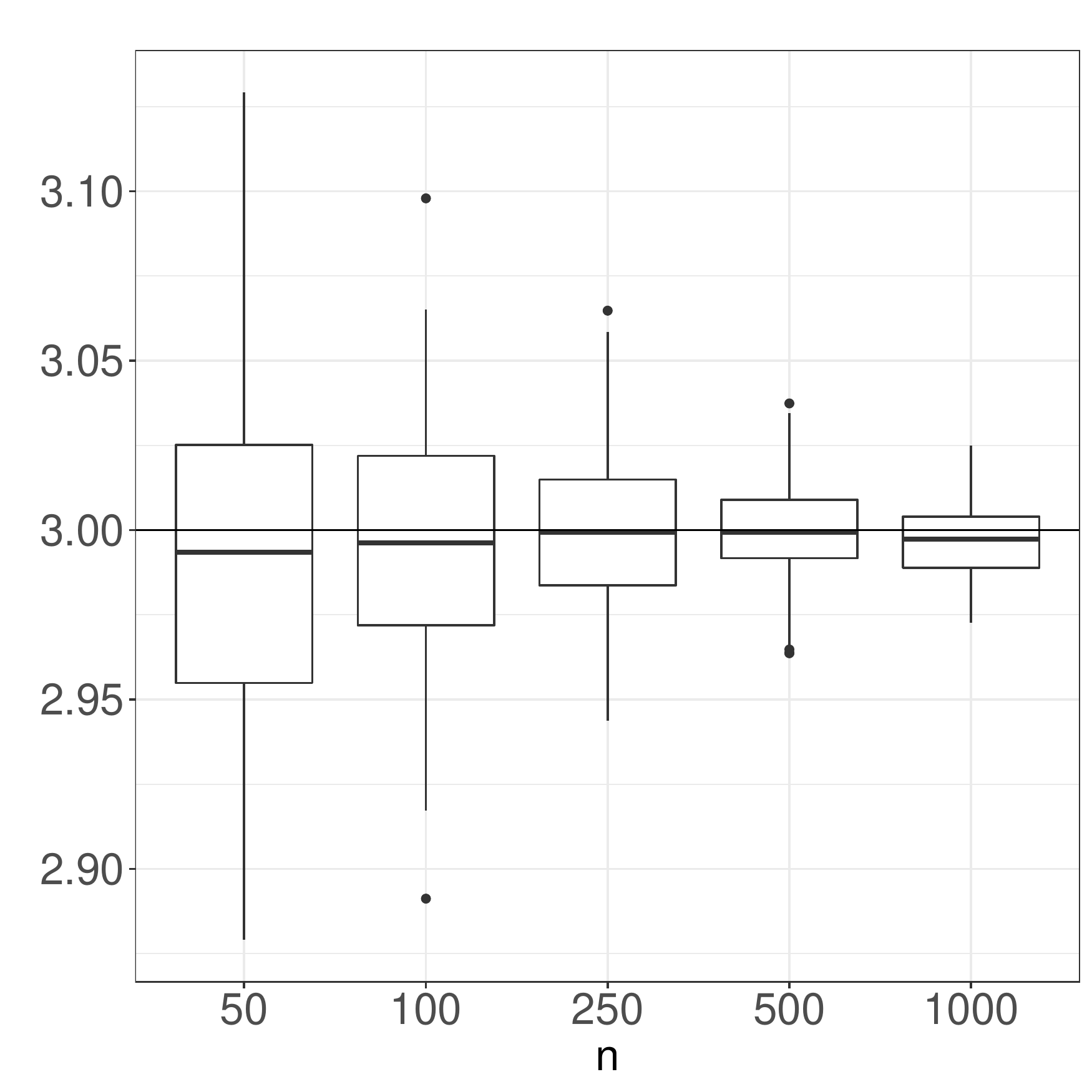}
  \includegraphics[scale=0.28]{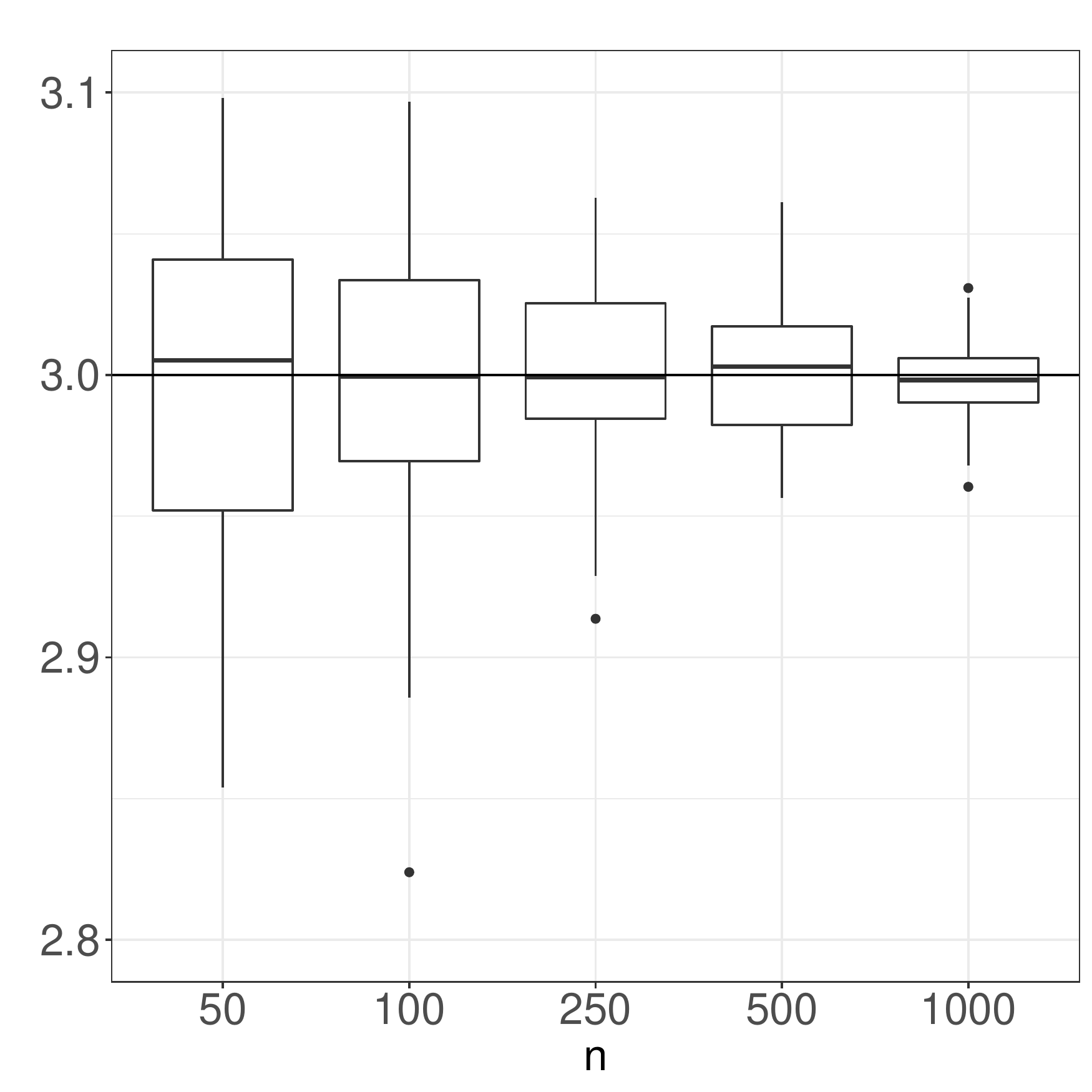}
  \caption{Boxplots for the estimations of $\beta_0^\star=3$ in Model (\ref{eq:mut_Wt}) with no regressor and $q=1$ (left), $q=2$ (middle) and $q=3$ (right). 
The horizontal lines correspond to the value of $\beta_0^\star$.\label{fig:estim_beta}}
\end{figure}

\begin{figure}[!htbp]
  \includegraphics[scale=0.28]{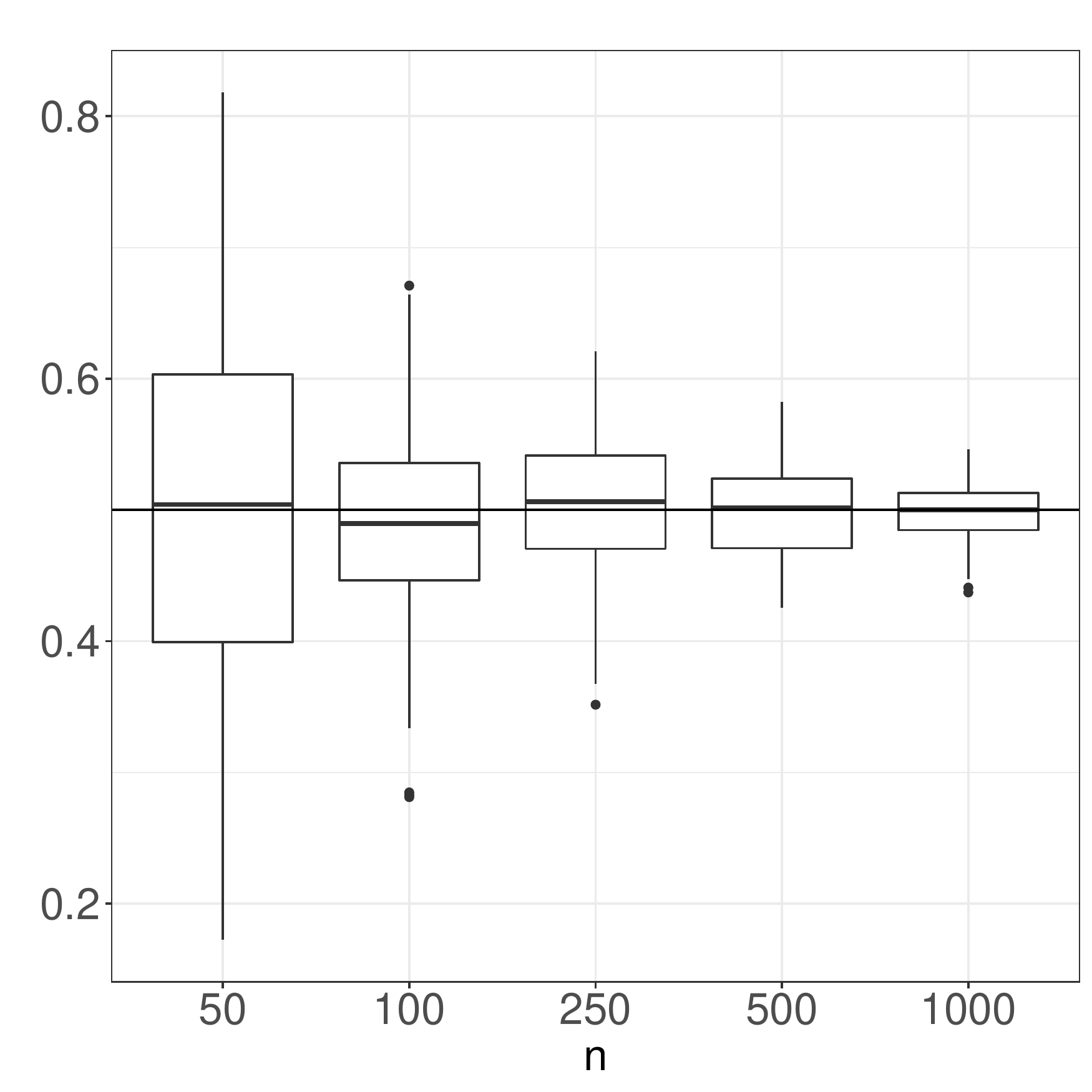}
  \includegraphics[scale=0.28]{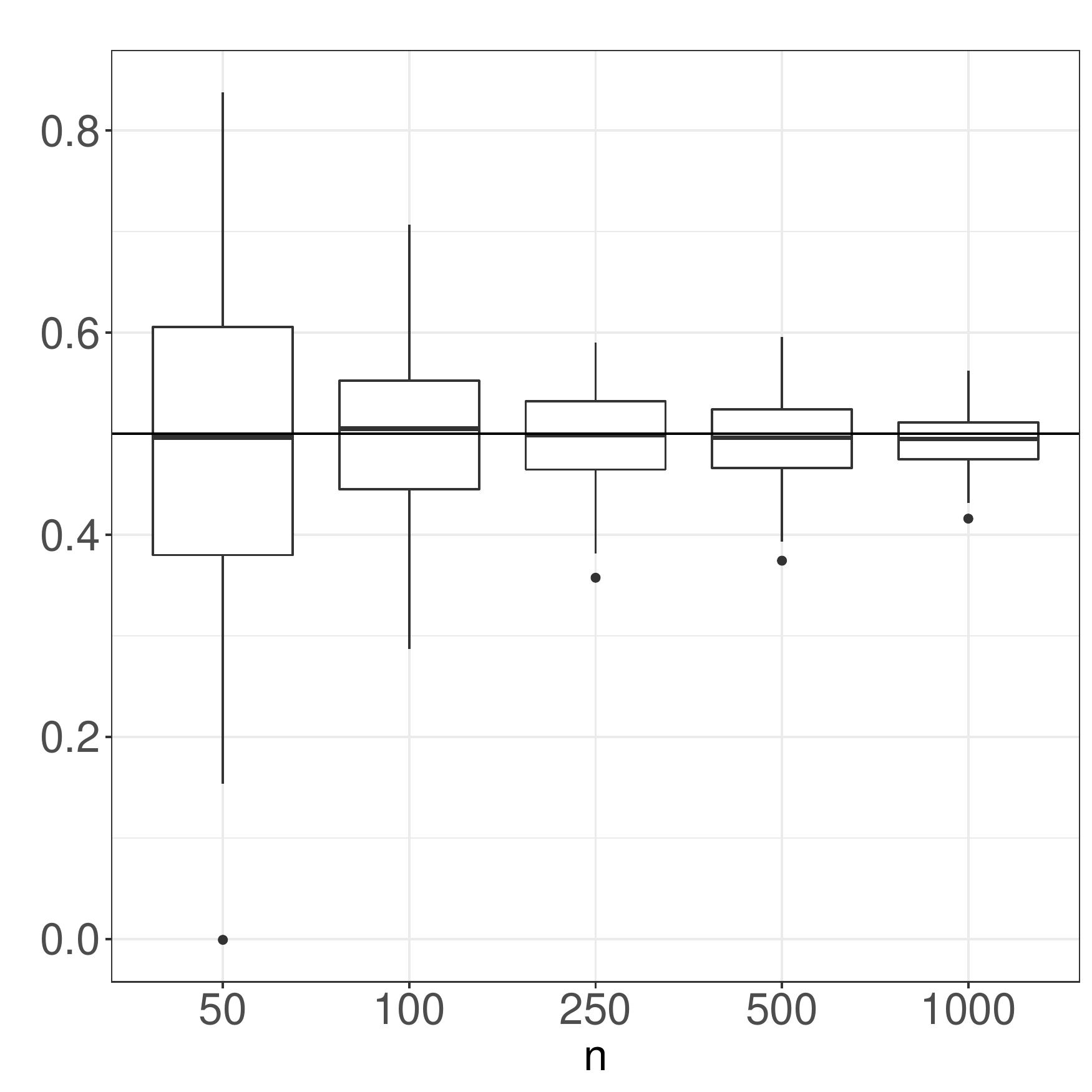}
  \includegraphics[scale=0.28]{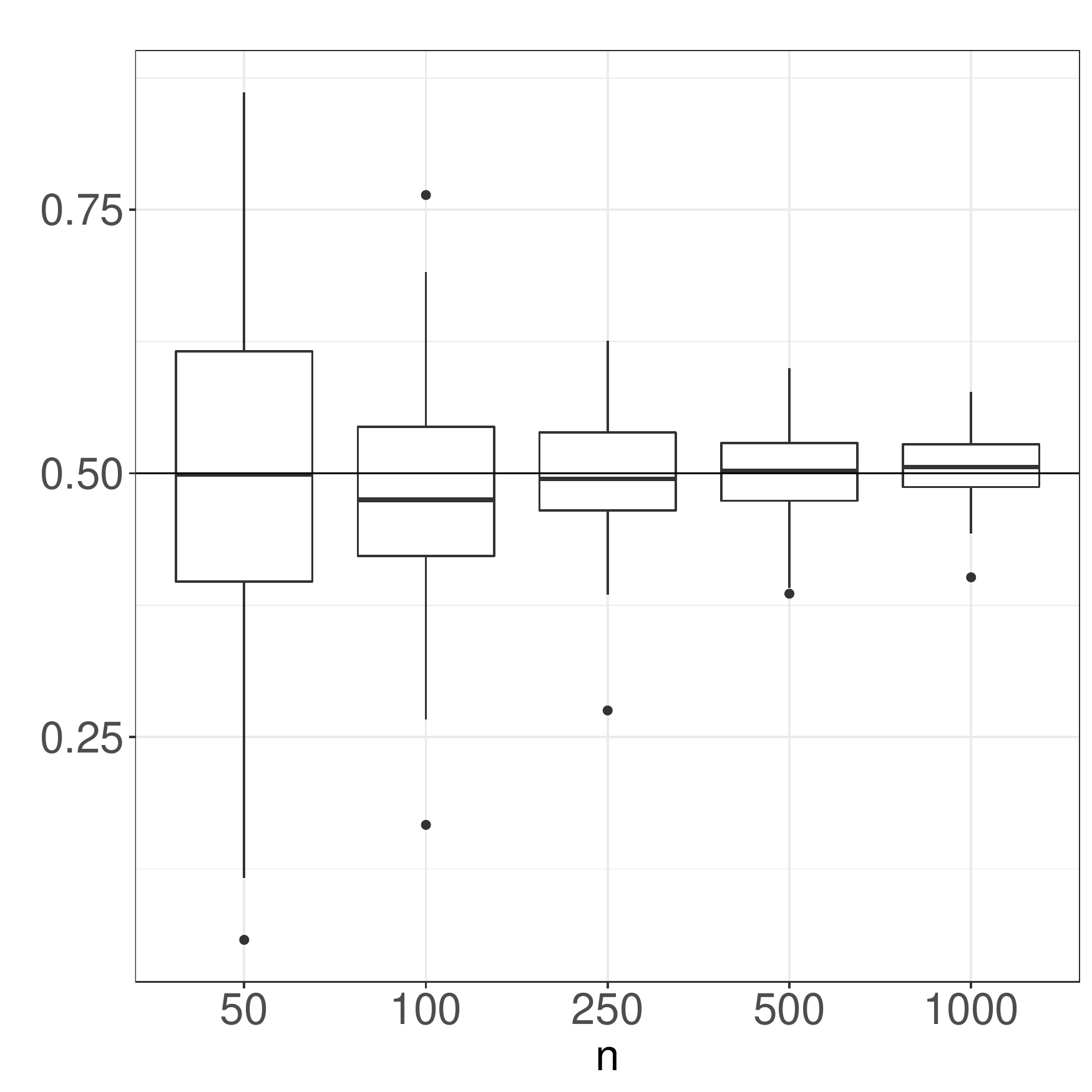}
  \caption{Boxplots for the estimations of $\gamma_1^\star=0.5$ in Model (\ref{eq:mut_Wt}) with no regressor and $q=1$ (left), $q=2$ (middle) and $q=3$ (right).
  \textcolor{black}{The horizontal lines correspond to the value of $\gamma_1^\star$.} \label{fig:estim:gam1}}
\end{figure}

  \begin{figure}[!htbp]
  \includegraphics[scale=0.28]{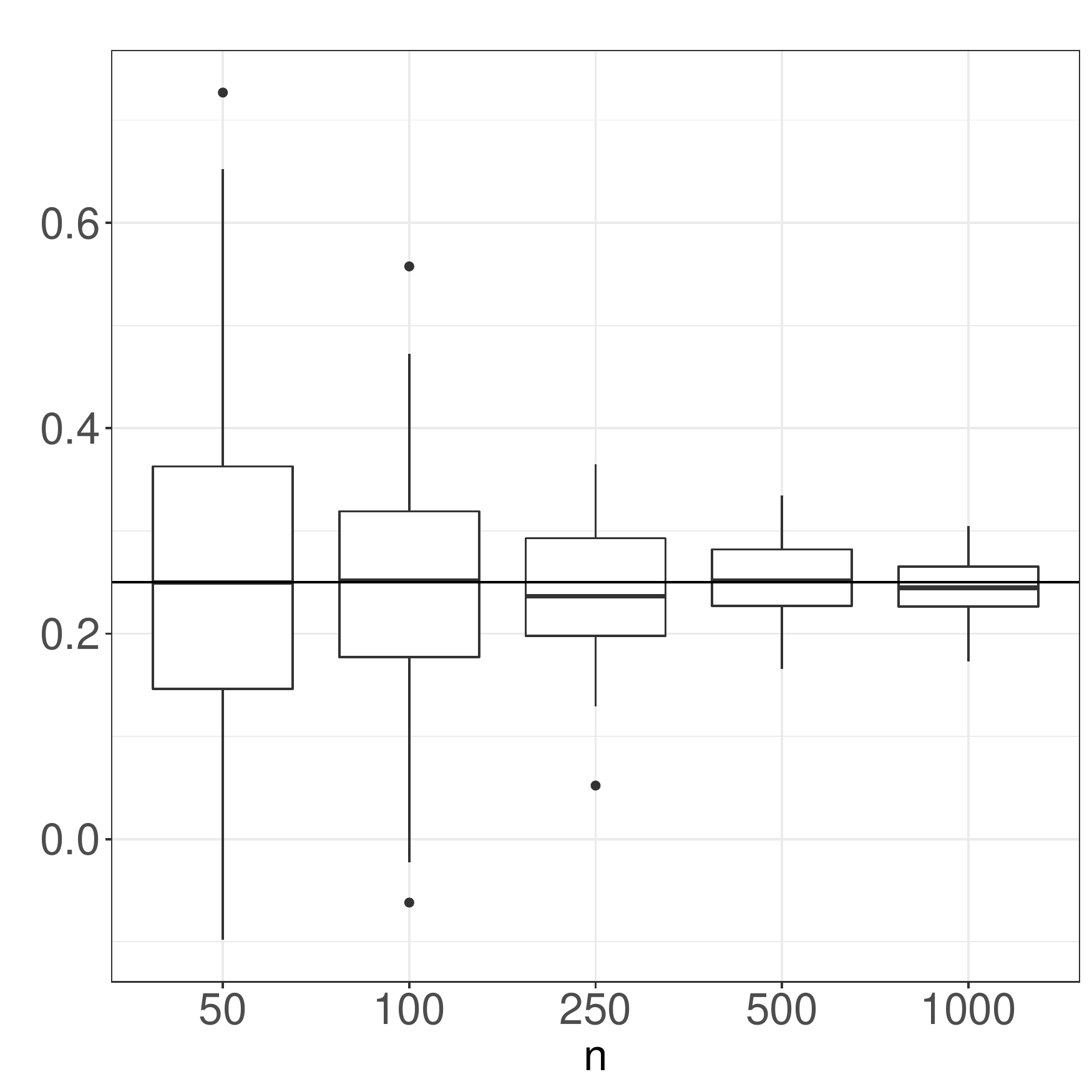}
  \includegraphics[scale=0.28]{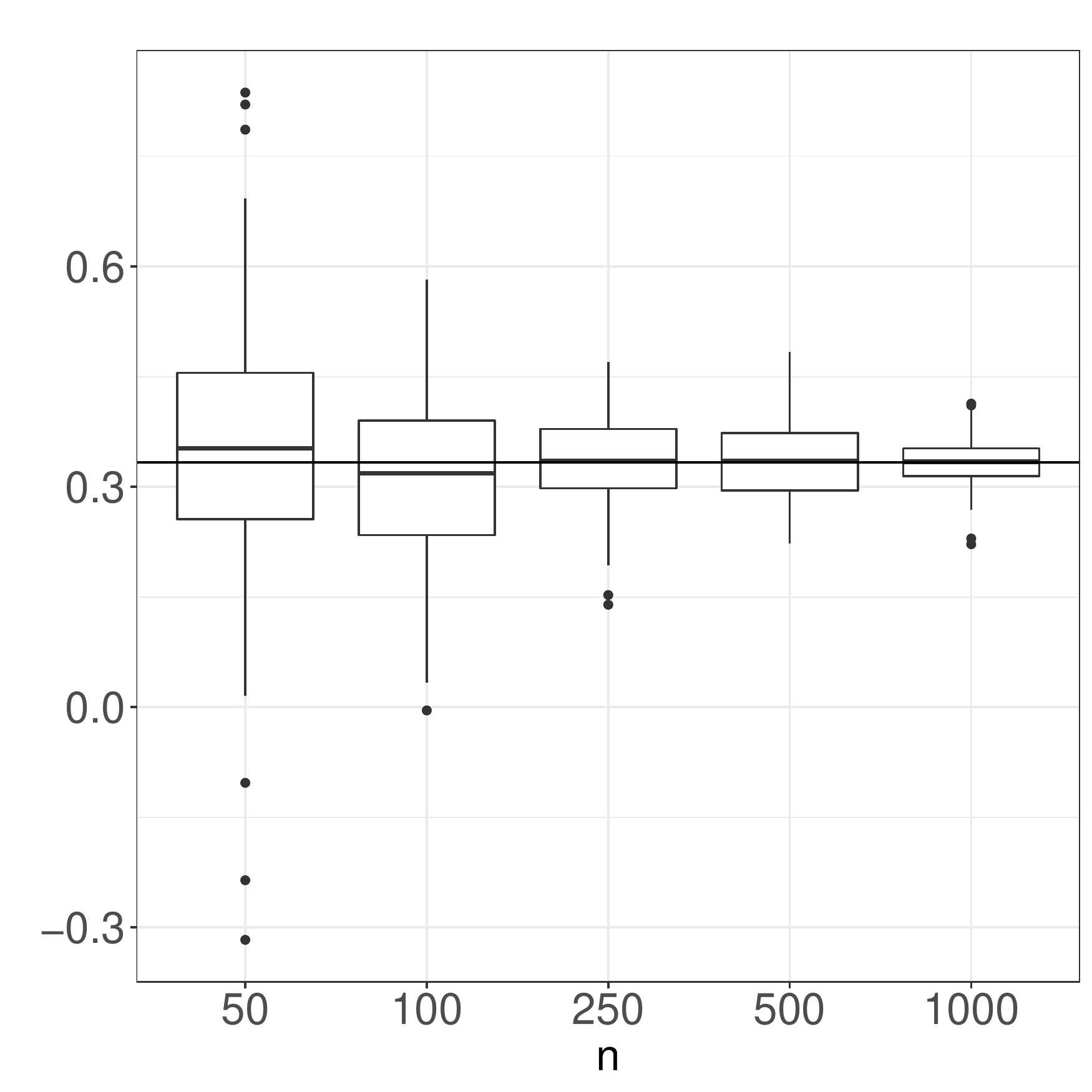}
  \includegraphics[scale=0.28]{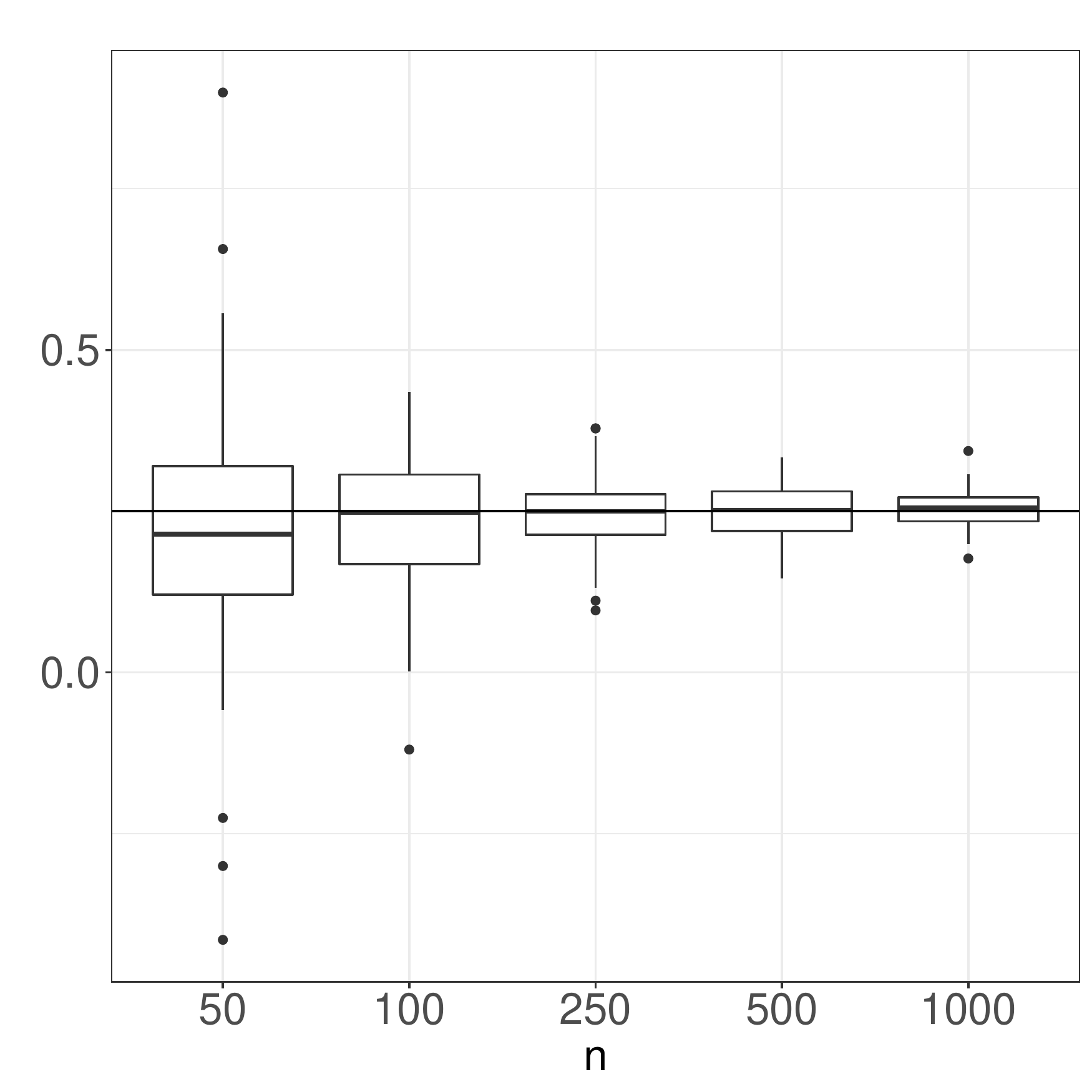}
  \caption{Boxplots for the estimations of $\gamma_2^\star=1/4$ in Model (\ref{eq:mut_Wt}) with no regressor and $q=2$ (left), $\gamma_2^\star=1/3$ in Model (\ref{eq:mut_Wt}) 
with no regressor and $q=3$ (middle) and of $\gamma_3^\star=1/4$
    in Model (\ref{eq:mut_Wt}) with no regressor and $q=3$ (right).
\textcolor{black}{The horizontal lines correspond to the true values of the parameters.}
  \label{fig:estim:gam2_3}}
\end{figure}

Moreover, it has to be noticed that in this particular context where there are no covariates ($p=0$), the
performance of our approach in terms of parameters estimation is similar to the one of the package \texttt{glarma}
described in \cite{glarma:package}.



\subsubsection{Estimation of the parameters when $p\geq 1$ and $\boldsymbol{\beta}^\star$ is sparse}\label{sec:sparse_estim}

In this section, we assess the performance of our methodology in terms of support recovery, namely the identification of the
non null coefficients of $\boldsymbol{\beta}^\star$, and of the estimation of $\boldsymbol{\gamma}^\star$. We shall consider $Y_1,\dots,Y_n$ satisfying the model
defined by (\ref{eq:Yt}), (\ref{eq:mut_Wt}) and (\ref{eq:Zt}) with covariates chosen in a Fourier basis, for $n=1000$ in the first two paragraphs,
$q\in\{1,2,3\}$, $p=100$ and two sparsity levels (5\% or 10\% of non null coefficients in $\boldsymbol{\beta}^\star$). More precisely, when the sparsity level is 5\%(resp. 10\%)  all the $\beta_i^\star$ are assumed to be equal to zero except for five (resp. ten) of them for which the values 
are given in the caption of Figure \ref{fig:TPR:FPR:1} (resp.\,in the caption of Figure \ref{fig:TPR:FPR:1:10} given in the Appendix).
Other values of $n$ (150, 200, 500, 1000) will be considered in the third paragraph to evaluate the impact of $n$ on the performance of our approach.

\textbf{Estimation of the support of $\boldsymbol{\beta}^\star$}

In this paragraph, we focus on the performance of our approach for retrieving the support of $\boldsymbol{\beta}^\star$ by computing the
True Positive Rates (TPR) and False Positive Rates (FPR). We shall consider
the two methods that are proposed in Section \ref{sec:variable}: standard stability selection (\verb|ss_cv| and \verb|ss_min|) and fast stability selection (\verb|fast_ss|). For comparison purpose, we shall
also consider the standard Lasso approach proposed by \cite{friedman:hastie:tibshirani:2010} in GLM where the parameter $\lambda$ is either chosen
thanks to the standard cross-validation (\verb|lasso_cv|) or by taking the optimal $\lambda$ which maximizes the difference between
the TPR and FPR (\verb|lasso_best|).

Figures \ref{fig:TPR:FPR:1}, \ref{fig:TPR:FPR:2} and \ref{fig:TPR:FPR:3} display the TPR and FPR of the previously mentioned approaches
with respect to the threshold defined at the end of Section \ref{sec:variable}
when $n=1000$, the sparsity level is equal to 5\% and $q=1$, 2 and 3, respectively.
We can see from these figures that when the threshold is well tuned, our approaches outperform the classical Lasso even when the parameter
$\lambda$ is chosen in an optimal way. More precisely, the thresholds 0.4, 0.7 and 0.8 achieve a satisfactory trade-off between the TPR and the FPR
for \verb|fast_ss|, \verb|ss_cv| and \verb|ss_min|, respectively.
The conclusions are similar in the case where the sparsity level is equal to 10\%, the corresponding figures (\ref{fig:TPR:FPR:1:10}, \ref{fig:TPR:FPR:2:10}
and \ref{fig:TPR:FPR:3:10}) are given in the Appendix.
We can observe from these figures that the performance of \verb|fast_ss| are slightly better than \verb|ss_cv| and \verb|ss_min| when the
sparsity level is equal to 5\% but it is the reverse when the sparsity level is equal to 10\%.

\begin{figure}[!htbp]
  \includegraphics[scale=0.28]{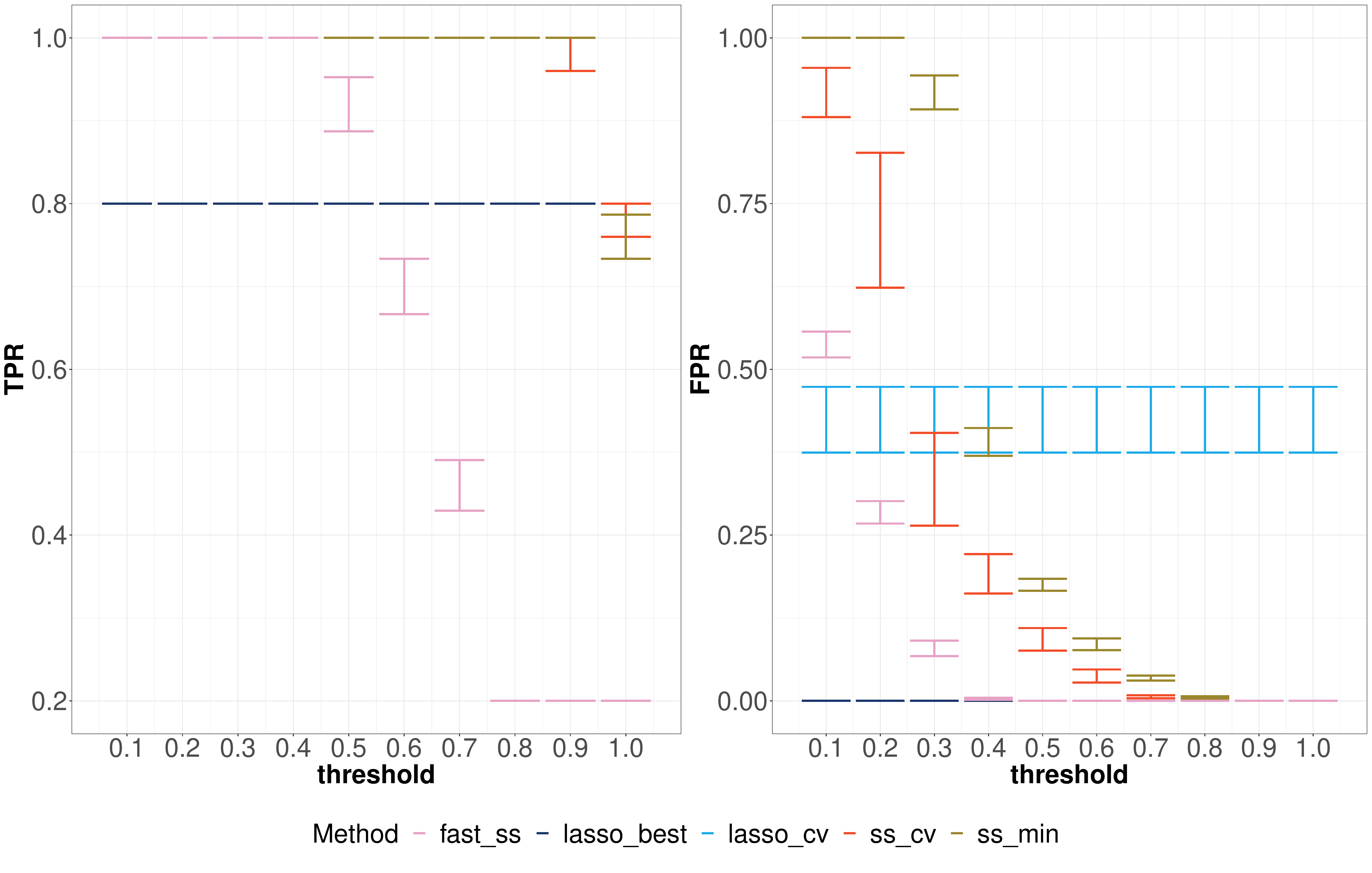}
  \caption{Error bars of the TPR and FPR associated to the support recovery of $\boldsymbol{\beta}^\star$ for five methods with respect to the thresholds when $n=1000$, $q=1$, $p=100$ and a 5\% sparsity level. All the $\beta_i^\star=0$ except for five of them: $\beta_1^\star=1.73$, $\beta_3^\star=0.38$, $\beta_{17}^\star=0.29$, $\beta_{33}^\star=-0.64$ and $\beta_{44}^\star=-0.13$.
 \label{fig:TPR:FPR:1}}
\end{figure}

\begin{figure}[!htbp]
  \includegraphics[scale=0.28]{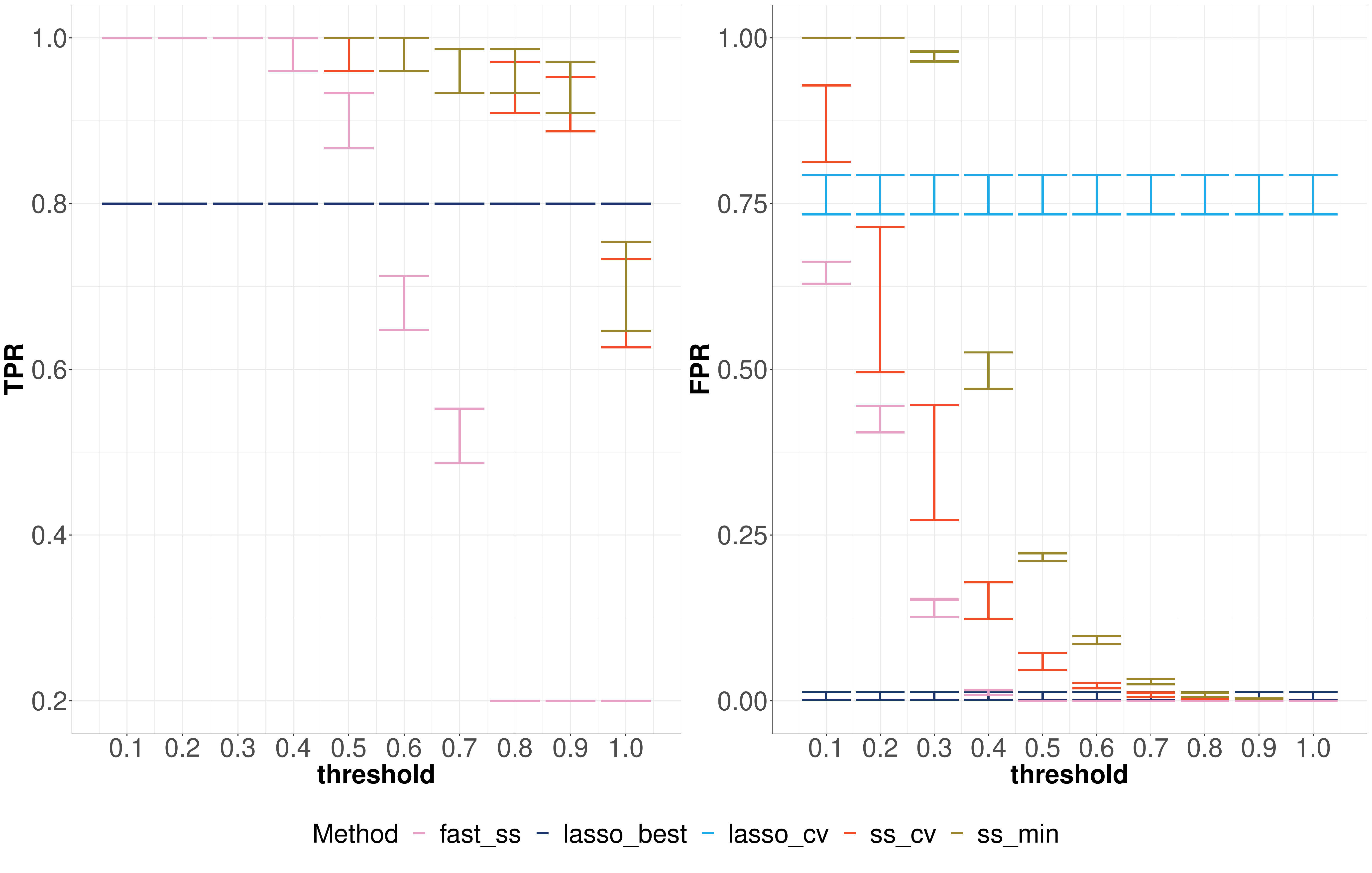}
  \caption{Error bars of the TPR and FPR associated to the support recovery of $\boldsymbol{\beta}^\star$ for five methods with respect to the thresholds when $n=1000$, $q=2$, $p=100$ and a 5\% sparsity level. All the $\beta_i^\star=0$ except for five of them: $\beta_1^\star=1.73$, $\beta_3^\star=0.38$, $\beta_{17}^\star=0.29$, $\beta_{33}^\star=-0.64$ and $\beta_{44}^\star=-0.13$.\label{fig:TPR:FPR:2}}
\end{figure}

\begin{figure}[!htbp]
  \includegraphics[scale=0.28]{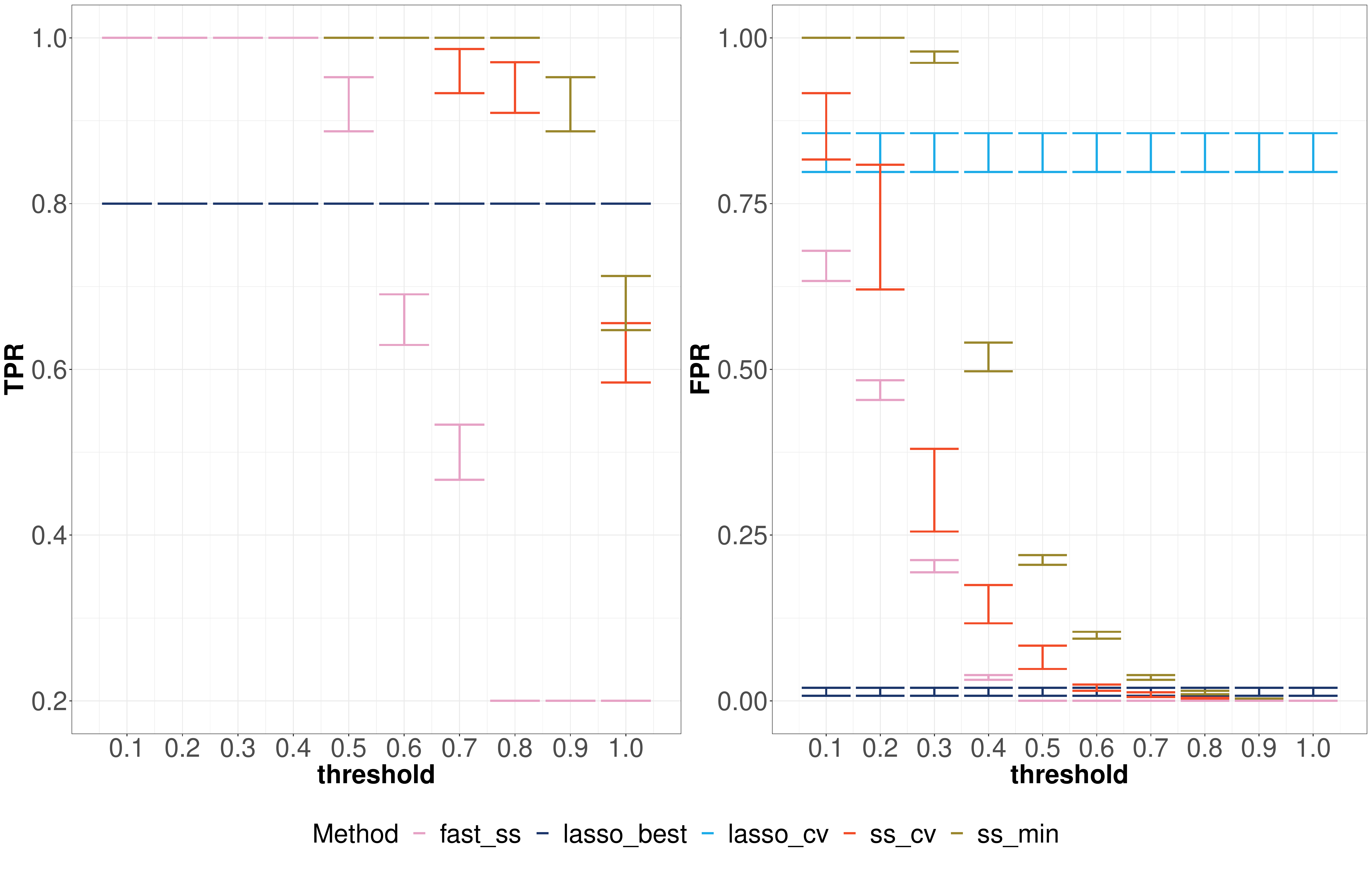}
 \caption{Error bars of the TPR and FPR associated to the support recovery of $\boldsymbol{\beta}^\star$ for five methods with respect to the thresholds when $n=1000$, $q=3$, $p=100$ and a 5\% sparsity level. All the $\beta_i^\star=0$ except for five of them: $\beta_1^\star=1.73$, $\beta_3^\star=0.38$, $\beta_{17}^\star=0.29$, $\beta_{33}^\star=-0.64$ and $\beta_{44}^\star=-0.13$.\label{fig:TPR:FPR:3}}
\end{figure}

We also compare our approach with the method implemented in the \texttt{glarma} package of \cite{glarma:package} in the case where $q=1$ and
when the sparsity level is equal to 5\%.
Since this method is not devised for performing variable selection, we consider that a given component of $\boldsymbol{\beta}^\star$ is estimated by 0
if its estimation obtained by the \texttt{glarma} package is smaller than a given threshold.
The results are displayed in Figure \ref{fig:TPR:FPR:glarma} for different thresholds ranging from $10^{-9}$ to 0.1.
We can see from this figure that for the best choice of the threshold the results of the variable selection provided by the \texttt{glarma} package
underperform our method.

\begin{figure}[!htbp]
  \includegraphics[scale=0.28]{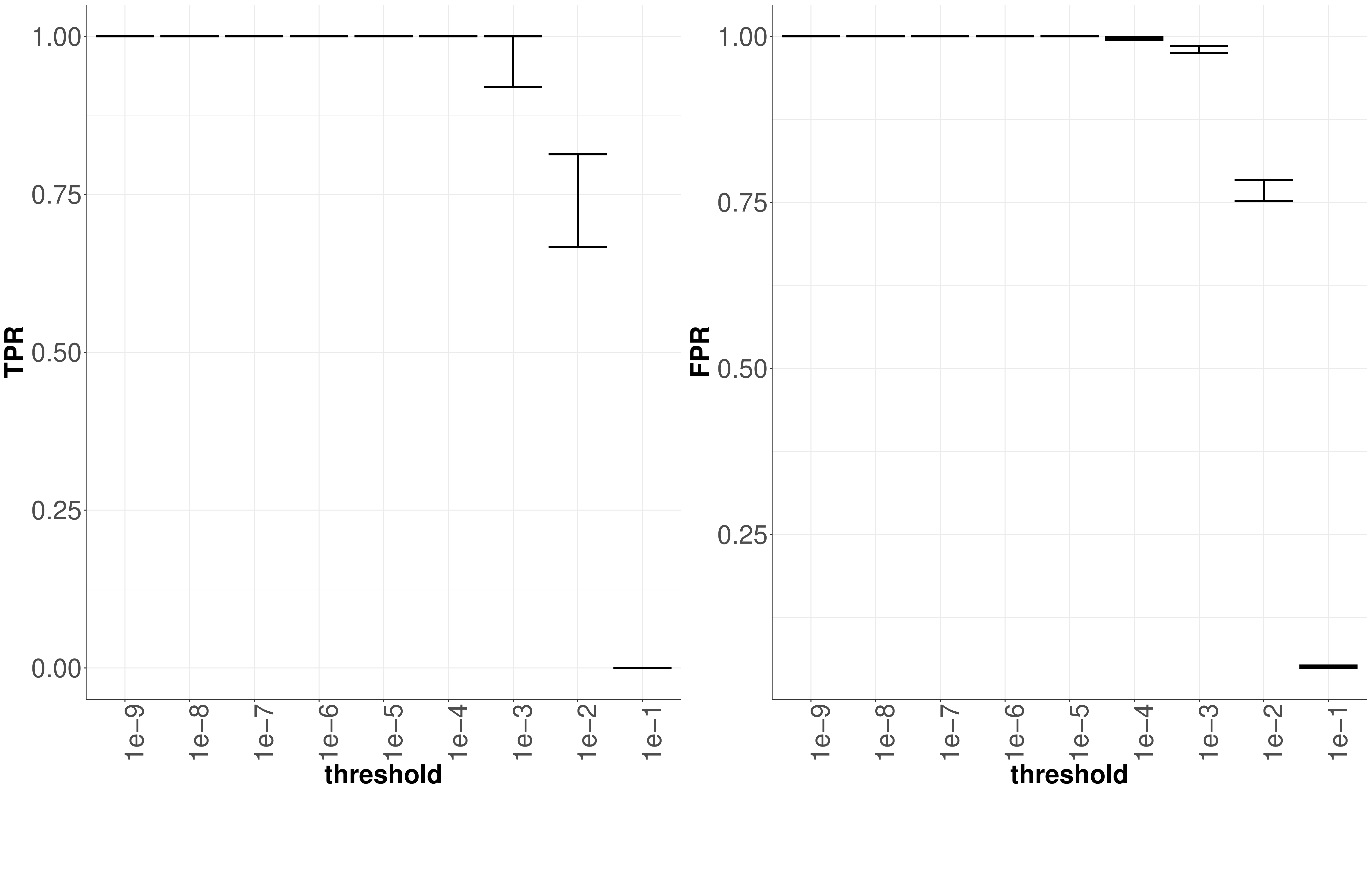}
  \caption{Error bars of the TPR and FPR associated to the support recovery of $\boldsymbol{\beta}^\star$ obtained with the \texttt{glarma} package
    for different thresholds when $n=1000$, $q=1$, $p=100$ and a 5\% sparsity level. All the $\beta_i^\star=0$ except for five of them: $\beta_1^\star=1.73$, $\beta_3^\star=0.38$, $\beta_{17}^\star=0.29$, $\beta_{33}^\star=-0.64$ and $\beta_{44}^\star=-0.13$.
 \label{fig:TPR:FPR:glarma}}
\end{figure}

\textbf{Estimation of $\boldsymbol{\gamma}^\star$}

Figures \ref{fig:gamma:5:cv}, \ref{fig:gamma:5:fast} and \ref{fig:gamma:5:min} display the boxplots for the estimations of $\boldsymbol{\gamma}^\star$ in Model
(\ref{eq:mut_Wt}) with a 5\% sparsity level and $q=1,2,3$ obtained by \verb|ss_cv|, \verb|fast_ss| and \verb|ss_min|, respectively.
The threshold chosen for each of these methods is the one achieving a satisfactory trade-off between the TPR and the FPR, namely 0.7, 0.4 and 0.8.
We can see from these figures that all these approaches provide accurate estimations of $\boldsymbol{\gamma}^\star$ from the second iteration.
The conclusions are similar in the case where the sparsity level is equal to 10\%, the corresponding figures \ref{fig:gamma:10:cv},
\ref{fig:gamma:10:fast} and \ref{fig:gamma:10:min} are given in the Appendix.

\begin{figure}[!htbp]
  \begin{center}
\begin{tabular}{ccc}
  \includegraphics[width=0.32\textwidth, height=4.5cm]{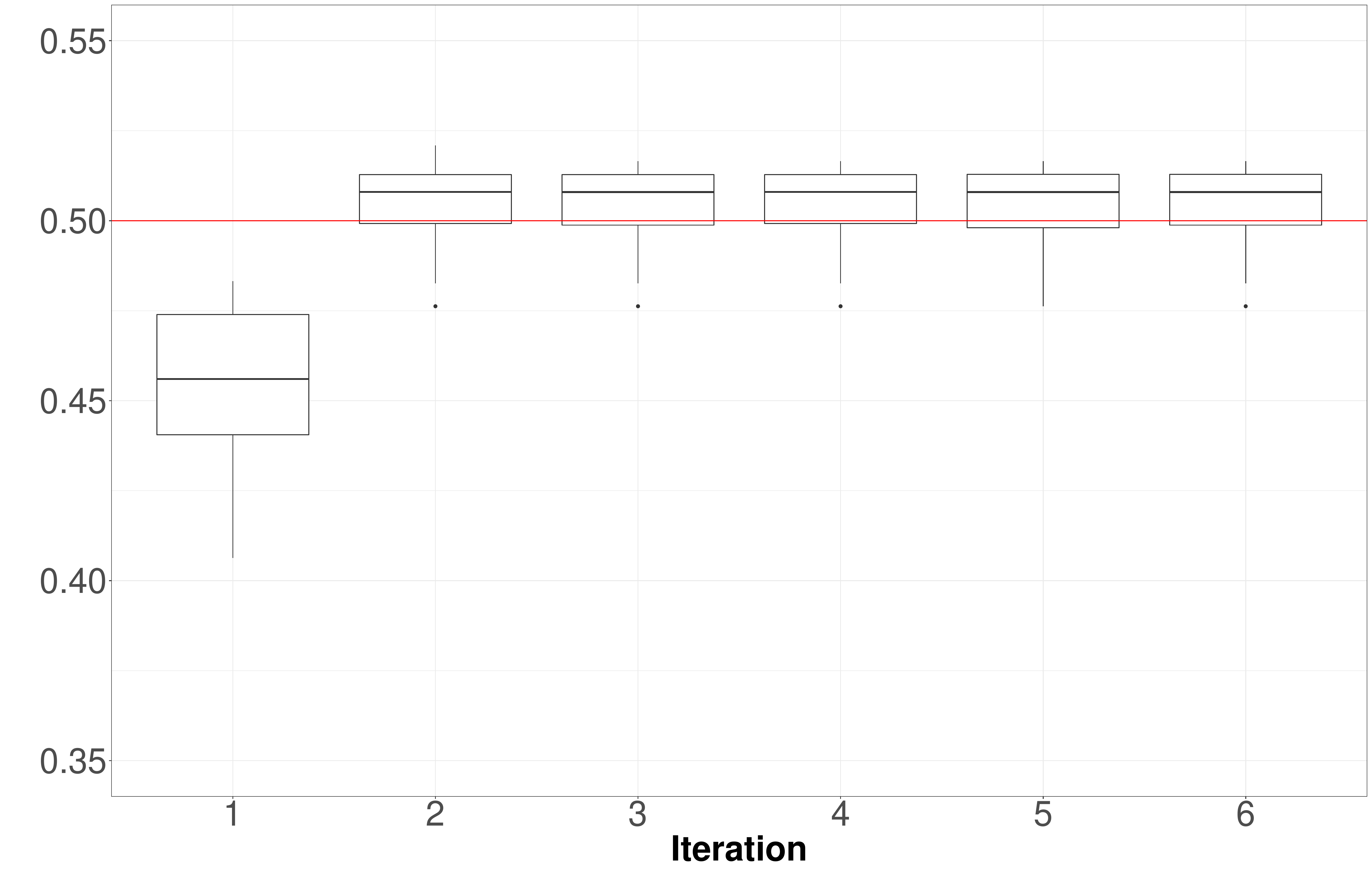}
  &  \includegraphics[width=0.32\textwidth, height=4.5cm]{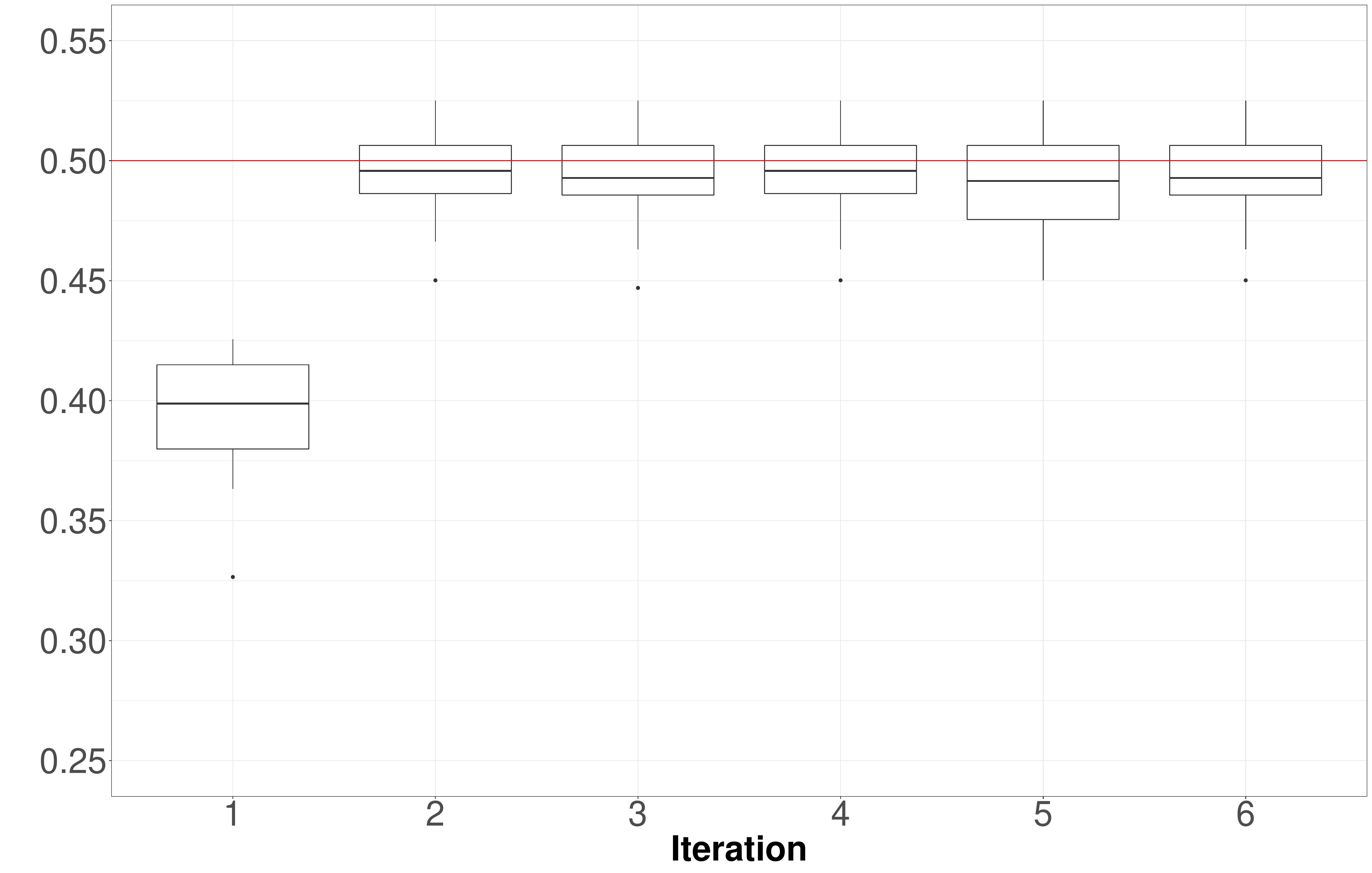} 
  & \includegraphics[width=0.32\textwidth, height=4.5cm]{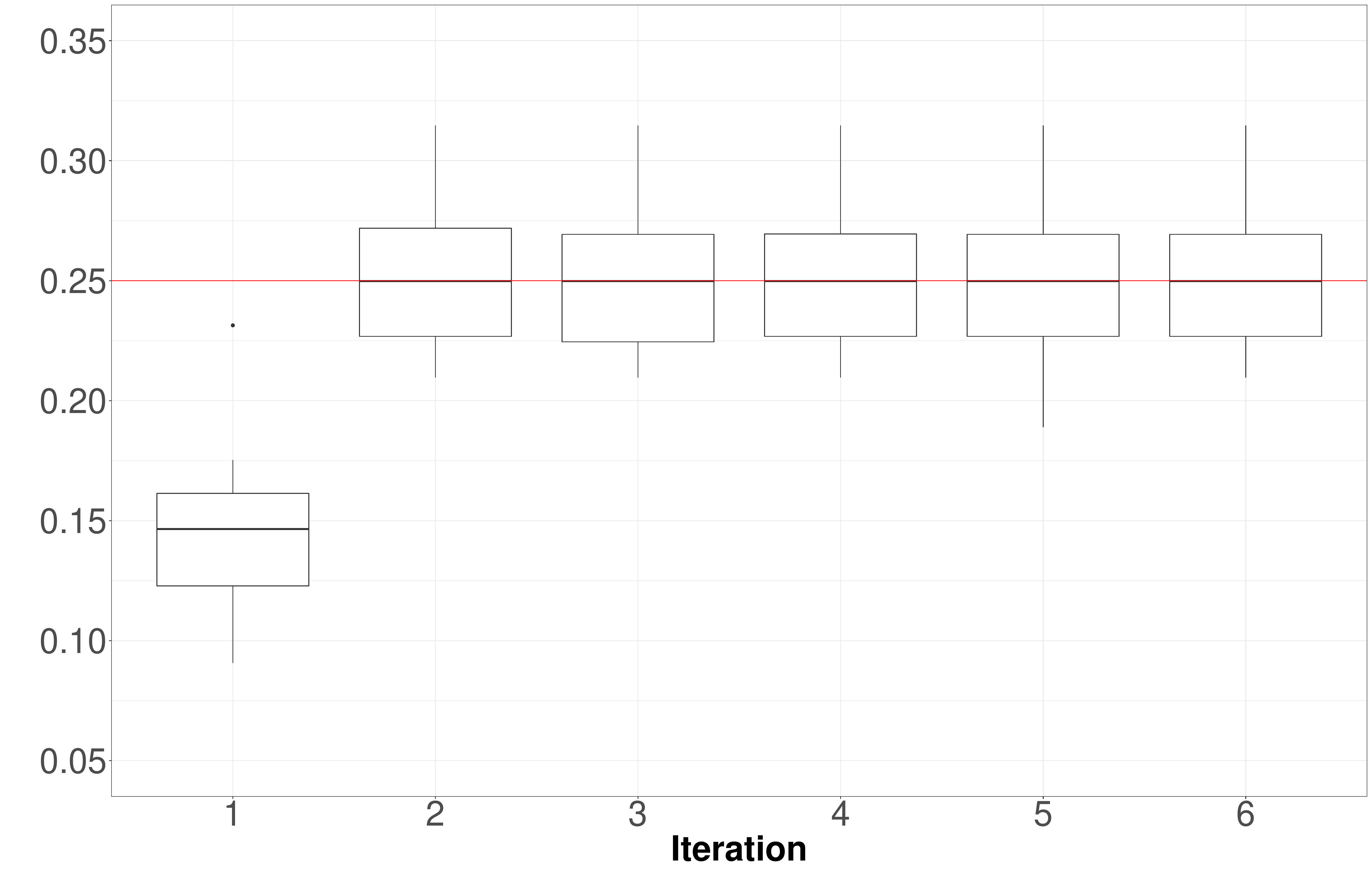}\\
\includegraphics[width=0.32\textwidth, height=4.5cm]{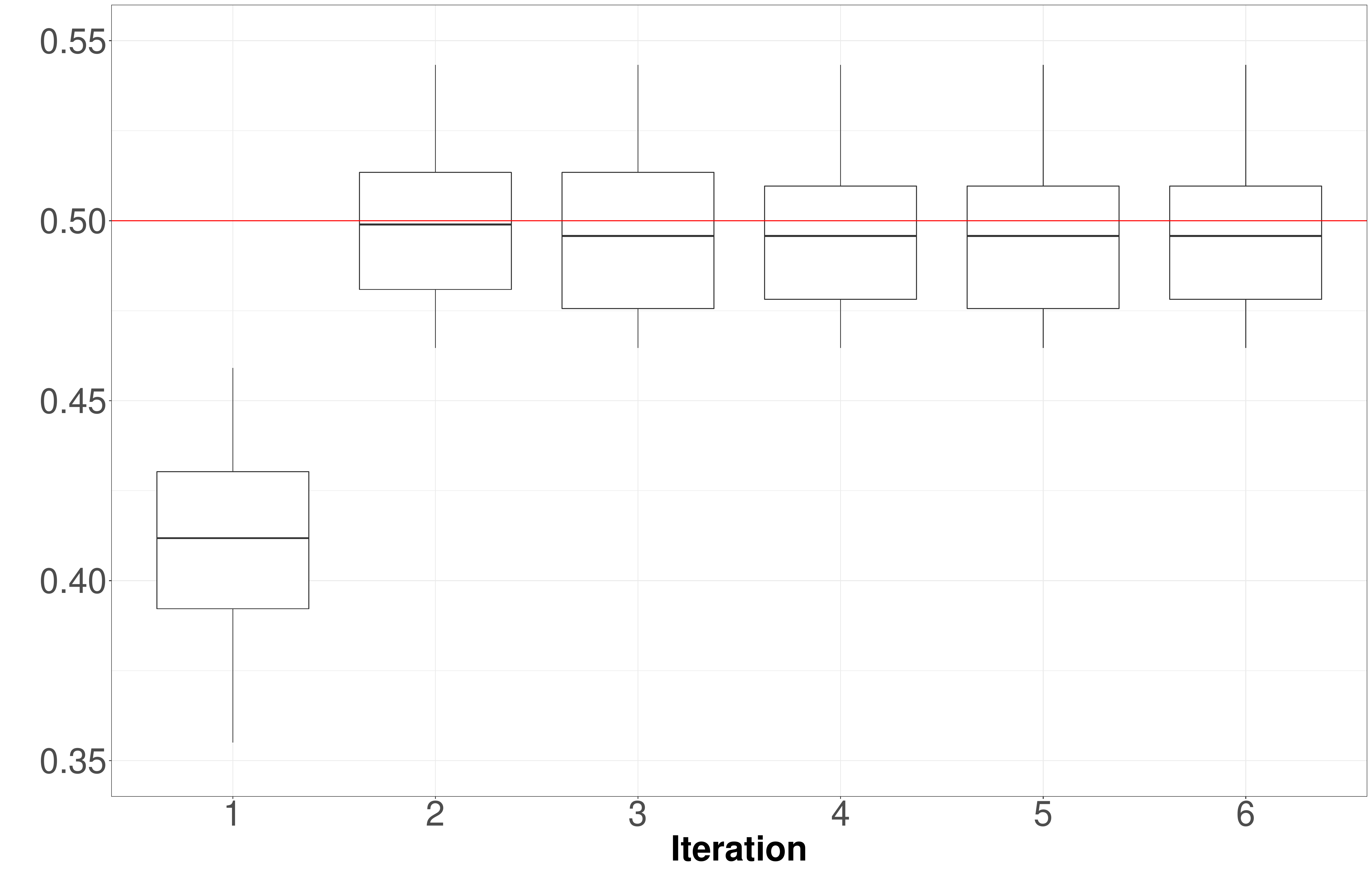}
  &  \includegraphics[width=0.32\textwidth, height=4.5cm]{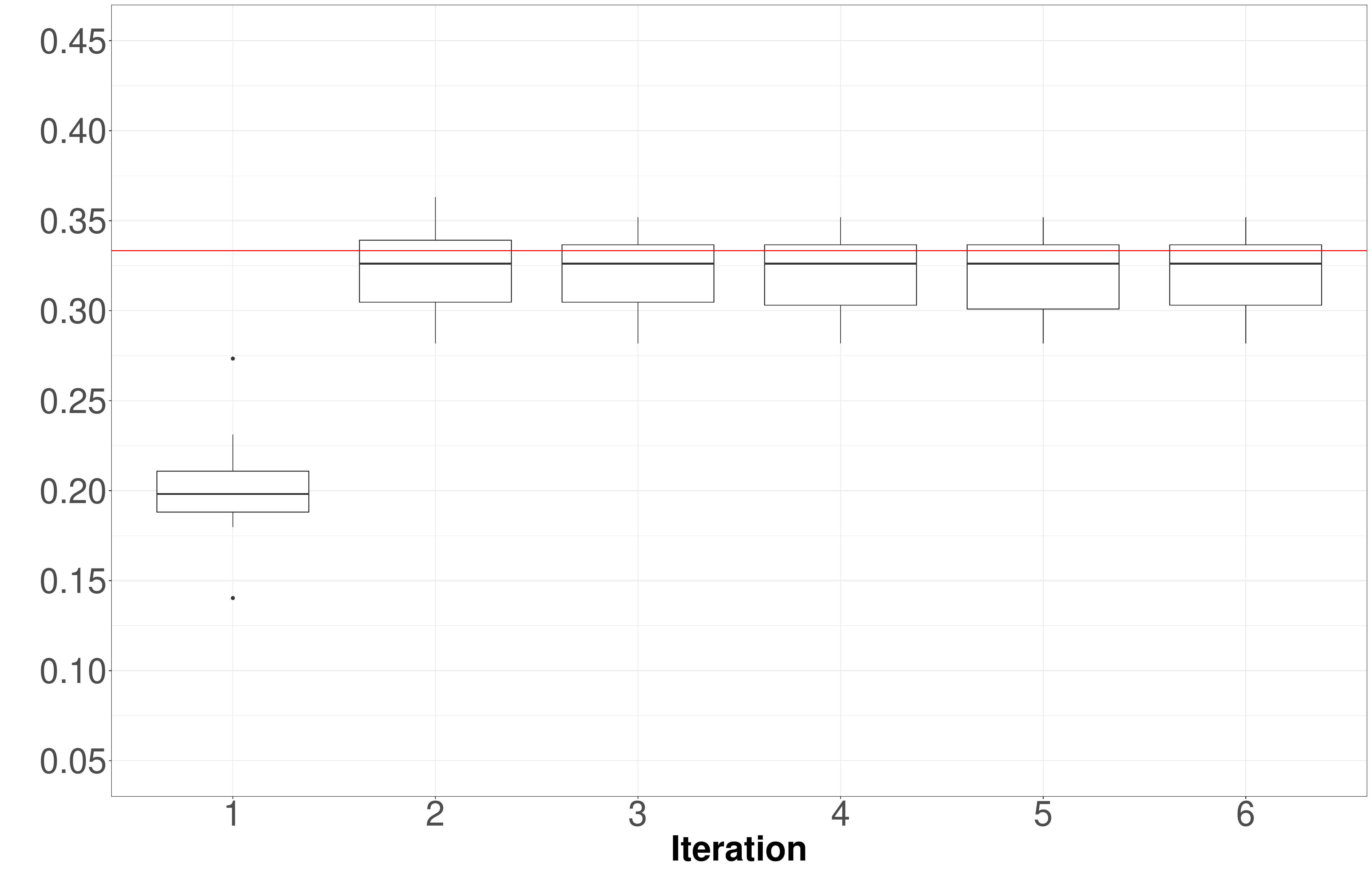} 
  & \includegraphics[width=0.32\textwidth, height=4.5cm]{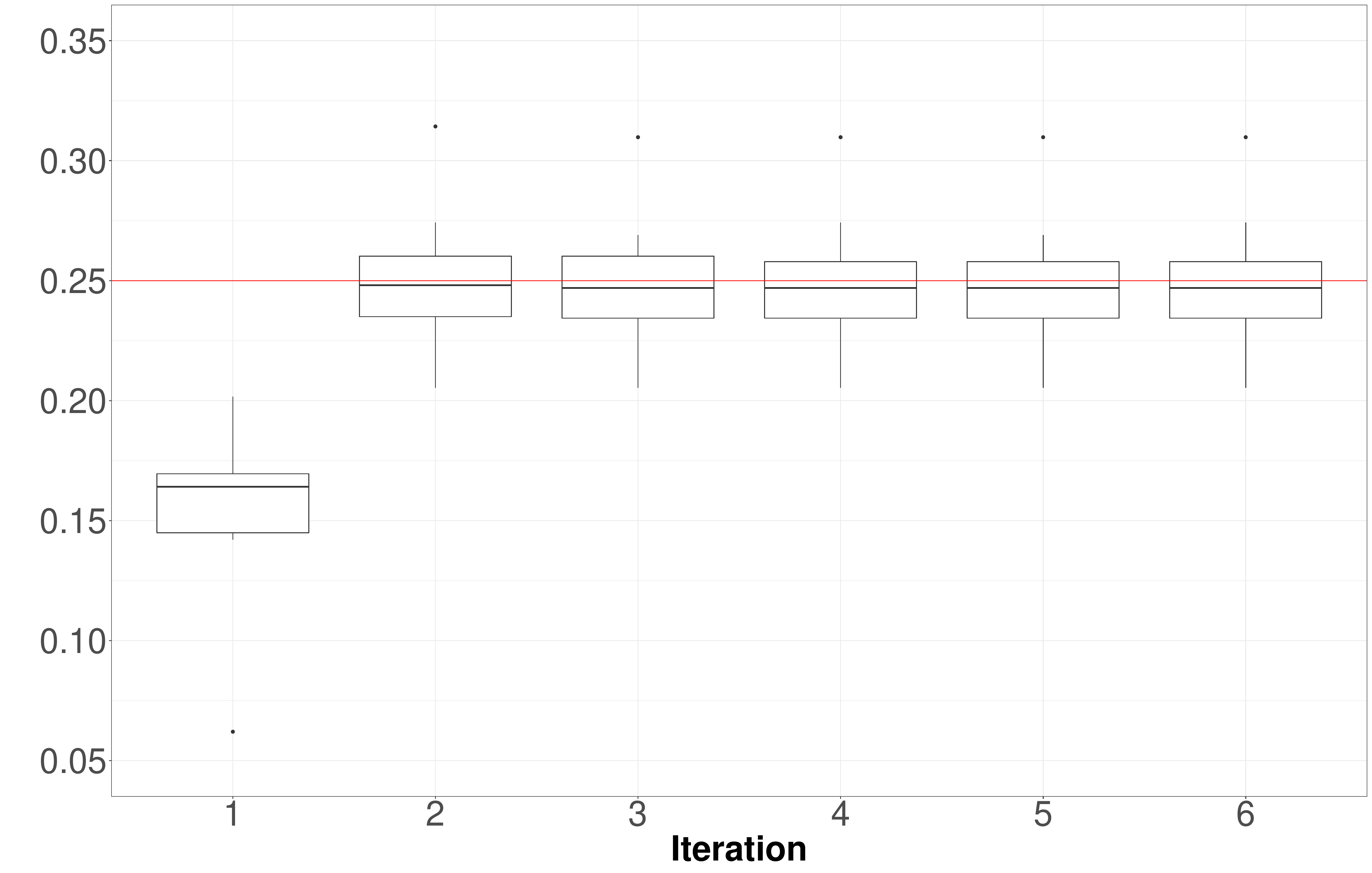}\\
\end{tabular}  
\caption{Boxplots for the estimations of $\boldsymbol{\gamma}^\star$ in Model (\ref{eq:mut_Wt}) with a 5\% sparsity level and $q=1,2,3$ obtained by
  \texttt{ss\_cv}.
Top: $q=1$ and $\gamma_1^\star=0.5$ (left), $q=2$ and $\gamma_1^\star=0.5$ (middle), $q=2$ and $\gamma_2^\star=0.25$ (right). Bottom: $q=3$ and $\gamma_1^\star=0.5$ (left), $q=3$ and  $\gamma_2^\star=1/3$ (middle), $q=3$ and $\gamma_3^\star=0.25$ (right).
The horizontal lines correspond to the values of the $\gamma_i^\star$'s. \label{fig:gamma:5:cv}}
 \end{center}
\end{figure}

\begin{figure}[!htbp]
  \begin{center}
\begin{tabular}{ccc}
  \includegraphics[width=0.32\textwidth, height=4.5cm]{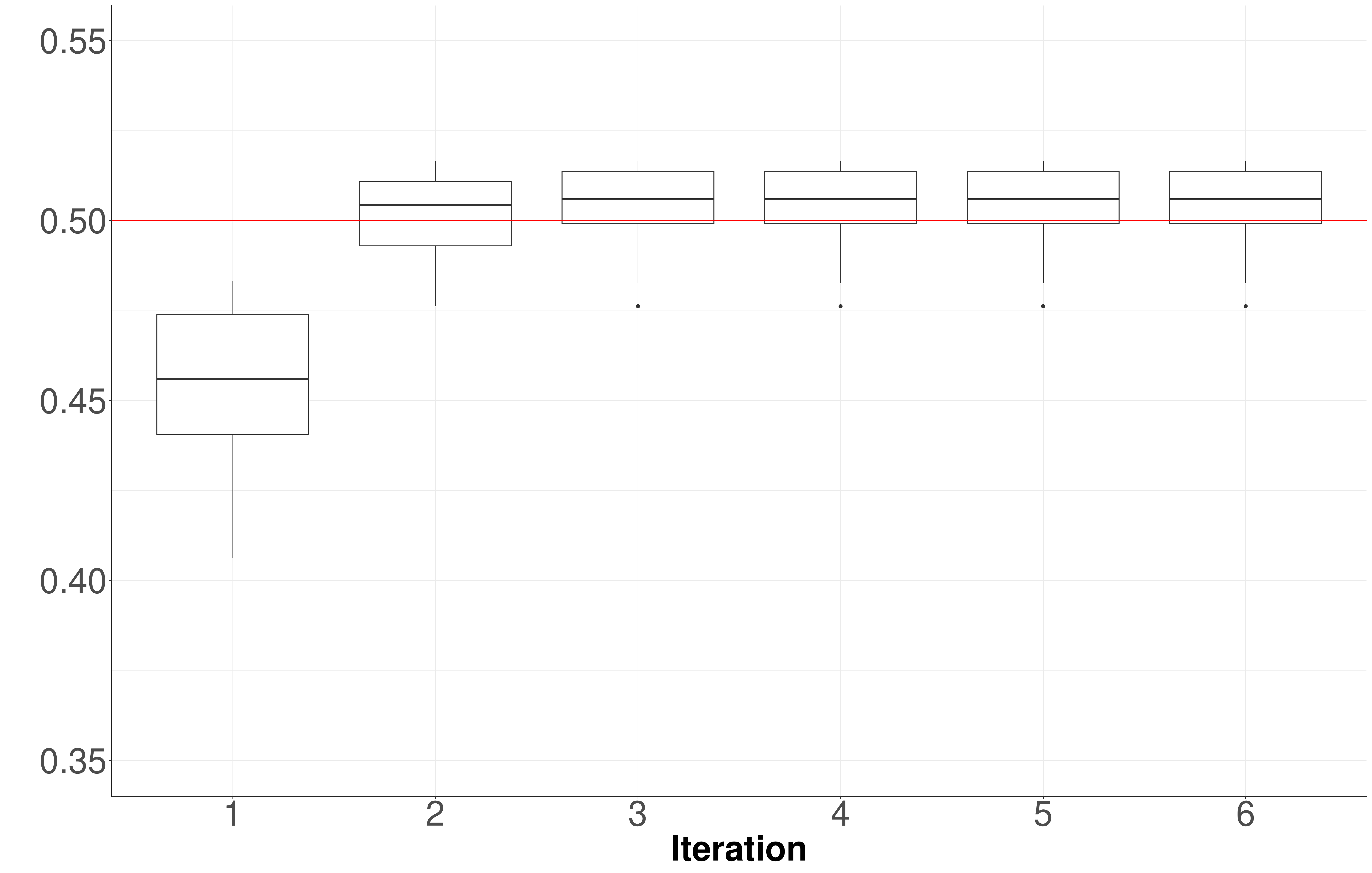}
  &  \includegraphics[width=0.32\textwidth, height=4.5cm]{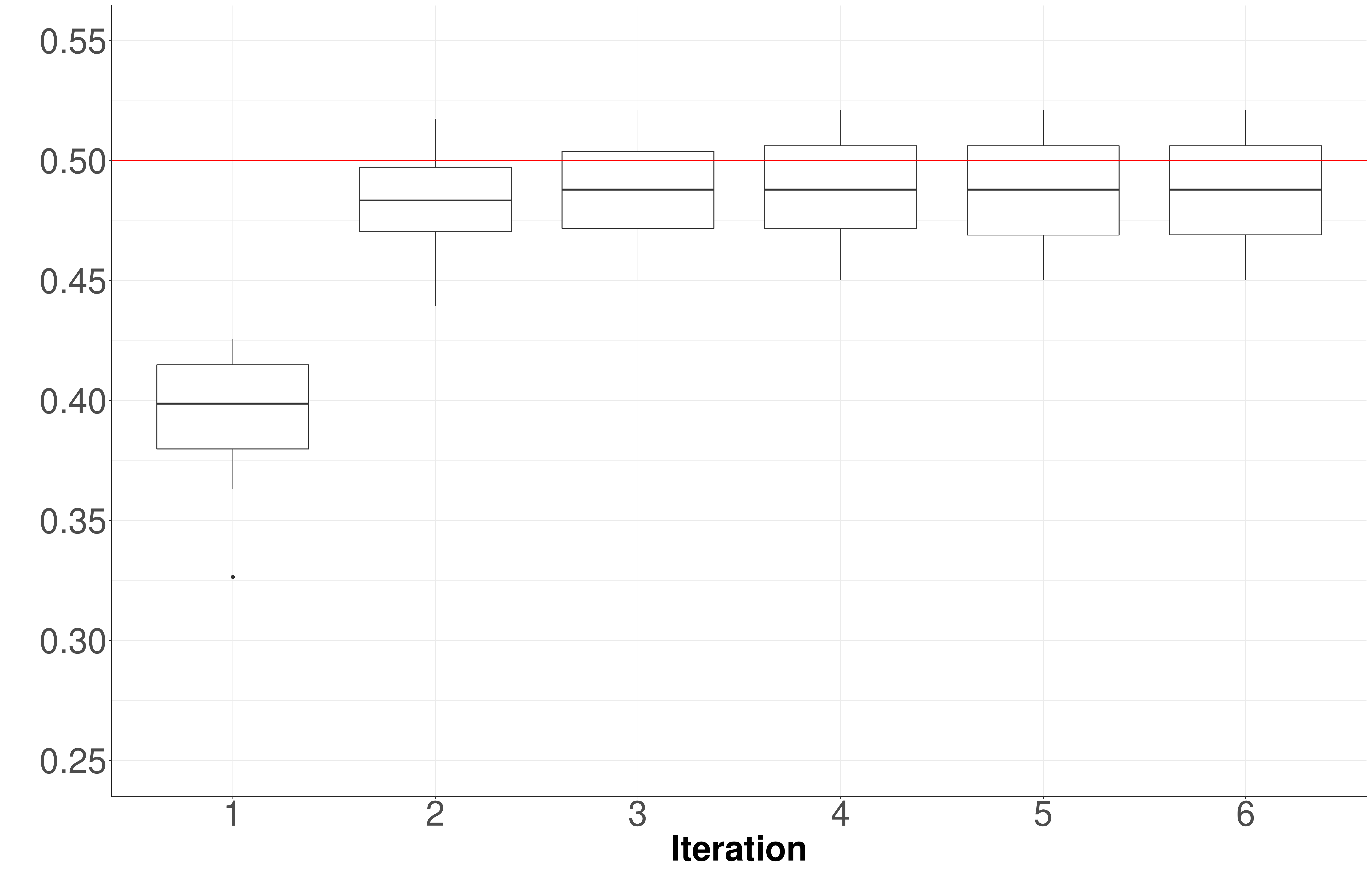} 
  & \includegraphics[width=0.32\textwidth, height=4.5cm]{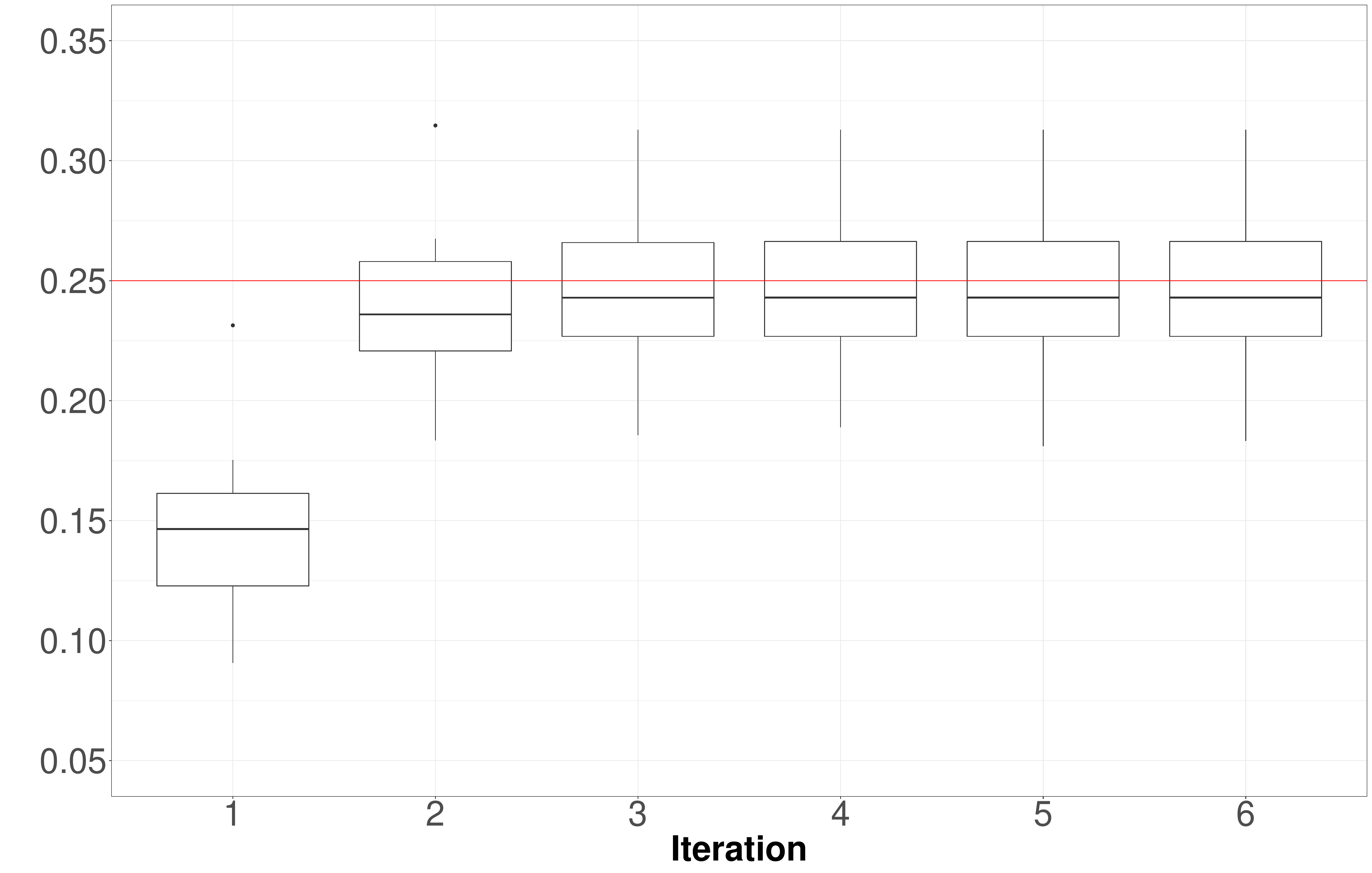}\\
\includegraphics[width=0.32\textwidth, height=4.5cm]{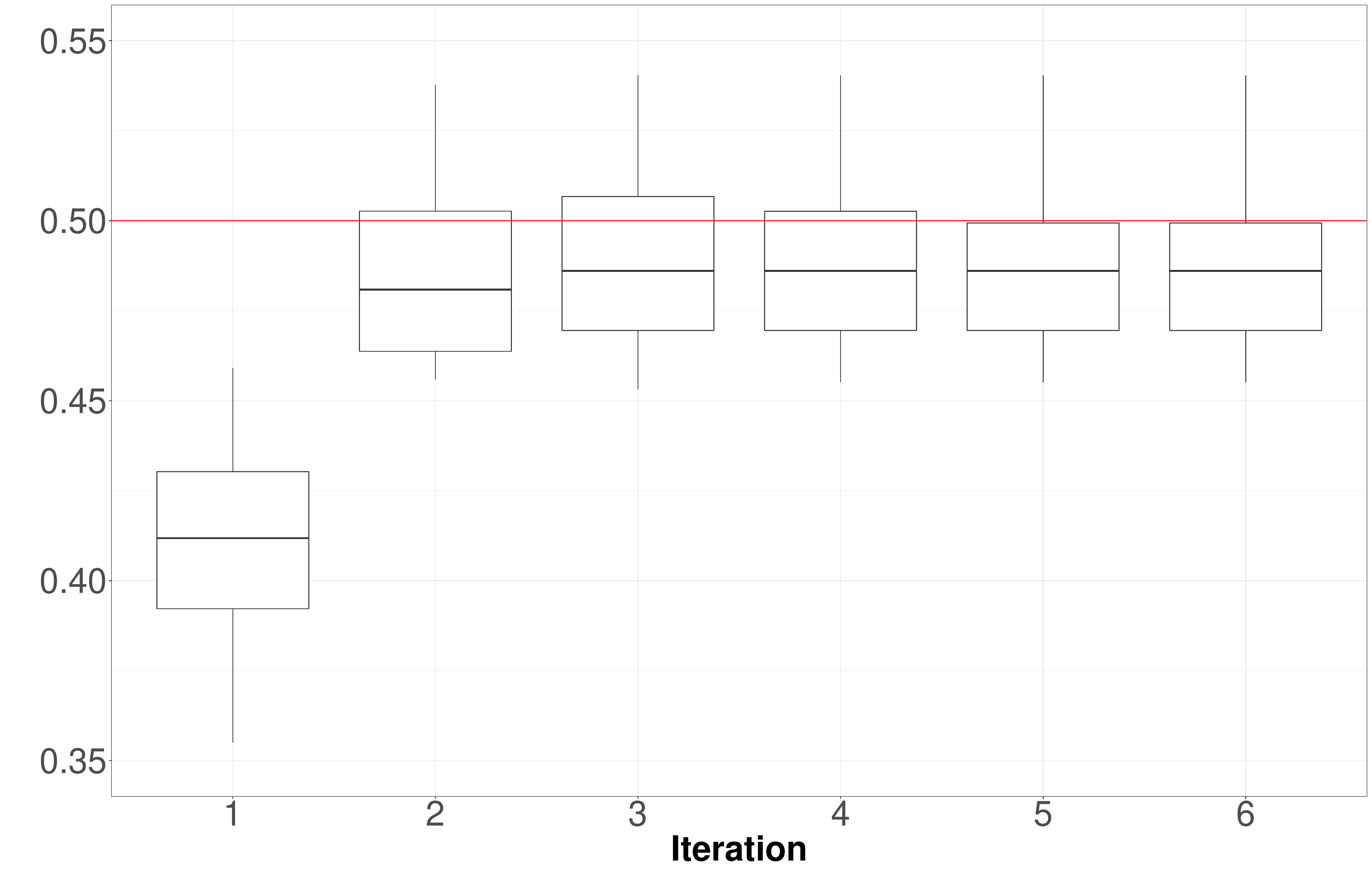}
  &  \includegraphics[width=0.32\textwidth, height=4.5cm]{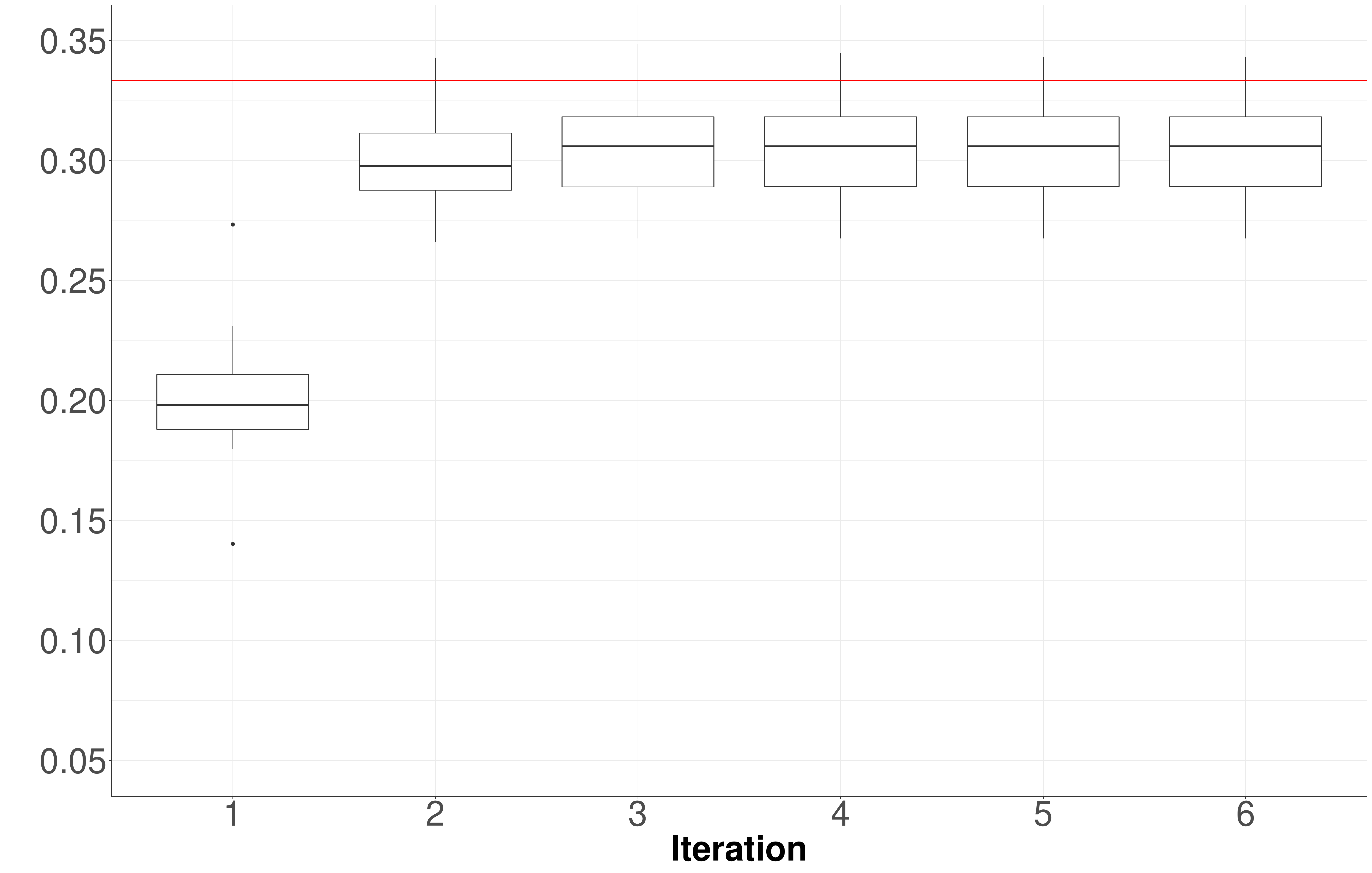} 
  & \includegraphics[width=0.32\textwidth, height=4.5cm]{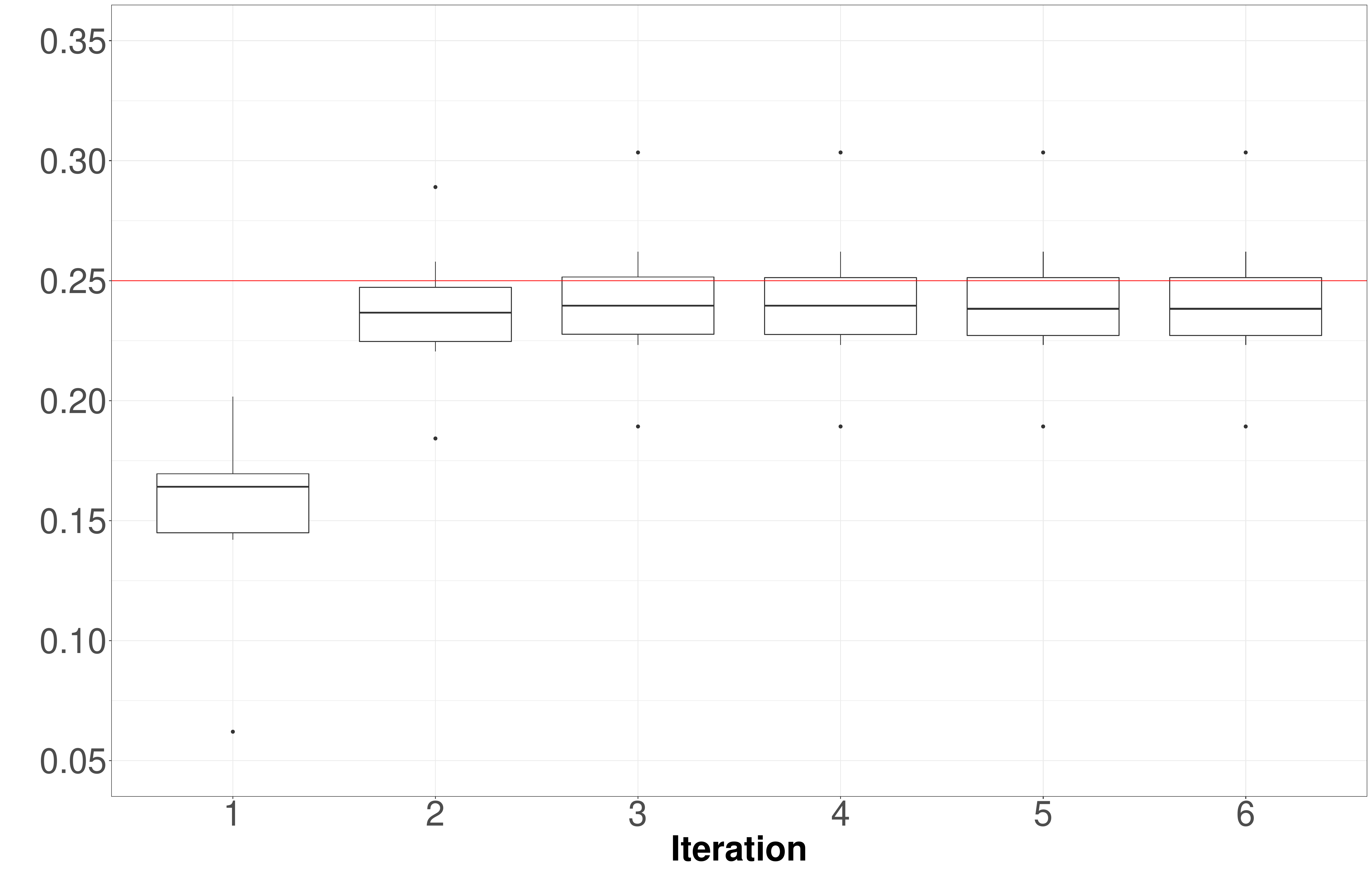}\\
\end{tabular}  
\caption{Boxplots for the estimations of $\boldsymbol{\gamma}^\star$ in Model (\ref{eq:mut_Wt}) with a 5\% sparsity level and $q=1,2,3$ obtained by
  \texttt{fast\_ss}.
 Top: $q=1$ and $\gamma_1^\star=0.5$ (left), $q=2$ and $\gamma_1^\star=0.5$ (middle), $q=2$ and $\gamma_2^\star=0.25$ (right). Bottom: $q=3$ and $\gamma_1^\star=0.5$ (left), $q=3$ and  $\gamma_2^\star=1/3$ (middle), $q=3$ and $\gamma_3^\star=0.25$ (right).
   The horizontal lines correspond to the values of the $\gamma_i^\star$'s.\label{fig:gamma:5:fast}}
 \end{center}
\end{figure}

\begin{figure}[!htbp]
  \begin{center}
\begin{tabular}{ccc}
  \includegraphics[width=0.32\textwidth, height=4.5cm]{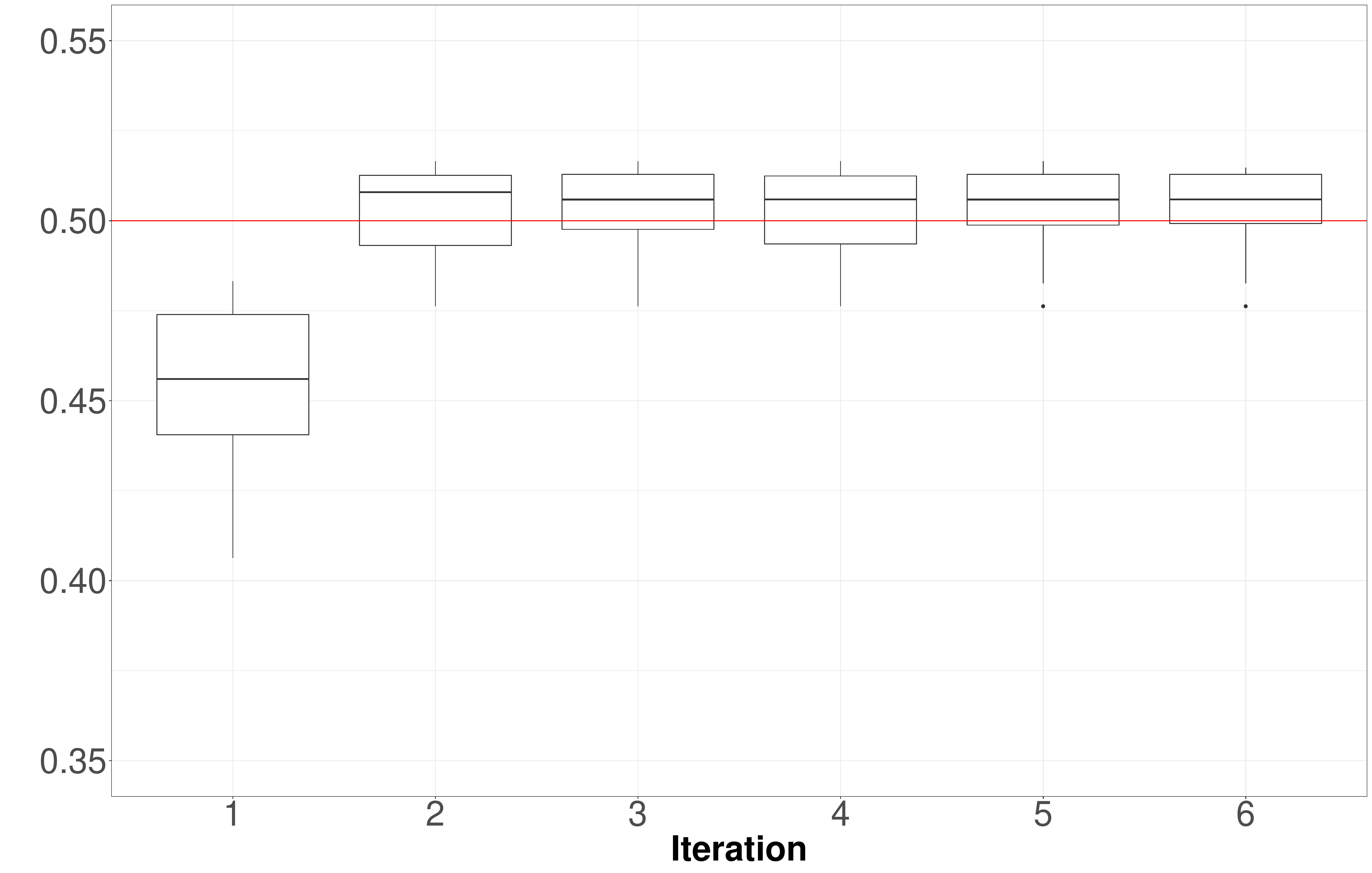}
  &  \includegraphics[width=0.32\textwidth, height=4.5cm]{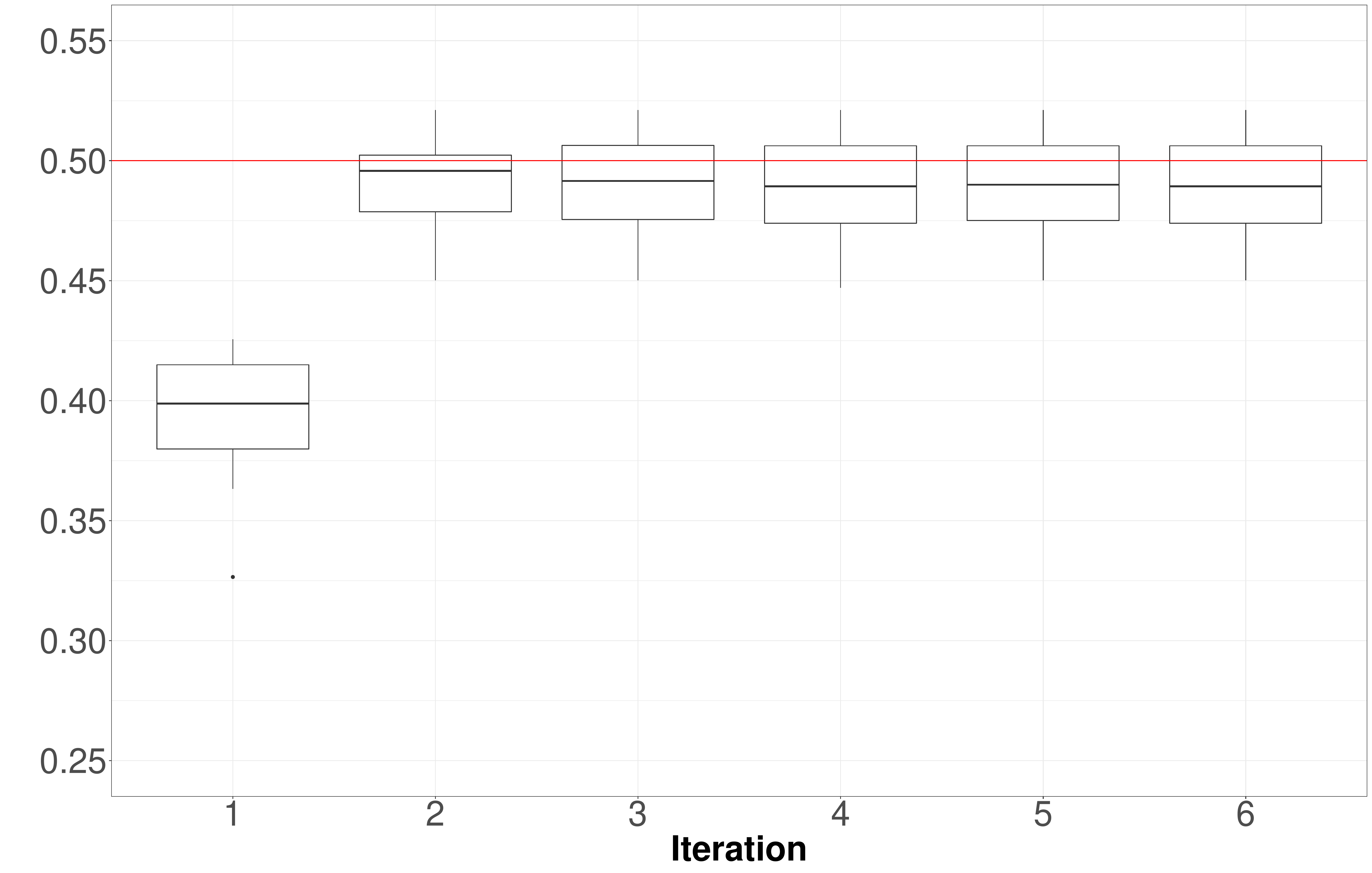} 
  & \includegraphics[width=0.32\textwidth, height=4.5cm]{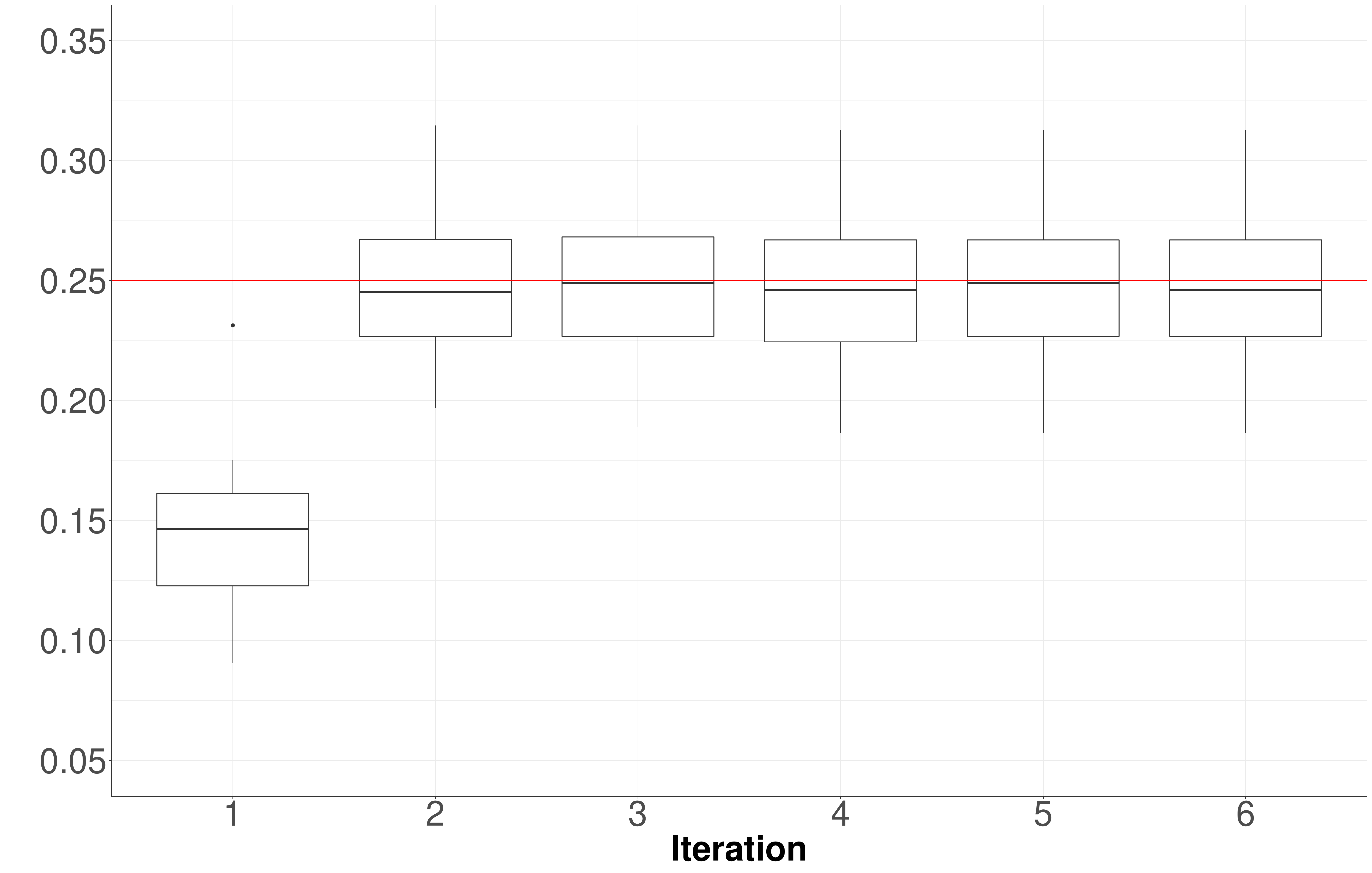}\\
\includegraphics[width=0.32\textwidth, height=4.5cm]{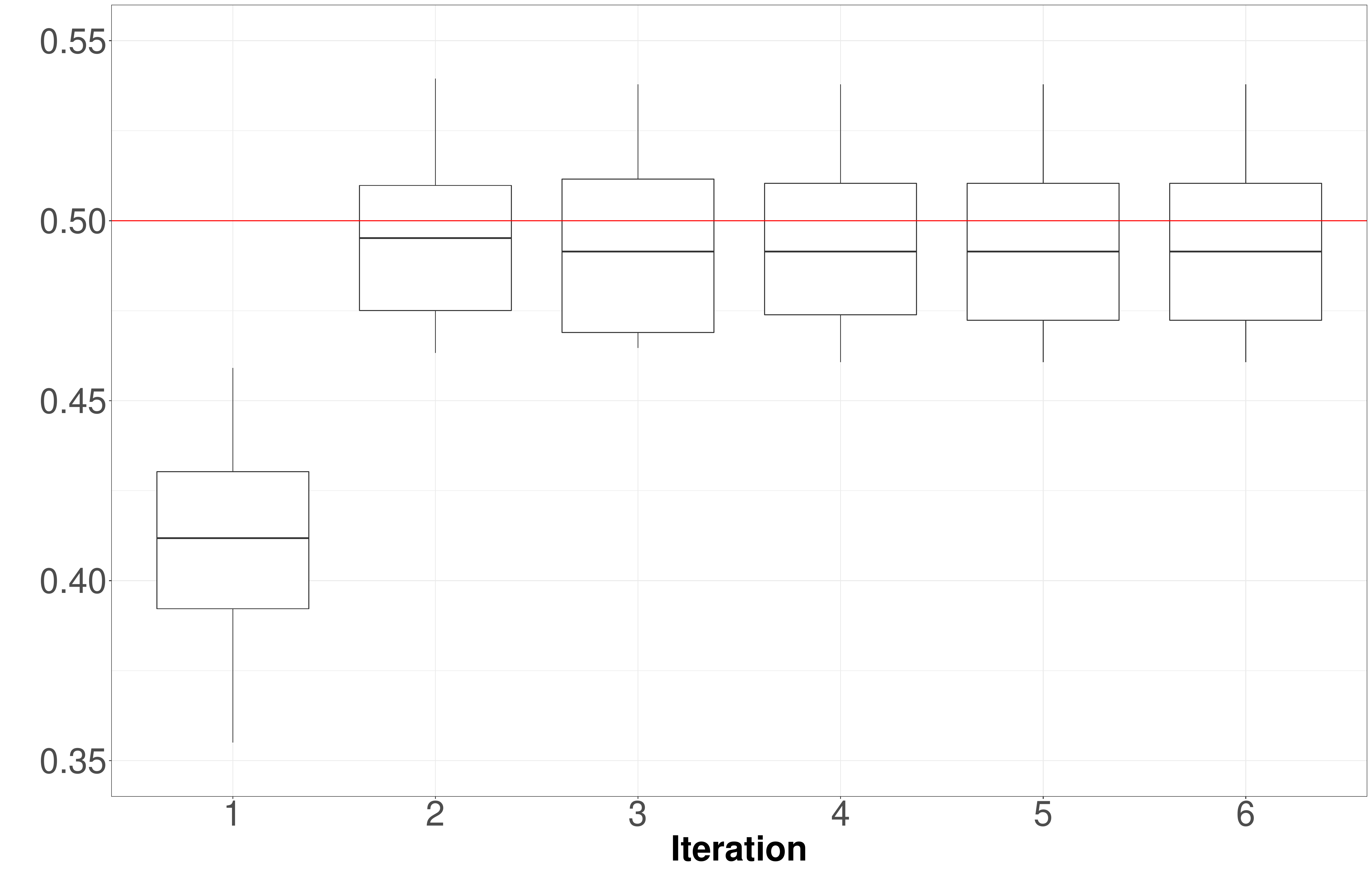}
  &  \includegraphics[width=0.32\textwidth, height=4.5cm]{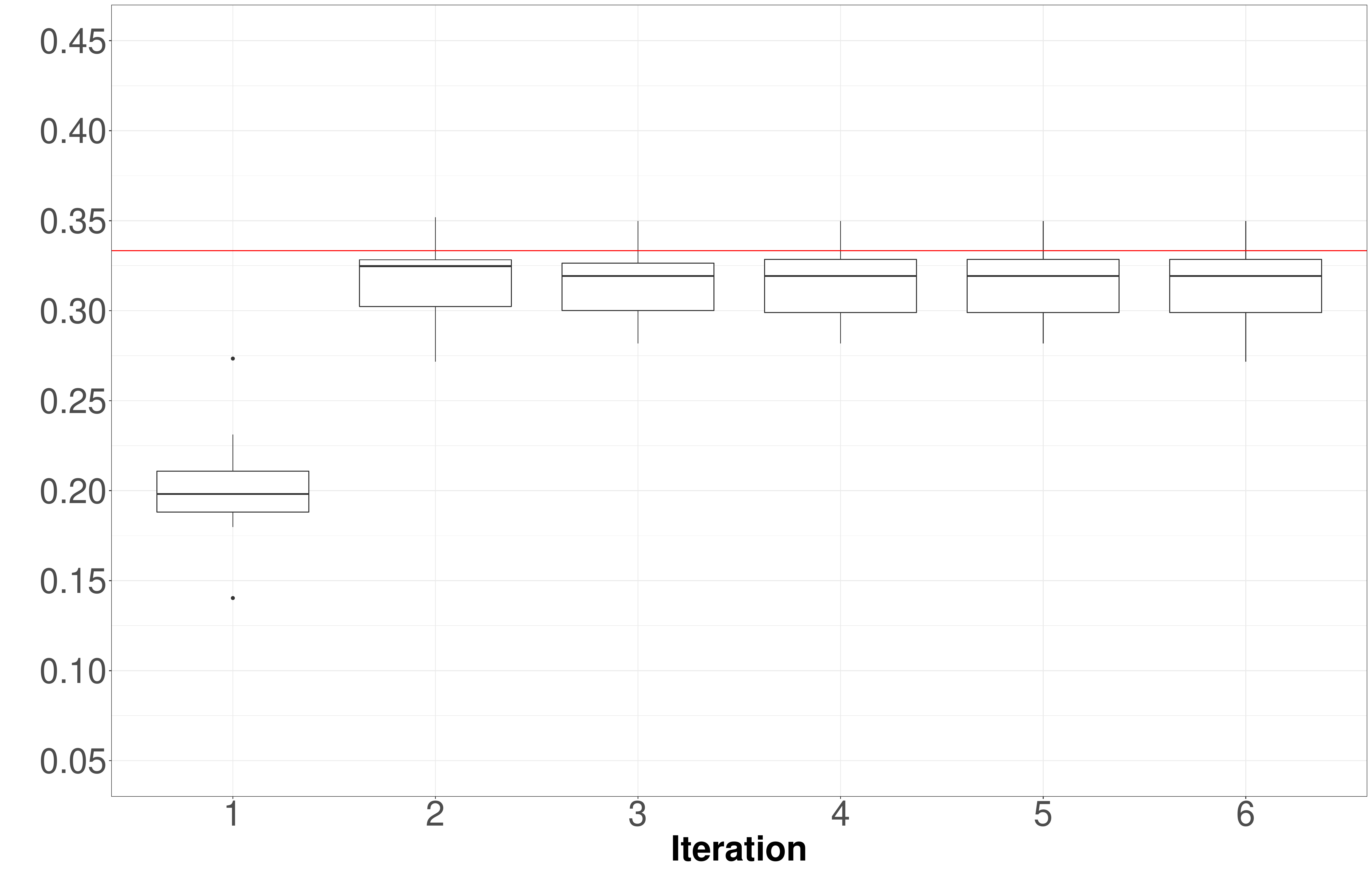} 
  & \includegraphics[width=0.32\textwidth, height=4.5cm]{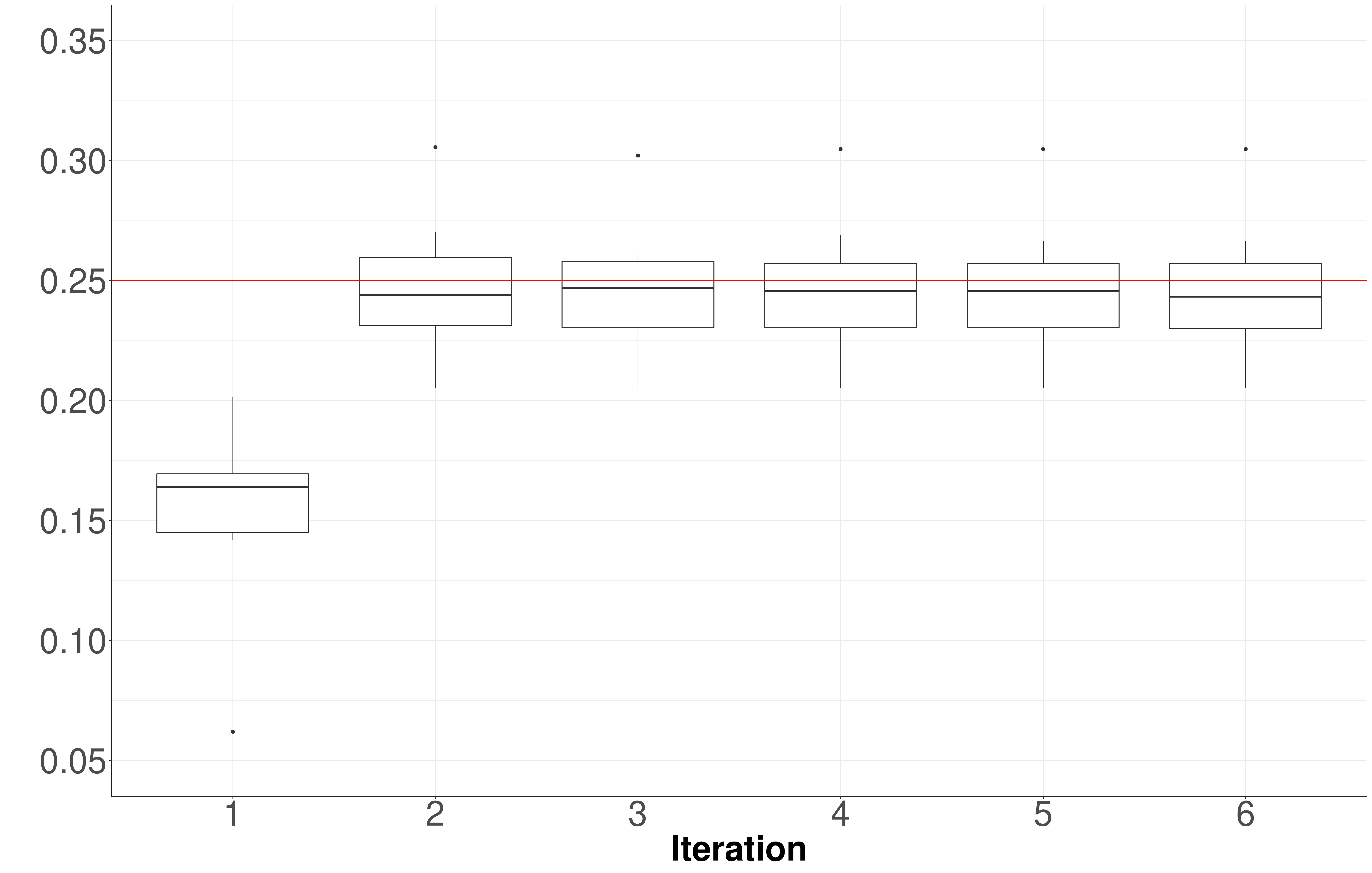}\\
\end{tabular}  
\caption{Boxplots for the estimations of $\boldsymbol{\gamma}^\star$ in Model (\ref{eq:mut_Wt}) with a 5\% sparsity level and $q=1,2,3$ obtained by
  \texttt{ss\_min}.
 Top: $q=1$ and $\gamma_1^\star=0.5$ (left), $q=2$ and $\gamma_1^\star=0.5$ (middle), $q=2$ and $\gamma_2^\star=0.25$ (right). Bottom: $q=3$ and $\gamma_1^\star=0.5$ (left), $q=3$ and  $\gamma_2^\star=1/3$ (middle), $q=3$ and $\gamma_3^\star=0.25$ (right).
   The horizontal lines correspond to the values of the $\gamma_i^\star$'s.\label{fig:gamma:5:min}}
 \end{center}
\end{figure}

\textbf{Impact of the value of $n$}

In this paragraph, we study the impact of the value of $n$ on the TPR and the FPR associated to the support recovery of $\boldsymbol{\beta}^\star$
and on the estimation of $\boldsymbol{\gamma}^\star$ for
\texttt{ss\_min}, the other approaches providing similar results.

Based on Figures \ref{fig:TPR:FPR:thresh_5} and \ref{fig:TPR:FPR:thresh_10},
we chose a threshold equal to 0.7 for both sparsity levels (5\% and 10\%) which provides a good trade-off between TPR and FPR for all values of $n$.
We can see from Figure \ref{fig:TPR:FPR:n} that \texttt{ss\_min} with this threshold
outperforms \texttt{lasso\_cv} when the sparsity level is equal to 5\% and all the values of $n$ considered. In the case where the
sparsity level is equal to 10\%, \texttt{lasso\_cv} has a slightly larger TPR for $n=150$ and $n=200$. However, the FPR of \texttt{ss\_min}
is much smaller.

\begin{figure}[!htbp]
  \includegraphics[scale=0.22]{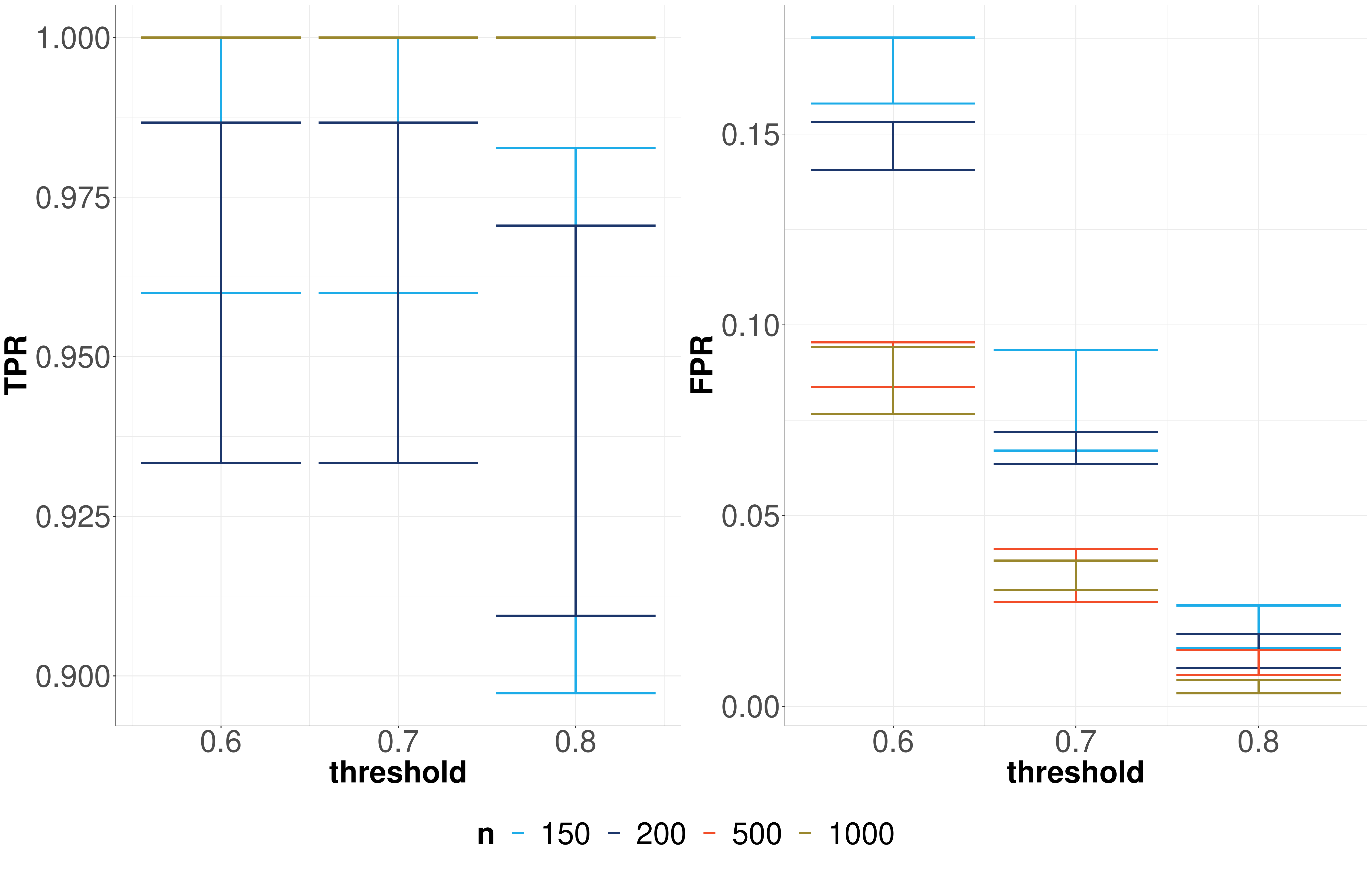}
  \caption{Error bars of the TPR and FPR associated to the support recovery of $\boldsymbol{\beta}^\star$ for \texttt{ss\_min} with respect to the thresholds
    for different values of $n$, $q=1$, $p=100$ and a 5\% sparsity level. 
 \label{fig:TPR:FPR:thresh_5}}
\end{figure}

\begin{figure}[!htbp]
  \includegraphics[scale=0.22]{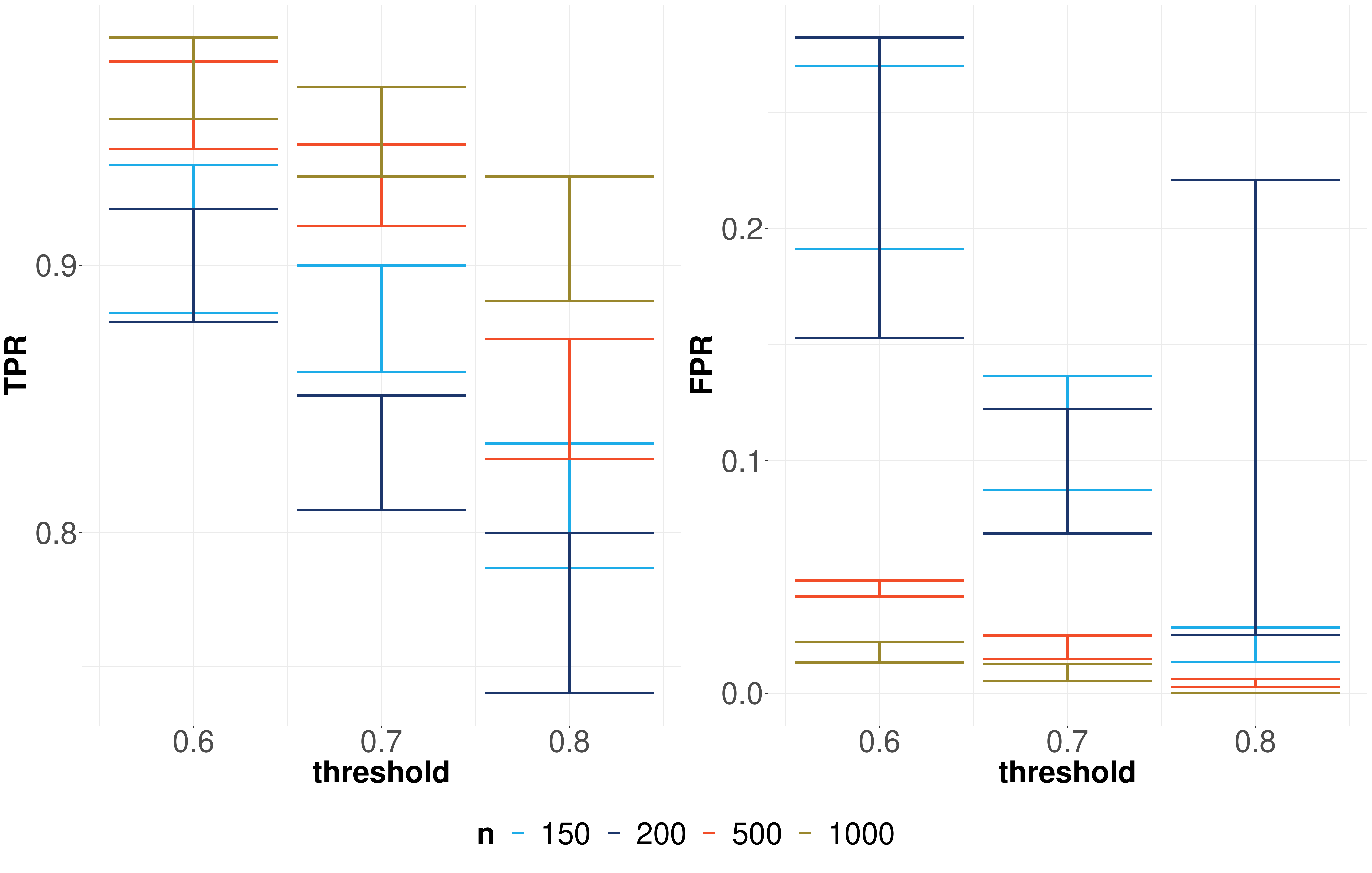}
  \caption{Error bars of the TPR and FPR associated to the support recovery of $\boldsymbol{\beta}^\star$ for \texttt{ss\_min} with respect to the thresholds
    for different values of $n$, $q=1$, $p=100$ and a 10\% sparsity level. 
 \label{fig:TPR:FPR:thresh_10}}
\end{figure}

\begin{figure}[!htbp]
  \includegraphics[scale=0.22]{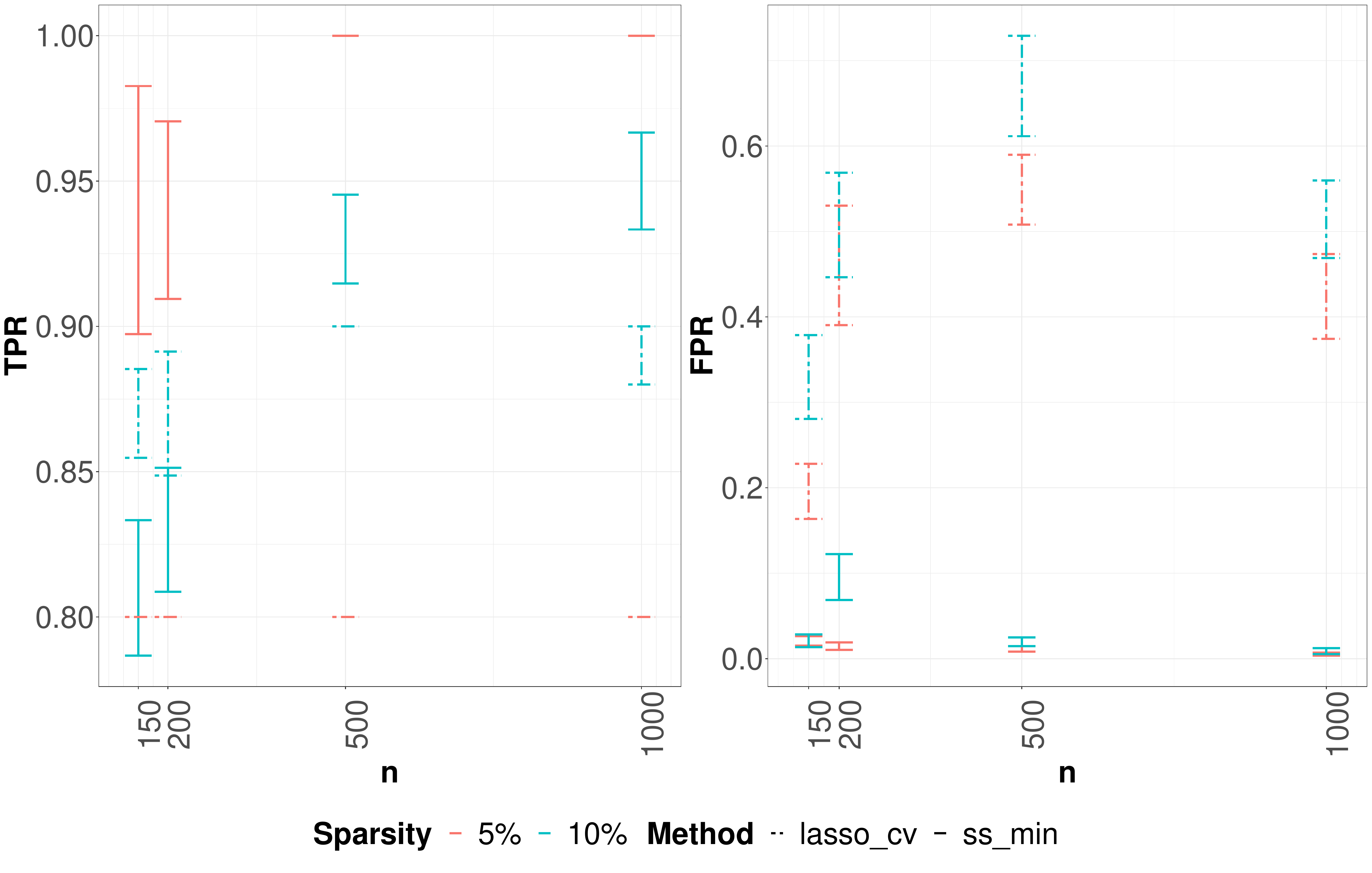}
  \caption{Error bars of the TPR and FPR associated to the support recovery of $\boldsymbol{\beta}^\star$ for \texttt{ss\_min} and \texttt{lasso\_cv}
    for different values of $n$, $q=1$, $p=100$ and different sparsity levels. \label{fig:TPR:FPR:n}}
\end{figure}

Figure \ref{fig:gamma:iter} displays the boxplots for the estimations of $\boldsymbol{\gamma}^\star$ in Model (\ref{eq:mut_Wt}) for $q=1$, $p=100$, different values of $n$ (150, 200, 500, 1000) and sparsity levels (5\% and 10\%) obtained by \texttt{ss\_min} with a threshold of 0.7 for six iterations.
We can see from this figure that this approach provides accurate estimations of $\gamma_1^\star$ from Iteration 2 especially when 
$n$ is larger than 200.

\begin{figure}[!htbp]
  \includegraphics[scale=0.22]{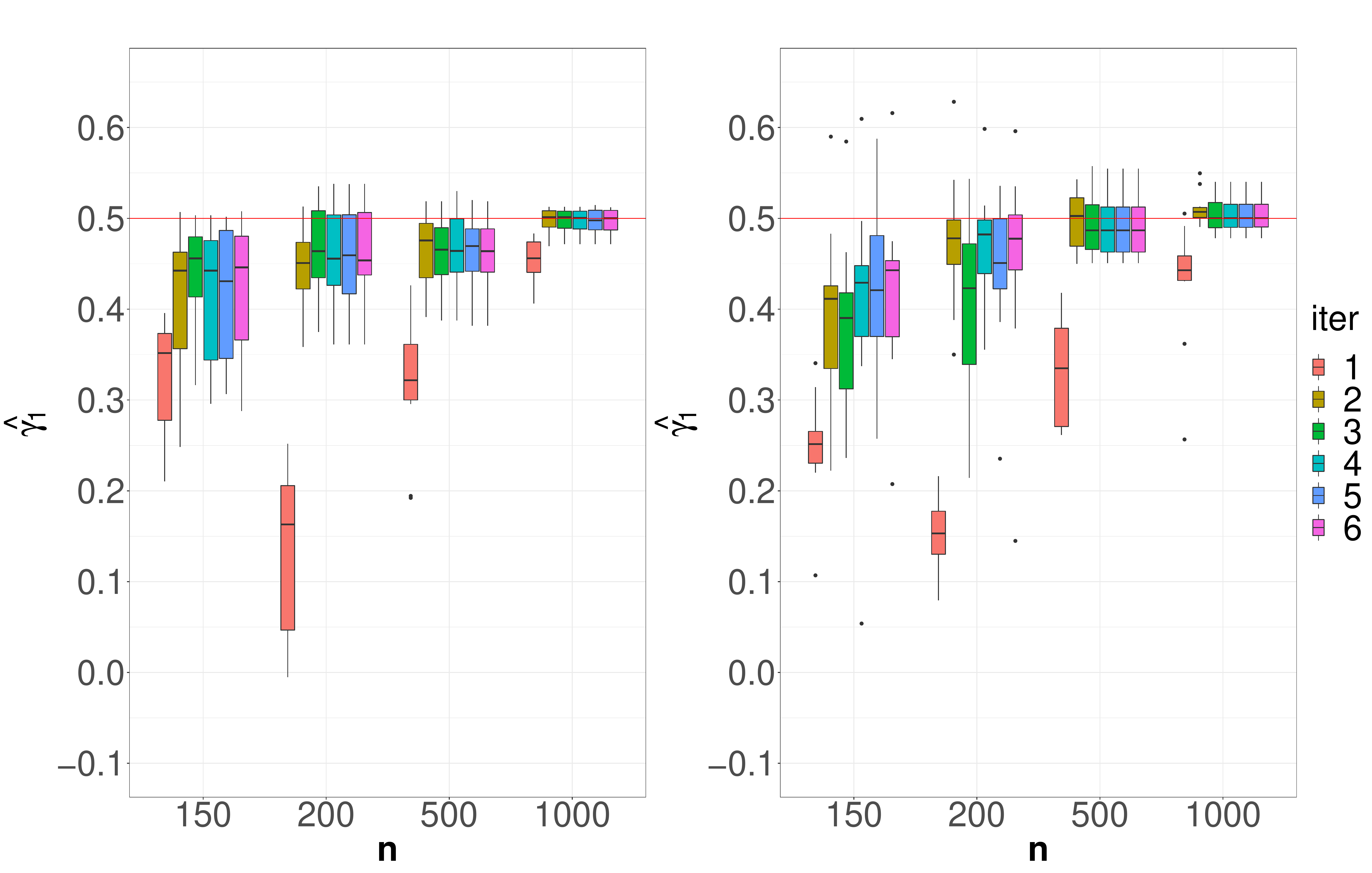}
  \caption{Boxplots for the estimations of $\boldsymbol{\gamma}^\star$ in Model (\ref{eq:mut_Wt}) for $q=1$, $p=100$, different values of $n$ and 
   sparsity levels (left: 5\%, right: 10\%) obtained by \texttt{ss\_min} with a threshold of 0.7 for different iterations (\texttt{iter}). 
 \label{fig:gamma:iter}}
\end{figure}

\subsection{Numerical performance}

Figure \ref{fig:time} displays the means of the computational times for \texttt{ss\_min} and \texttt{fast\_ss}. The performance
of \texttt{ss\_cv} are not displayed since they are similar to the one of \texttt{ss\_min}.
We can see from this figure that it takes around 1 minute to process observations $Y_1,\dots,Y_n$
satisfying Model (\ref{eq:Yt}) for a given threshold and one iteration, when $n=1000$ and $p=100$.
Moreover, we can observe that the computational burden of \texttt{fast\_ss} is slightly
smaller than the one of \texttt{ss\_min}. 

\begin{figure}[!htbp]
\includegraphics[scale=0.2]{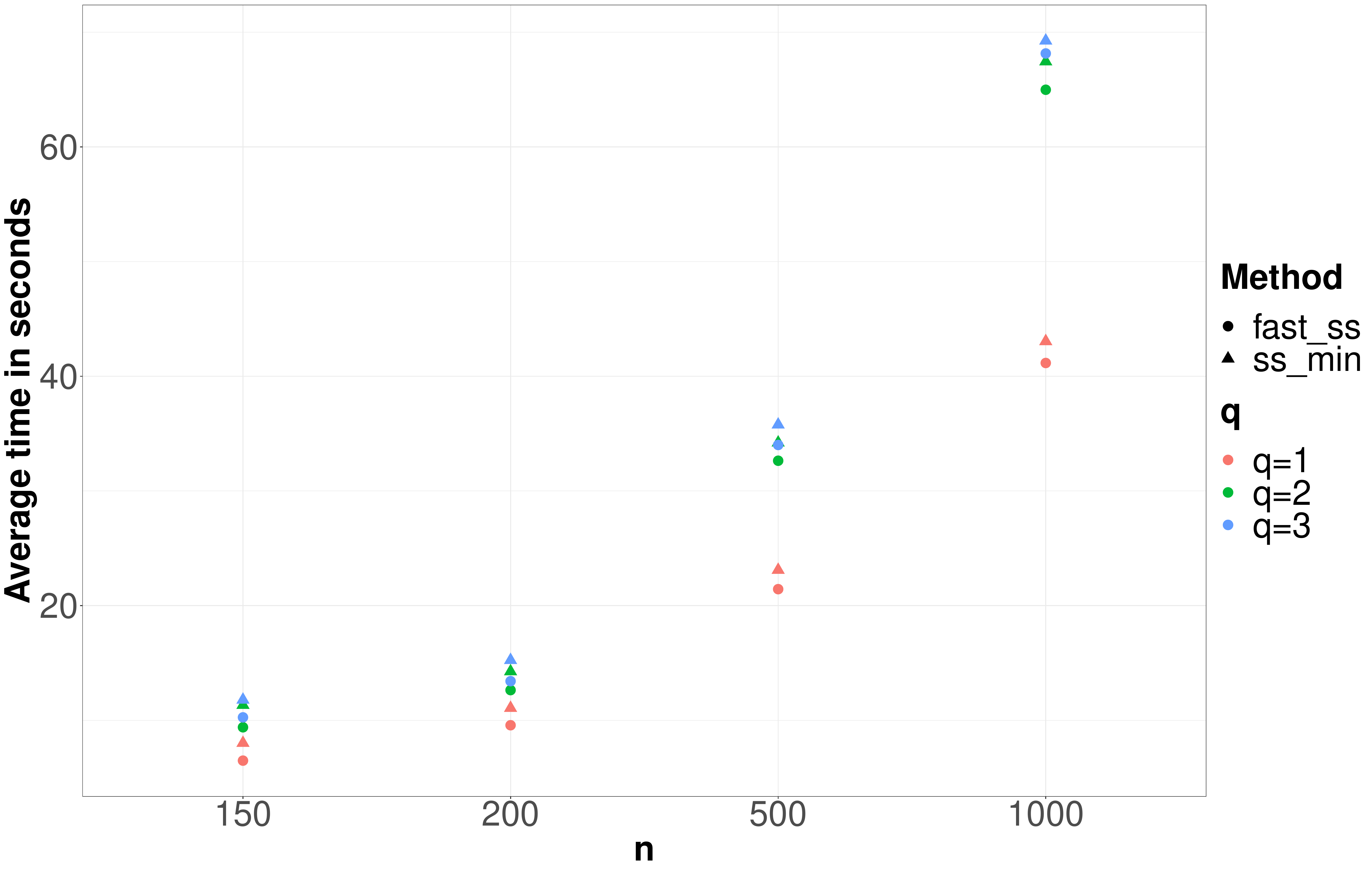}
\caption{Means of the computational times in seconds for \texttt{ss\_min} and \texttt{fast\_ss} in the case where $p=100$, and different values
 of $n$ and $q$, a given threshold and one iteration.\label{fig:time}}
\end{figure}

\section{Proofs}\label{sec:proofs}

\subsection{\textcolor{black}{Computation of the first and second derivatives of $W_t$ defined in (\ref{eq:Wt})}}

The computations given below are similar to those provided in \cite{davis:dunsmuir:street:2005} but are specific to the parametrization 
$\boldsymbol{\delta}=(\boldsymbol{\beta}',\boldsymbol{\gamma}')$ considered in this paper.

\subsubsection{\textcolor{black}{Computation of the first derivatives of $W_t$ }}\label{subsub:first_derive}

By the definition of $W_t$ given in (\ref{eq:Wt}), we get
\begin{equation*}
\frac{\partial W_t}{\partial \boldsymbol{\delta}}(\boldsymbol{\delta})=\frac{\partial\boldsymbol{\beta}' x_t}{\partial \boldsymbol{\delta}}+\frac{\partial Z_t}{\partial \boldsymbol{\delta}}
(\boldsymbol{\delta}),
\end{equation*}
where $\boldsymbol{\beta}$, $x_t$ and $Z_t$ are defined in (\ref{eq:Wt}). 
More precisely, for all $k\in\{0,\dots,p\}$, $\ell\in\{1,\dots,q\}$ and $t\in\{1,\dots,n\}$, by (\ref{eq:Et}),
\begin{align}\label{eq:gradW_beta}
\frac{\partial W_t}{\partial \beta_k}&=x_{t,k}+\frac{\partial Z_t}{\partial \beta_k}=x_{t,k}+\sum_{j=1}^{q\wedge (t-1)}\gamma_j\frac{\partial E_{t-j}}{\partial \beta_k}\nonumber\\
&=x_{t,k}-\sum_{j=1}^{q\wedge (t-1)}\gamma_j Y_{t-j}\frac{\partial W_{t-j}}{\partial \beta_k}\exp(-W_{t-j})=x_{t,k}-\sum_{j=1}^{q\wedge (t-1)}\gamma_j(1+E_{t-j})\frac{\partial W_{t-j}}{\partial \beta_k},\\
\frac{\partial W_t}{\partial \gamma_\ell}&=E_{t-\ell}+\sum_{j=1}^{q\wedge (t-1)} \gamma_j\frac{\partial E_{t-j}}{\partial\gamma_\ell}\nonumber\\\label{eq:gradW_gamma}
&=E_{t-\ell}-\sum_{j=1}^{q\wedge (t-1)}\gamma_j Y_{t-j}\frac{\partial W_{t-j}}{\partial \gamma_\ell}\exp(-W_{t-j})=E_{t-\ell}-\sum_{j=1}^{q\wedge (t-1)}\gamma_j(1+E_{t-j})\frac{\partial W_{t-j}}{\partial \gamma_\ell},
\end{align}
where we used that  $E_t=0,\; \forall t\leq 0$.

The first derivatives of $W_t$ are thus obtained from the following recursive expressions. For all $k\in\{0,\dots,p\}$ 
\begin{align*}
\frac{\partial W_1}{\partial \beta_k}&=x_{1,k},\\
\frac{\partial W_2}{\partial \beta_k}&=x_{2,k}-\gamma_1(1+E_{1})\frac{\partial W_{1}}{\partial \beta_k},
\end{align*}
where
\begin{equation}\label{eq:E1}
W_1=\boldsymbol{\beta}' x_1 \textrm{ and } E_1=Y_1\exp(-W_1)-1.
\end{equation}
Moreover,
\begin{equation*}
\frac{\partial W_3}{\partial \beta_k}=x_{3,k}-\gamma_1(1+E_{2})\frac{\partial W_{2}}{\partial \beta_k}-\gamma_2(1+E_{1})\frac{\partial W_{1}}{\partial \beta_k},
\end{equation*}
where
\begin{equation}\label{eq:E2}
W_2=\boldsymbol{\beta}' x_2  +\gamma_1 E_{1},\; E_2=Y_2\exp(-W_2)-1,
\end{equation}
and so on. In the same way, for all $\ell\in\{1,\dots,q\}$
\begin{align*}
\frac{\partial W_1}{\partial \gamma_\ell}&=0,\\
\frac{\partial W_2}{\partial \gamma_\ell}&=E_{2-\ell},\\
\frac{\partial W_3}{\partial \gamma_\ell}&=E_{3-\ell}-\gamma_1(1+E_{2})\frac{\partial W_{2}}{\partial \gamma_\ell}
\end{align*}
and so on, where $E_t=0,\; \forall t\leq 0$ and $E_1$, $E_2$ are defined in (\ref{eq:E1}) and (\ref{eq:E2}), respectively.

\subsubsection{\textcolor{black}{Computation of the second derivatives of $W_t$}}\label{subsub:second_derive}

Using (\ref{eq:gradW_beta}) and (\ref{eq:gradW_gamma}), we get that for all $j,k\in\{0,\dots,p\}$, $\ell,m\in\{1,\dots,q\}$ and $t\in\{1,\dots,n\}$,
\begin{align*}
\frac{\partial^2 W_t}{\partial \beta_j\partial \beta_k}&=-\sum_{i=1}^{q\wedge (t-1)}\gamma_i(1+E_{t-i})\frac{\partial^2 W_{t-i}}{\partial \beta_j\partial \beta_k}
-\sum_{i=1}^{q\wedge (t-1)}\gamma_i\frac{\partial E_{t-i}}{\partial\beta_j}\frac{\partial W_{t-i}}{\partial \beta_k}\\
&=-\sum_{i=1}^{q\wedge (t-1)}\gamma_i(1+E_{t-i})\frac{\partial^2 W_{t-i}}{\partial \beta_j\partial \beta_k}
+\sum_{i=1}^{q\wedge (t-1)}\gamma_i(1+E_{t-i})\frac{\partial W_{t-i}}{\partial \beta_j}\frac{\partial W_{t-i}}{\partial \beta_k},\\
\frac{\partial^2 W_t}{\partial \beta_k\partial\gamma_\ell}&=-(1+E_{t-\ell})\frac{\partial W_{t-\ell}}{\partial \beta_k}
-\sum_{i=1}^{q\wedge (t-1)}\gamma_i\left\{\frac{\partial W_{t-i}}{\partial \beta_k}\frac{\partial E_{t-i}}{\partial\gamma_\ell}
                                                            +(1+E_{t-i})\frac{\partial^2 W_{t-i}}{\partial \beta_k\partial\gamma_\ell}\right\}\\
&=-(1+E_{t-\ell})\frac{\partial W_{t-\ell}}{\partial \beta_k}
-\sum_{i=1}^{q\wedge (t-1)}\gamma_i\left\{-(1+E_{t-i})\frac{\partial W_{t-i}}{\partial\beta_k}\frac{\partial W_{t-i}}{\partial \gamma_\ell}
                                                            +(1+E_{t-i})\frac{\partial^2 W_{t-i}}{\partial \beta_k\partial\gamma_\ell}\right\},\\
\frac{\partial^2 W_t}{\partial \gamma_\ell\partial\gamma_m}&=\frac{\partial E_{t-\ell}}{\partial \gamma_m}
-(1+E_{t-m})\frac{\partial W_{t-m}}{\partial \gamma_\ell} 
-\sum_{i=1}^{q\wedge (t-1)}\gamma_i\left\{\frac{\partial W_{t-i}}{\partial \gamma_\ell} \frac{\partial E_{t-i}}{\partial \gamma_m}
+(1+E_{t-i})\frac{\partial^2 W_{t-i}}{\partial \gamma_\ell\partial \gamma_m}\right\}\\
&=-(1+E_{t-\ell})\frac{\partial W_{t-\ell}}{\partial \gamma_m}-(1+E_{t-m})\frac{\partial W_{t-m}}{\partial \gamma_\ell} \\
&-\sum_{i=1}^{q\wedge (t-1)}\gamma_i\left\{-(1+E_{t-i})\frac{\partial W_{t-i}}{\partial \gamma_\ell}\frac{\partial W_{t-i}}{\partial \gamma_m}
+(1+E_{t-i})\frac{\partial^2 W_{t-i}}{\partial \gamma_\ell\partial \gamma_m}\right\}.\\
\end{align*}

To compute the second derivatives of $W_t$, we shall use the following recursive expressions for all $j,k\in\{0,\dots,p\}$
\begin{align*}
\frac{\partial^2 W_1}{\partial \beta_j\partial \beta_k}&=0,\\
\frac{\partial^2 W_2}{\partial \beta_j\partial \beta_k}&=\gamma_1(1+E_1)x_{1,j}x_{1,k},
\end{align*}
where $E_1$ is defined in (\ref{eq:E1}) and so on. Moreover, for all $k\in\{0,\dots,p\}$ and $\ell\in\{1,\dots,q\}$
\begin{align*}
\frac{\partial^2 W_1}{\partial \beta_k\partial\gamma_\ell}&=0,\\
\frac{\partial^2 W_2}{\partial \beta_k\partial\gamma_\ell}&=-(1+E_{2-\ell})\frac{\partial W_{2-\ell}}{\partial \beta_k},
\end{align*}
where $E_t=0$ for all $t\leq 0$ and the first derivatives of $W_t$ are computed in (\ref{eq:gradW_beta}).
Note also that
\begin{align*}
\frac{\partial^2 W_1}{\partial \gamma_\ell\partial\gamma_m}&=0,\\
\frac{\partial^2 W_2}{\partial \gamma_\ell\partial\gamma_m}&=0
\end{align*}
and so on.

\subsection{Computational details for obtaining Criterion (\ref{eq:beta_hat})}\label{sub:var_sec}

By \eqref{eq:Ltilde},
\begin{align*}
\widetilde{L}(\boldsymbol{\beta})=\widetilde{L}(\boldsymbol{\beta}^{(0)})+\frac{\partial L}{\partial \boldsymbol{\beta}}(\boldsymbol{\beta}^{(0)},\widehat{\boldsymbol{\gamma}})
U(\boldsymbol{\nu}-\boldsymbol{\nu}^{(0)})-\frac12 (\boldsymbol{\nu}-\boldsymbol{\nu}^{(0)})' \Lambda (\boldsymbol{\nu}-\boldsymbol{\nu}^{(0)}),
\end{align*}
where $\boldsymbol{\nu}-\boldsymbol{\nu}^{(0)}=U'(\boldsymbol{\beta}-\boldsymbol{\beta}^{(0)})$.
Hence,
\begin{align*}
\widetilde{L}(\boldsymbol{\beta})&=\widetilde{L}(\boldsymbol{\beta}^{(0)})+\sum_{k=0}^p
\left(\frac{\partial L}{\partial \boldsymbol{\beta}}(\boldsymbol{\beta}^{(0)},\widehat{\boldsymbol{\gamma}}) U\right)_k (\nu_k-\nu_{k}^{(0)})
-\frac12\sum_{k=0}^p\lambda_k (\nu_k-\nu_{k}^{(0)})^2\\
&=\widetilde{L}(\boldsymbol{\beta}^{(0)})-\frac12\sum_{k=0}^p\lambda_k\left(\nu_k-\nu_{k}^{(0)}-\frac{1}{\lambda_k}
\left(\frac{\partial L}{\partial \boldsymbol{\beta}}(\boldsymbol{\beta}^{(0)},\widehat{\boldsymbol{\gamma}}) U\right)_k\right)^2
+\sum_{k=0}^p\frac{1}{2\lambda_k}\left(\frac{\partial L}{\partial \boldsymbol{\beta}}(\boldsymbol{\beta}^{(0)},\widehat{\boldsymbol{\gamma}}) U\right)_k^2,
\end{align*}
where the $\lambda_k$'s are the diagonal terms of $\Lambda$.

Since the only term depending on $\boldsymbol{\beta}$ is the second one in the last expression of $\widetilde{L}(\boldsymbol{\beta})$,
we define $\widetilde{L}_Q(\boldsymbol{\beta})$ appearing in Criterion (\ref{eq:beta_hat}) as follows:
\begin{eqnarray*}
-\widetilde{L}_Q(\boldsymbol{\beta})&=&\frac12\sum_{k=0}^p\lambda_k\left(\nu_k-\nu_{k}^{(0)}-\frac{1}{\lambda_k}
\left(\frac{\partial L}{\partial \boldsymbol{\beta}}(\boldsymbol{\beta}^{(0)},\widehat{\boldsymbol{\gamma}}) U\right)_k\right)^2\\
&=&\frac12 \left\|\Lambda^{1/2}\left(\boldsymbol{\nu}-\boldsymbol{\nu}^{(0)}-\Lambda^{-1} \left(\frac{\partial L}{\partial \boldsymbol{\beta}}(\boldsymbol{\beta}^{(0)},\widehat{\boldsymbol{\gamma}}) U\right)'
\right)\right\|_2^2\\
&=&\frac12 \left\|\Lambda^{1/2}U'(\boldsymbol{\beta}-\boldsymbol{\beta}^{(0)})-\Lambda^{-1/2} U' \left(\frac{\partial L}{\partial \boldsymbol{\beta}}(\boldsymbol{\beta}^{(0)},\widehat{\boldsymbol{\gamma}})\right)'
\right\|_2^2\\
&=&\frac12 \left\|\Lambda^{1/2}U'(\boldsymbol{\beta}^{(0)}-\boldsymbol{\beta})+\Lambda^{-1/2} U' \left(\frac{\partial L}{\partial \boldsymbol{\beta}}(\boldsymbol{\beta}^{(0)},\widehat{\boldsymbol{\gamma}})\right)'\right\|_2^2\\
&=&\frac12\|\mathcal{Y}-\mathcal{X}\boldsymbol{\beta}\|_2^2,
\end{eqnarray*}
where
\begin{equation*}
\mathcal{Y}=\Lambda^{1/2}U'\boldsymbol{\beta}^{(0)}
+\Lambda^{-1/2}U'\left(\frac{\partial L}{\partial \boldsymbol{\beta}}(\boldsymbol{\beta}^{(0)},\widehat{\boldsymbol{\gamma}})\right)' ,\;  \mathcal{X}=\Lambda^{1/2}U'.
\end{equation*}

\subsection{Proofs of Propositions \ref{prop1}, \ref{prop2} and \ref{prop3} and of Lemma \ref{lem:aperiodic_doeblin}}

This section contains the proofs of Propositions \ref{prop1}, \ref{prop2} and \ref{prop3}.

\subsubsection{\textcolor{black}{Proof of Proposition \ref{prop1}}}

\textcolor{black}{We first establish the following lemma for proving Proposition \ref{prop1}.}

\begin{lemma}\label{lem:aperiodic_doeblin}
$(W_t^\star)$ is an aperiodic Markov process satisfying Doeblin's condition.
\end{lemma}

\begin{proof}[Proof of Lemma \ref{lem:aperiodic_doeblin}]
By (\ref{eq:mut_simple}) and (\ref{eq:Zt}), we observe that:
\begin{equation}\label{eq:Wtstar}
W_t^\star=(\beta_0^\star-\gamma_1^\star)+\gamma_1^\star Y_{t-1}\exp(-W_{t-1}^\star).
\end{equation}
Thus, $\mathcal{F}_{t-2}=\mathcal{F}_{t-1}^{W^\star}:=\sigma(W_s,s\leq t-1)$.
By (\ref{eq:Yt}), the distribution of $Y_{t-1}$ conditionally to $\mathcal{F}_{t-2}$ is $\mathcal{P}(\exp(W_{t-1}^\star))$.
Hence, the distribution of $W_t^\star$ conditionally to $\mathcal{F}_{t-1}^{W^\star}$
is the same as distribution of $W_t^\star$ conditionally to $W_{t-1}^\star$, which means that $(W_t^\star)$ has the Markov property.

Let us now prove that $(W_t^\star)$ is strongly aperiodic which implies that it is aperiodic.
$$
\mathbb{P}(W^\star_{t}=\beta_0^\star-\gamma_1^\star | W^\star_{t-1}=\beta_0^\star-\gamma_1^\star)
=\mathbb{P}(Y_{t-1}=0 | W^\star_{t-1}=\beta_0^\star-\gamma_1^\star)=\exp(-\exp(\beta_0^\star-\gamma_1^\star))>0,
$$
where the first equality comes from (\ref{eq:Wtstar}) and the last equality comes from (\ref{eq:Yt}) since $\mathcal{F}_{t-2}=\mathcal{F}_{t-1}^{W^\star}$.

To prove that $(W_t^\star)$ satisfies Doeblin's condition namely that there exists a probability measure $\nu$ with the property that, for some $m\geq 1$, $\varepsilon>0$
and $\delta >0$,
\begin{equation}\label{eq:doeblin}
\nu(B)>\varepsilon \Longrightarrow \mathbb{P}(W_{t+m-1}\in B,W_{t+m-2}\in B\dots,W_{t+1}\in B,W_{t}\in B|W_{t-1}=x)\geq\delta,
\end{equation}
for all $x$ in the state space $X$ of $W_t^\star$ and $B$ in the Borel sets of $X$, we refer the reader to the proof of Proposition 2 in \cite{davis:dunsmuir:streett:2003}.

\end{proof}

\begin{proof}[Proof of Proposition \ref{prop1}]
For proving Proposition \ref{prop1}, we shall use Theorems 1.3.3 and 1.3.5 of
\cite{taniguchibook:2012}. In order to apply these theorems it is enough to prove that $(W_t^\star)$ is a strictly stationary and ergodic process
since $Y_t W_t(\beta_0^\star,\gamma_1)-\exp(W_t(\beta_0^\star,\gamma_1))$ is a measurable function of $W_{t+1}^\star,W_t^\star,\dots,W_2^\star$. 
Note that the latter fact comes from (\ref{eq:mut_simple}) and (\ref{eq:Zt}) for $Y_t$ and from (\ref{eq:Wt}) with $q=1$ and $p=0$ for $W_t$.

In order to prove that $(W_t^\star)$ is a strictly stationary and ergodic process, we have first to prove that $(W_t^\star)$ is an aperiodic Markov process satisfying Doeblin's condition, 
see Lemma \ref{lem:aperiodic_doeblin}.
cv
The statement of Lemma \ref{lem:aperiodic_doeblin} corresponds to Assertion (iv) of Theorem 16.0.2 of \cite{meyn:tweedie}  which is equivalent to Assertion (i) of this theorem, 
and implies that $(W_t^\star)$ is uniformly ergodic.

Hence, by Definition (16.6) of uniform ergodicity given in \cite{meyn:tweedie}, there exists a unique stationary invariant measure for $(W_t^\star)$, see also the paragraph below Equation (1.3) 
of \cite{Sandric:2017} for an additional justification. Combining that existence of a unique stationary invariant measure for $(W_t^\star)$ with the following arguments shows that $(W_t^\star)$ 
is a strictly stationary process and also an ergodic Markov process.

By Theorem 3.6.3, Corollary 3.6.1 and Definition 3.6.6 of \cite{stout:1974}, if the process $(W_t^\star)$ is started with its unique stationary invariant distribution, $(W_t^\star)$ is a strictly stationary process.

By Definition 3.6.8 of  \cite{stout:1974}, the existence of a unique stationary invariant measure for $(W_t^\star)$ 
means that $(W_t^\star)$ is an ergodic Markov process, see also the paragraph below (b) \cite[p. 717]{Sandric:2017}.

Finally, by Theorem 3.6.5 of  \cite{stout:1974}, since $(W_t^\star)$ is an ergodic Markov process and a strictly stationary process, 
$(W_t^\star)$ is an ergodic and strictly stationary process in the sense of the assumption of Theorem 1.3.5 of \cite{taniguchibook:2012}.
\end{proof}

\subsubsection{\textcolor{black}{Proof of Proposition \ref{prop2}}}

Note that for all $\gamma_1$,
\begin{align*}
\mathcal{L}(\gamma_1)&=\mathbb{E}\left[Y_3 W_3(\beta_0^\star,\gamma_1)-\exp(W_3(\beta_0^\star,\gamma_1))\right]
=\mathbb{E}\left[\mathbb{E}\left[Y_3 W_3(\beta_0^\star,\gamma_1)-\exp(W_3(\beta_0^\star,\gamma_1))|\mathcal{F}_2\right]\right]\\
&=\mathbb{E}\left[\exp(W_3^\star) W_3(\beta_0^\star,\gamma_1)-\exp(W_3(\beta_0^\star,\gamma_1))\right]\\
&=\mathbb{E}\left[\exp(W_3^\star) \left(W_3(\beta_0^\star,\gamma_1)-W_3^\star+W_3^\star-\exp(W_3(\beta_0^\star,\gamma_1)-W_3^\star)\right)\right]\\
&\leq \mathbb{E}\left[\exp(W_3^\star) \left(W_3^\star-1\right)\right]=\mathcal{L}(\gamma_1^\star),
\end{align*}
where the inequality comes from the following inequality $x-\exp(x)\leq -1$, for all $x\in\mathbb{R}$. This inequality is an equality only when 
$x=0$ which means that $\gamma_1=\gamma_1^\star$.

\subsubsection{\textcolor{black}{Proof of Proposition \ref{prop3}}}

The proof of this proposition comes from Proposition \ref{prop1} and the stochastic equicontinuity of $n^{-1}L(\beta_0^\star,\gamma_1)$. Thus, it is enough to prove that
there exists a positive $\delta$ such that
$$
\sup_{|\gamma_1-\gamma_2|\leq\delta}\left|\frac{L(\beta_0^\star,\gamma_1)}{n}-\frac{L(\beta_0^\star,\gamma_2)}{n}\right|\stackrel{p}{\longrightarrow}0, \textrm{ as $n$ tecvnds to infinity.}
$$
Observe that, by (\ref{eq:L:beta_0}),
\begin{align*}
\left|\frac{L(\beta_0^\star,\gamma_1)}{n}-\frac{L(\beta_0^\star,\gamma_2)}{n}\right|
&\leq\frac1n\sum_{t=1}^n Y_t\left|W_t(\beta_0^\star,\gamma_1)-W_t(\beta_0^\star,\gamma_2)\right|\\
&+\frac1n\sum_{t=1}^n\left|\exp\left(W_t(\beta_0^\star,\gamma_1)\right)-\exp\left(W_t(\beta_0^\star,\gamma_2)\right)\right|.
\end{align*}
Let us first focus on bounding the following expression for $t\geq 2$ (since $W_1(\beta_0^\star,\gamma)=\beta_0^\star$, for all $\gamma$). By (\ref{eq:W_Z})
\begin{align*}
&\left|W_t(\beta_0^\star,\gamma_1)-W_t(\beta_0^\star,\gamma_2)\right|=\left|Z_t(\gamma_1)-Z_t(\gamma_2)\right|=\left|\gamma_1 E_{t-1}(\gamma_1)-\gamma_2 E_{t-1}(\gamma_2)\right|\\
&=\left|\gamma_1 \left[Y_{t-1}\exp(-W_{t-1}(\beta_0^\star,\gamma_1))-1\right]-\gamma_2 \left[Y_{t-1}\exp(-W_{t-1}(\beta_0^\star,\gamma_2))-1\right]\right|\\
&=\left|Y_{t-1}\textrm{e}^{-\beta_0^\star}\left[\gamma_1\exp(-Z_{t-1}(\gamma_1))-\gamma_2\exp(-Z_{t-1}(\gamma_2))\right]+\gamma_2-\gamma_1\right|\\
&\leq Y_{t-1}\textrm{e}^{-\beta_0^\star}\left[\left|\gamma_1-\gamma_2\right|\exp(-Z_{t-1}(\gamma_1))+|\gamma_2|\left|\exp(-Z_{t-1}(\gamma_1))-\exp(-Z_{t-1}(\gamma_2))\right|\right]
+\left|\gamma_2-\gamma_1\right|\\
&\leq Y_{t-1}\textrm{e}^{-\beta_0^\star}\left|\gamma_1-\gamma_2\right|\exp(-Z_{t-1}(\gamma_1))\\
&+Y_{t-1}\textrm{e}^{-\beta_0^\star}|\gamma_2|\exp(-Z_{t-1}(\gamma_1))\left|Z_{t-1}(\gamma_1)-Z_{t-1}(\gamma_2)\right|
\exp(|Z_{t-1}(\gamma_1)-Z_{t-1}(\gamma_2)|)\\
&+\left|\gamma_2-\gamma_1\right|,
\end{align*}
where we used in the last inequality that for all $x$ and $y$ in $\mathbb{R}$,
\begin{equation}\label{eq:exp_x_y}
|\textrm{e}^x-\textrm{e}^y|=\textrm{e}^x|1-\textrm{e}^{y-x}|\leq \textrm{e}^x |y-x|\textrm{e}^{|y-x|}.
\end{equation}
Observing that
\begin{equation}\label{eq:exp_Zt}
\exp(-Z_{t}(\gamma_1))=\exp\left(-\gamma_1\left[Y_{t-1}\textrm{e}^{-\beta_0^\star}\exp(-Z_{t-1}(\gamma_1))-1\right]\right),
\end{equation}
and $|Z_{2}(\gamma_1)-Z_{2}(\gamma_2)|\leq\delta[Y_1\textrm{e}^{-\beta_0^\star}+1]$ we get, for $\gamma_1$ and $\gamma_2$ such that $|\gamma_1-\gamma_2|\leq\delta$, that
\begin{equation}\label{eq:diff_W}
\left|W_t(\beta_0^\star,\gamma_1)-W_t(\beta_0^\star,\gamma_2)\right|
\leq\delta\; F(Y_{t-1},Y_{t-2},\dots,Y_1),
\end{equation}
where $F$ is a measurable function.
By (\ref{eq:exp_x_y}),
\begin{align*}
&\left|\exp\left(W_t(\beta_0^\star,\gamma_1)\right)-\exp\left(W_t(\beta_0^\star,\gamma_2)\right)\right|\\
&\leq \exp\left(W_t(\beta_0^\star,\gamma_1)\right)\left|W_t(\beta_0^\star,\gamma_1)-W_t(\beta_0^\star,\gamma_2)\right|\exp\left(\left|W_t(\beta_0^\star,\gamma_1)-W_t(\beta_0^\star,\gamma_2)\right|\right)\\
&\leq \delta G(Y_{t-1},Y_{t-2},\dots,Y_1)
\end{align*}
where the last inequality comes from (\ref{eq:diff_W}), (\ref{eq:exp_Zt}) and (\ref{eq:W_Z}) and where $G$ is a measurable function.
Thus, we get that 
$$
\left|\frac{L(\beta_0^\star,\gamma_1)}{n}-\frac{L(\beta_0^\star,\gamma_2)}{n}\right|\leq\frac{\delta}{n}\sum_{t=1}^n H(Y_t,Y_{t-1},\dots,Y_1),
$$
which gives the result by using similar arguments as those given in the proof of Proposition \ref{prop1} namely that $(Y_t)$ is strictly stationary and ergodic.
By Theorem 1.3.3 of \cite{taniguchibook:2012}, $H(Y_t,Y_{t-1},\dots,Y_1)$ is strictly stationary and ergodic since $(Y_t)$ has these properties. Thus, $\mathbb{E}[|H(Y_t,Y_{t-1},\dots,Y_1)|]<\infty$, which
concludes the proof by Theorem 1.3.5 of \cite{taniguchibook:2012}.



\clearpage

\section*{Appendix}

This appendix contains additional results for the support recovery of $\boldsymbol{\beta}^\star$ and for the estimation of $\boldsymbol{\gamma}^\star$ discussed in
Section \ref{sec:sparse_estim}.


\begin{figure}[!h]
  \includegraphics[scale=0.28]{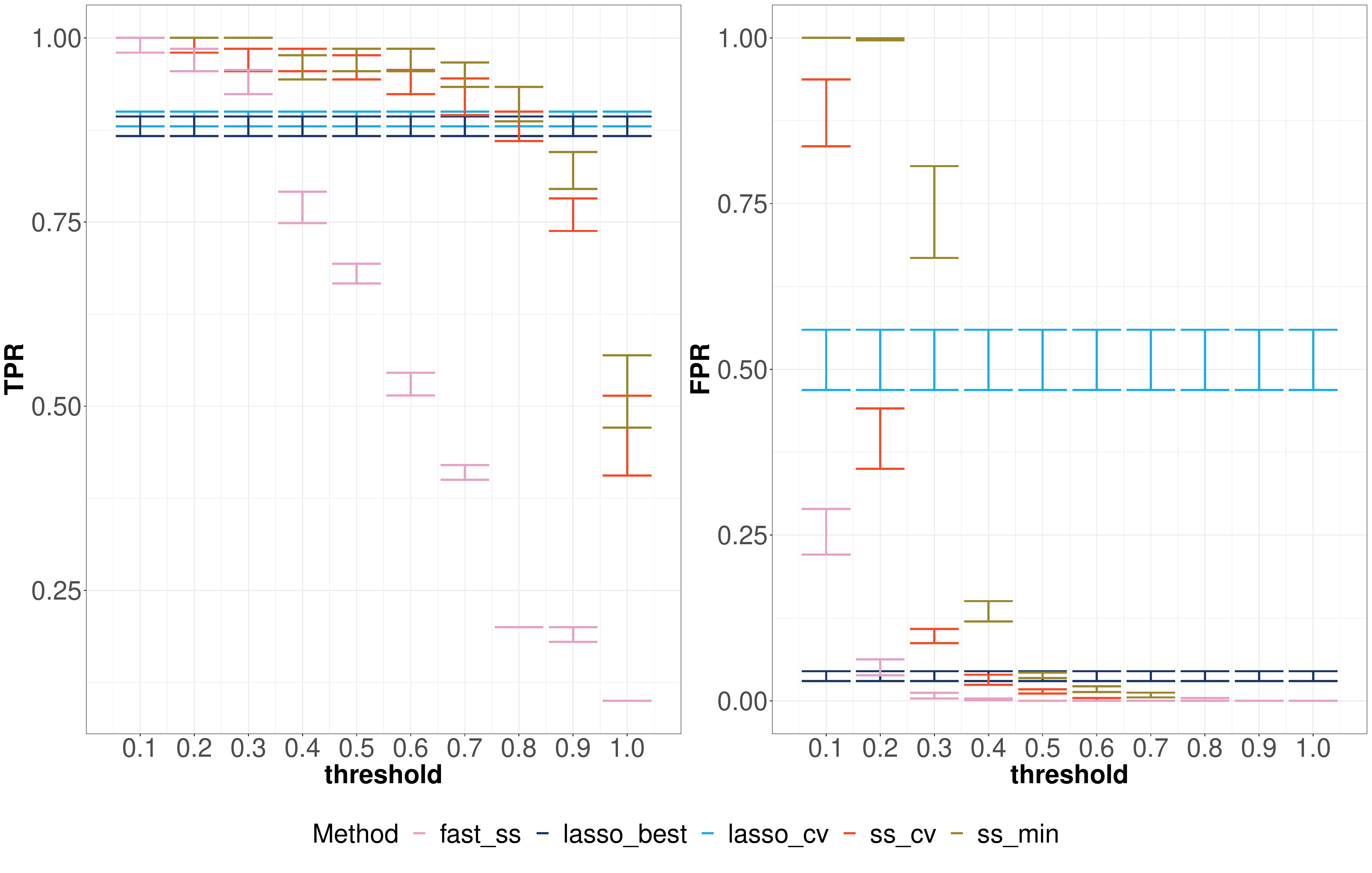}
  \caption{Error bars of the TPR and FPR associated to the support recovery of $\boldsymbol{\beta}^\star$ for five methods with respect to the thresholds when $n=1000$, $q=1$, $p=100$ and a 10\% sparsity level. All the $\beta_i^\star=0$ except for ten of them: $\beta_1^\star=1.73$, $\beta_3^\star=1.2$, $\beta_5^\star=0.67$, $\beta_{10}^\star=0.5$, $\beta_{14}^\star=-0.38$, $\beta_{17}^\star=0.29$,
$\beta_{30}^\star=-0.64$, $\beta_{33}^\star=-0.13$, $\beta_{38}^\star=-0.1$ and $\beta_{44}^\star=-0.07$. \label{fig:TPR:FPR:1:10}}
\end{figure}

\begin{figure}[!h]
  \includegraphics[scale=0.28]{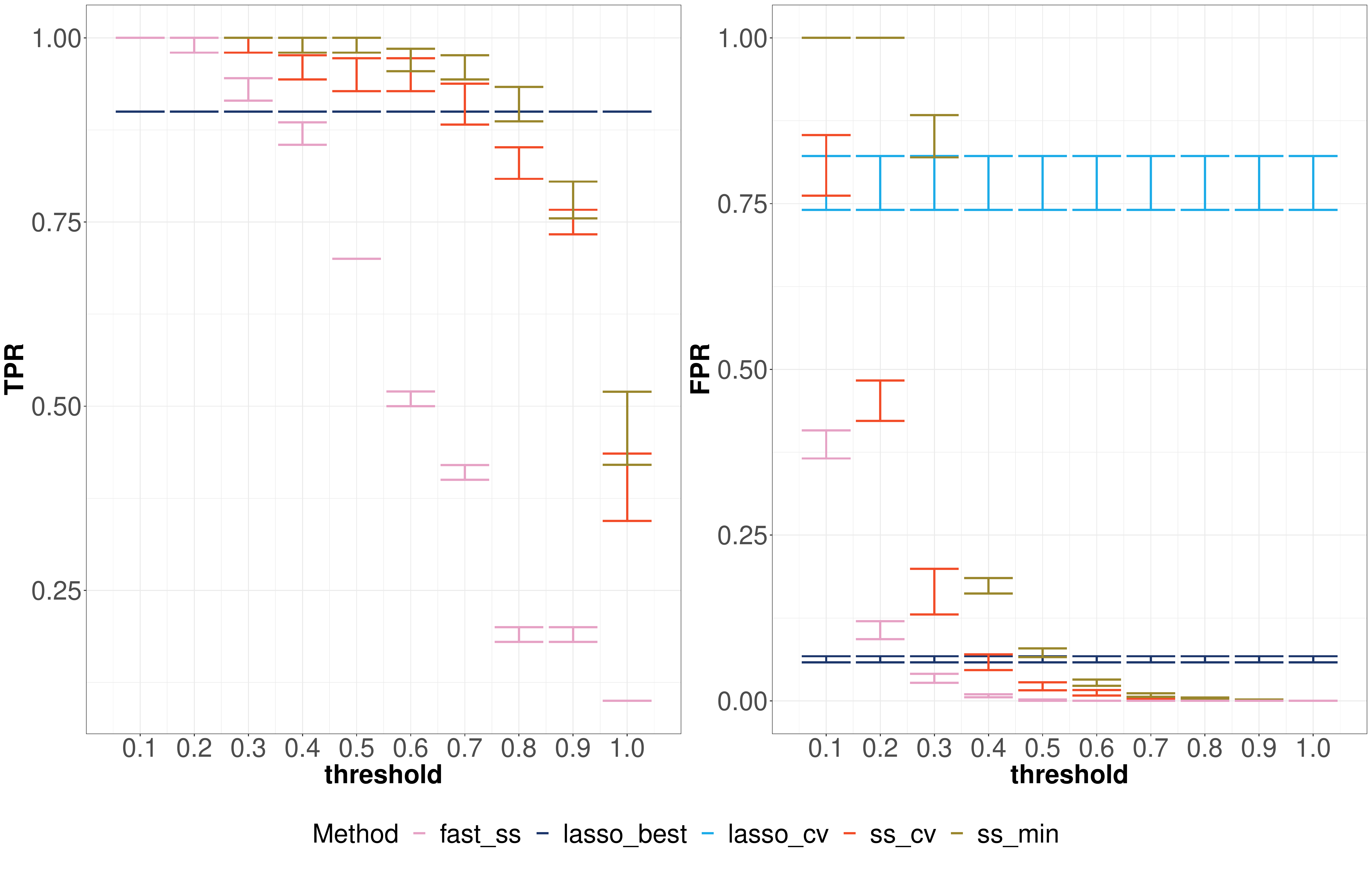}
  \caption{Error bars of the TPR and FPR associated to the support recovery of $\boldsymbol{\beta}^\star$ for five methods with respect to the thresholds when $n=1000$, $q=2$, $p=100$ and a 10\% sparsity level. All the $\beta_i^\star=0$ except for ten of them: $\beta_1^\star=1.73$, $\beta_3^\star=1.2$, $\beta_5^\star=0.67$, $\beta_{10}^\star=0.5$, $\beta_{14}^\star=-0.38$, $\beta_{17}^\star=0.29$,
$\beta_{30}^\star=-0.64$, $\beta_{33}^\star=-0.13$, $\beta_{38}^\star=-0.1$ and $\beta_{44}^\star=-0.07$.\label{fig:TPR:FPR:2:10}}
\end{figure}

\begin{figure}[!h]
 \includegraphics[scale=0.28]{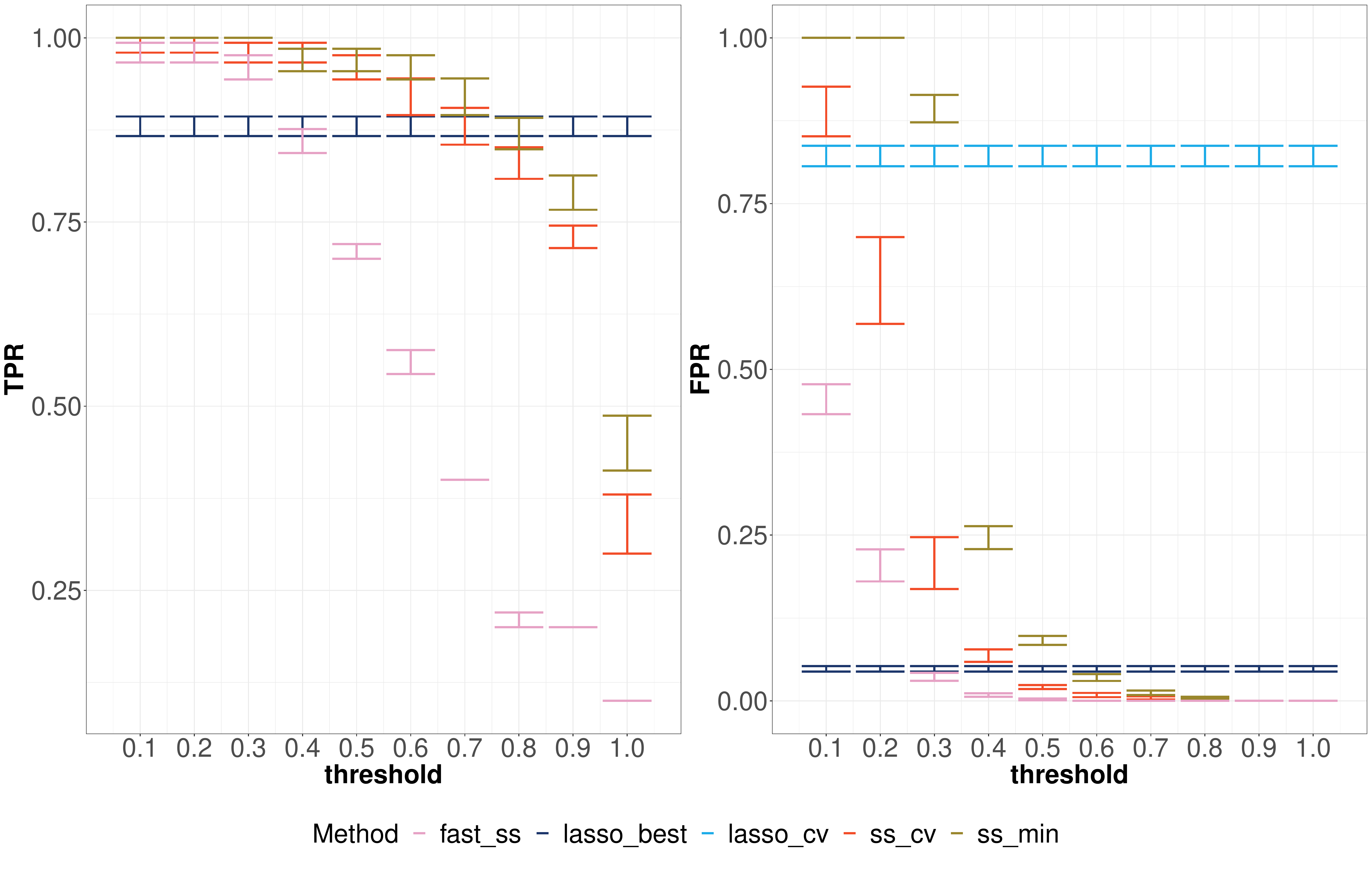}
  \caption{Error bars of the TPR and FPR giving the corresponding final sets of selected variables for five methods with respect to the thresholds when $n=1000$, $q=3$, $p=100$ and a 10\% sparsity level. All the $\beta_i^\star=0$ except for ten of them: $\beta_1^\star=1.73$, $\beta_3^\star=1.2$, $\beta_5^\star=0.67$, $\beta_{10}^\star=0.5$, $\beta_{14}^\star=-0.38$, $\beta_{17}^\star=0.29$,
$\beta_{30}^\star=-0.64$, $\beta_{33}^\star=-0.13$, $\beta_{38}^\star=-0.1$ and $\beta_{44}^\star=-0.07$.\label{fig:TPR:FPR:3:10}}
\end{figure}


\begin{figure}[!h]
  \begin{center}
\begin{tabular}{ccc}
  \includegraphics[width=0.32\textwidth, height=4.5cm]{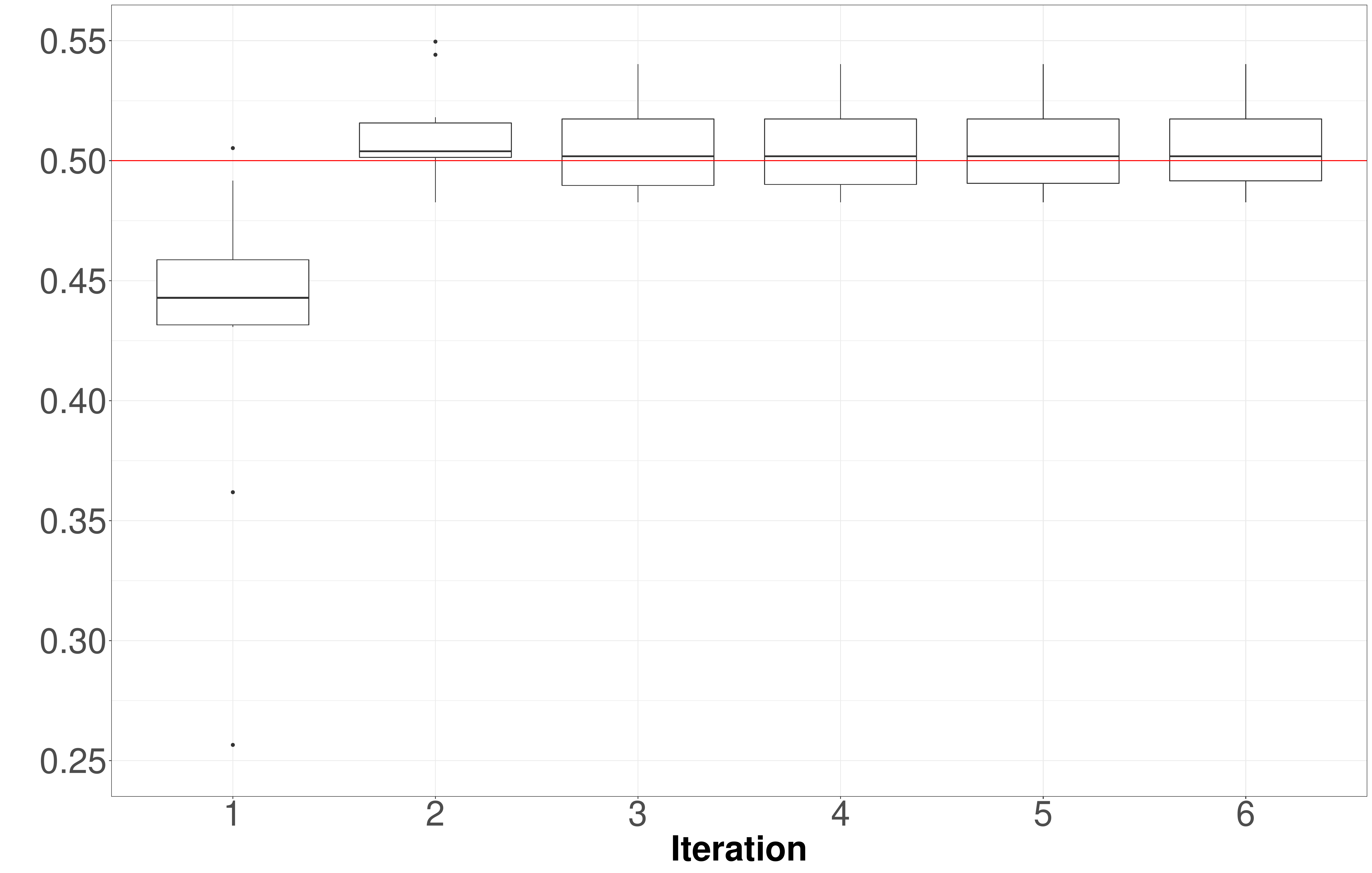}
  &  \includegraphics[width=0.32\textwidth, height=4.5cm]{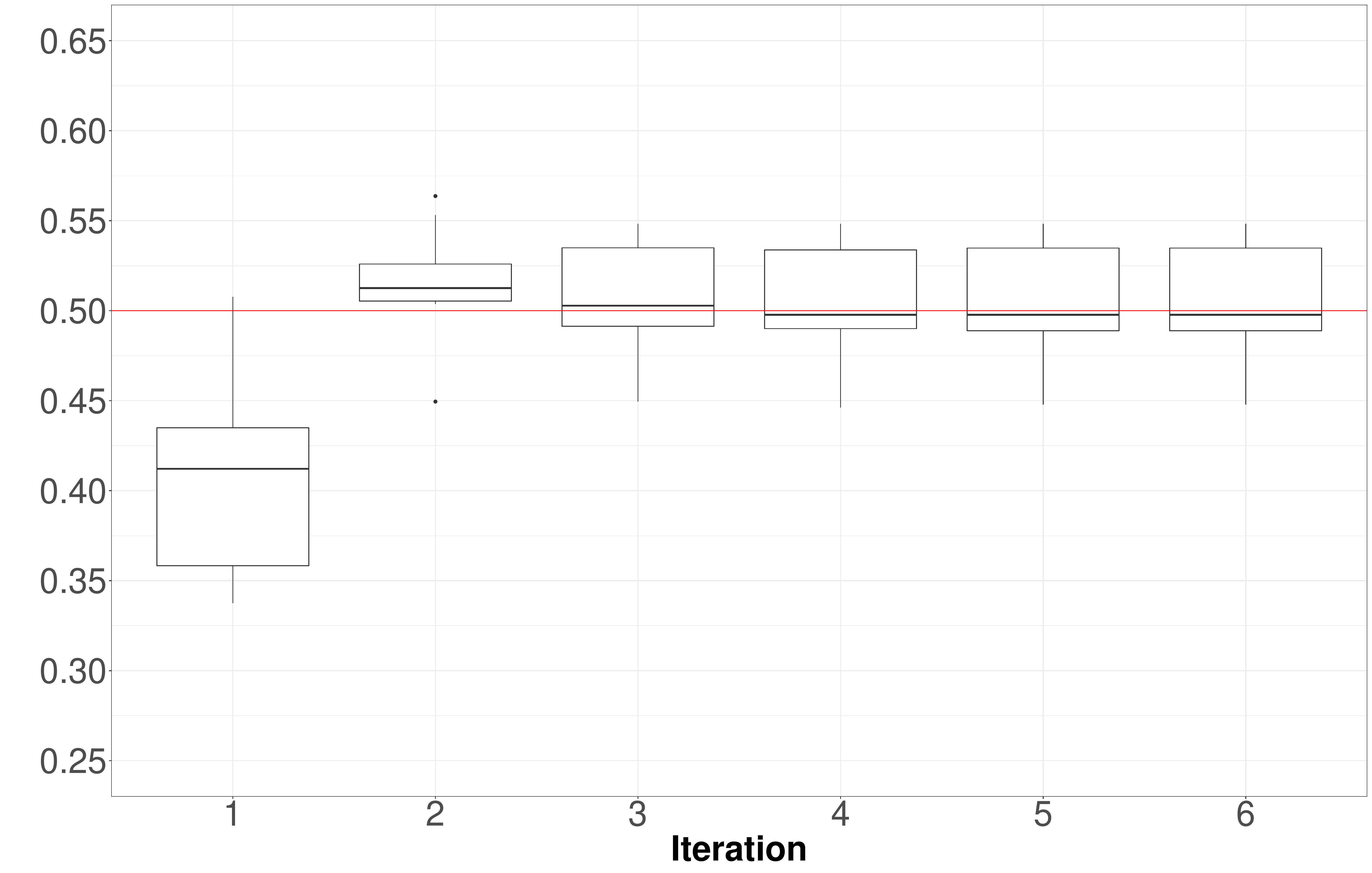} 
  & \includegraphics[width=0.32\textwidth, height=4.5cm]{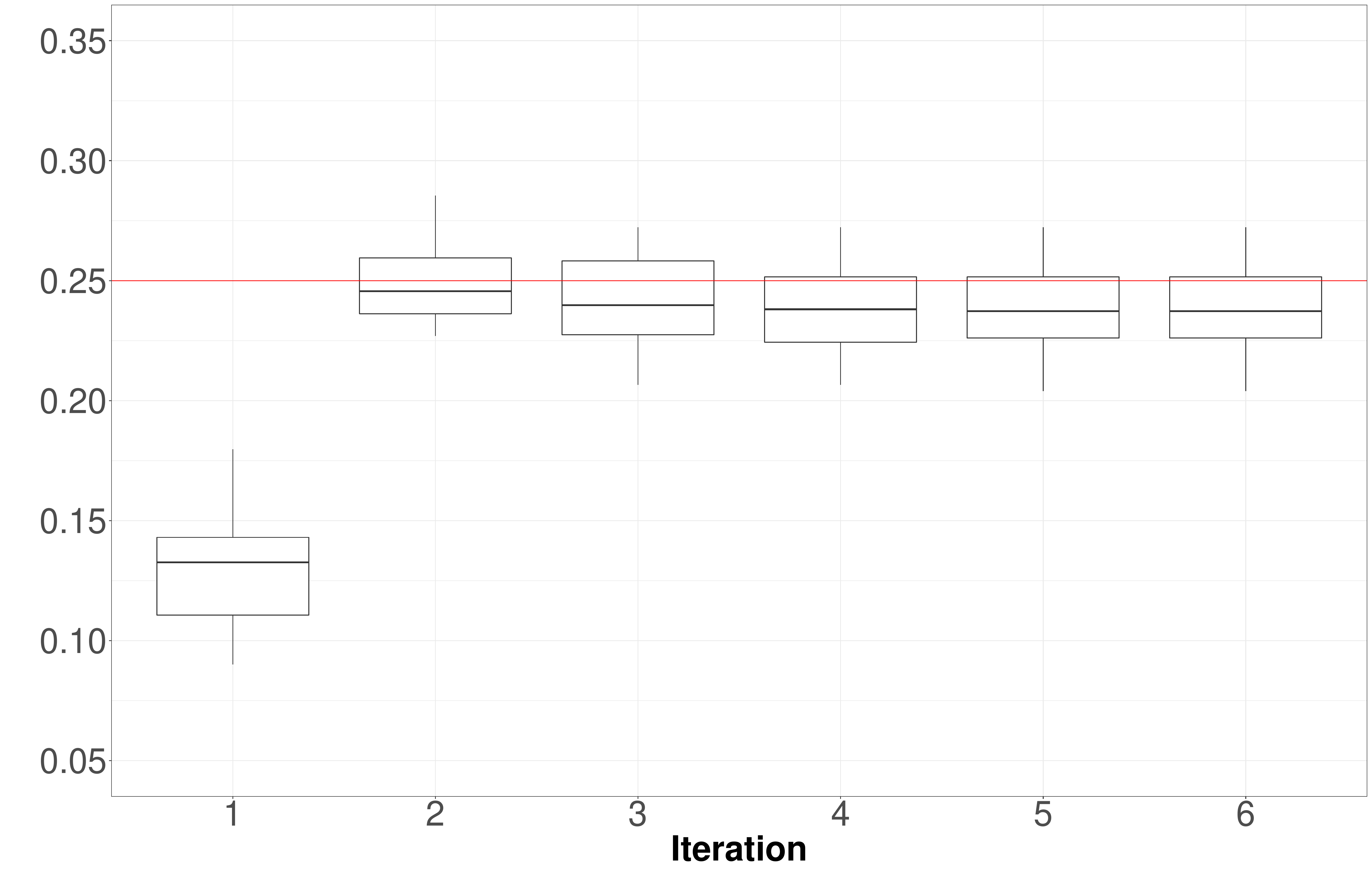}\\
\includegraphics[width=0.32\textwidth, height=4.5cm]{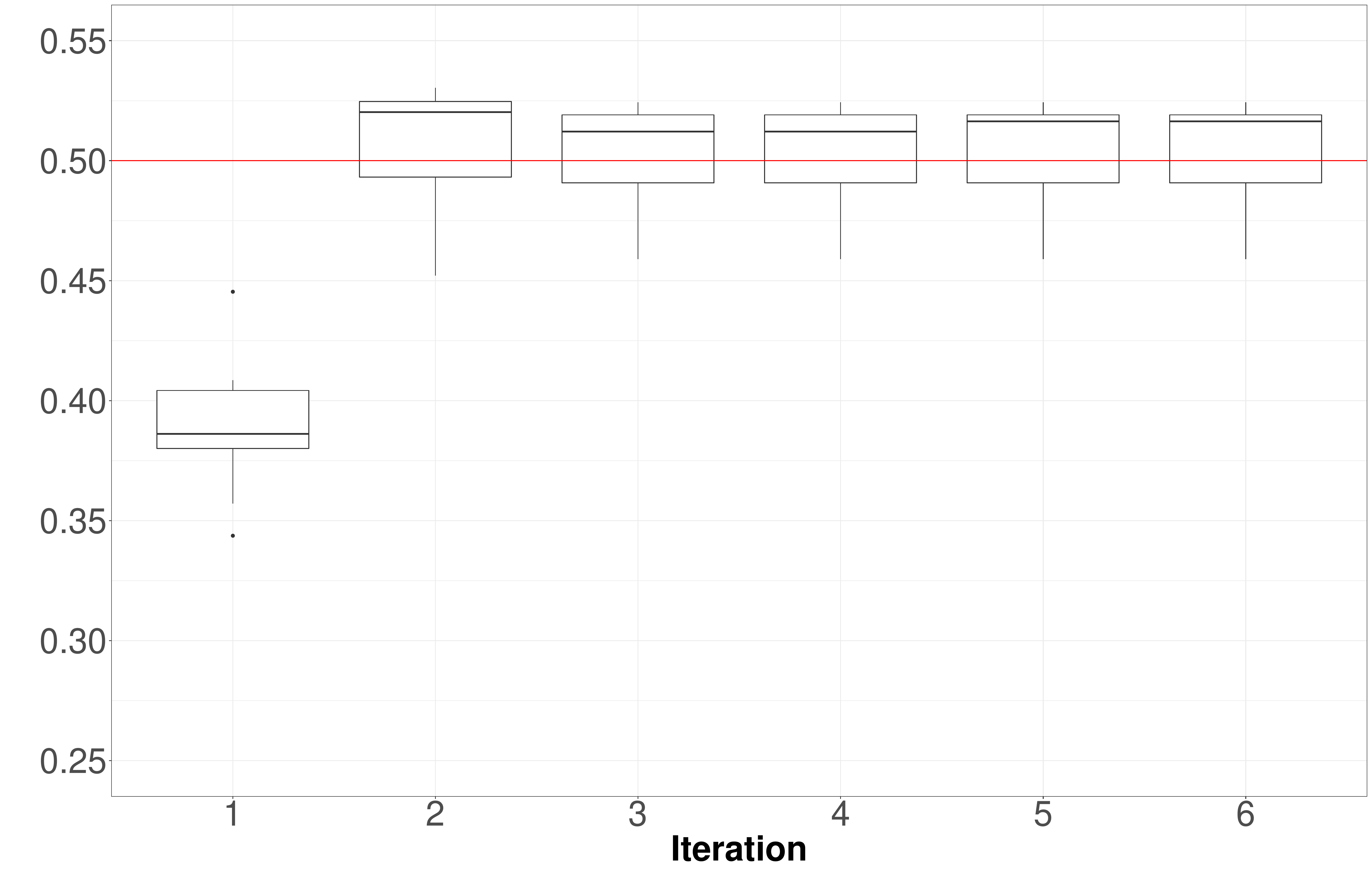}
  &  \includegraphics[width=0.32\textwidth, height=4.5cm]{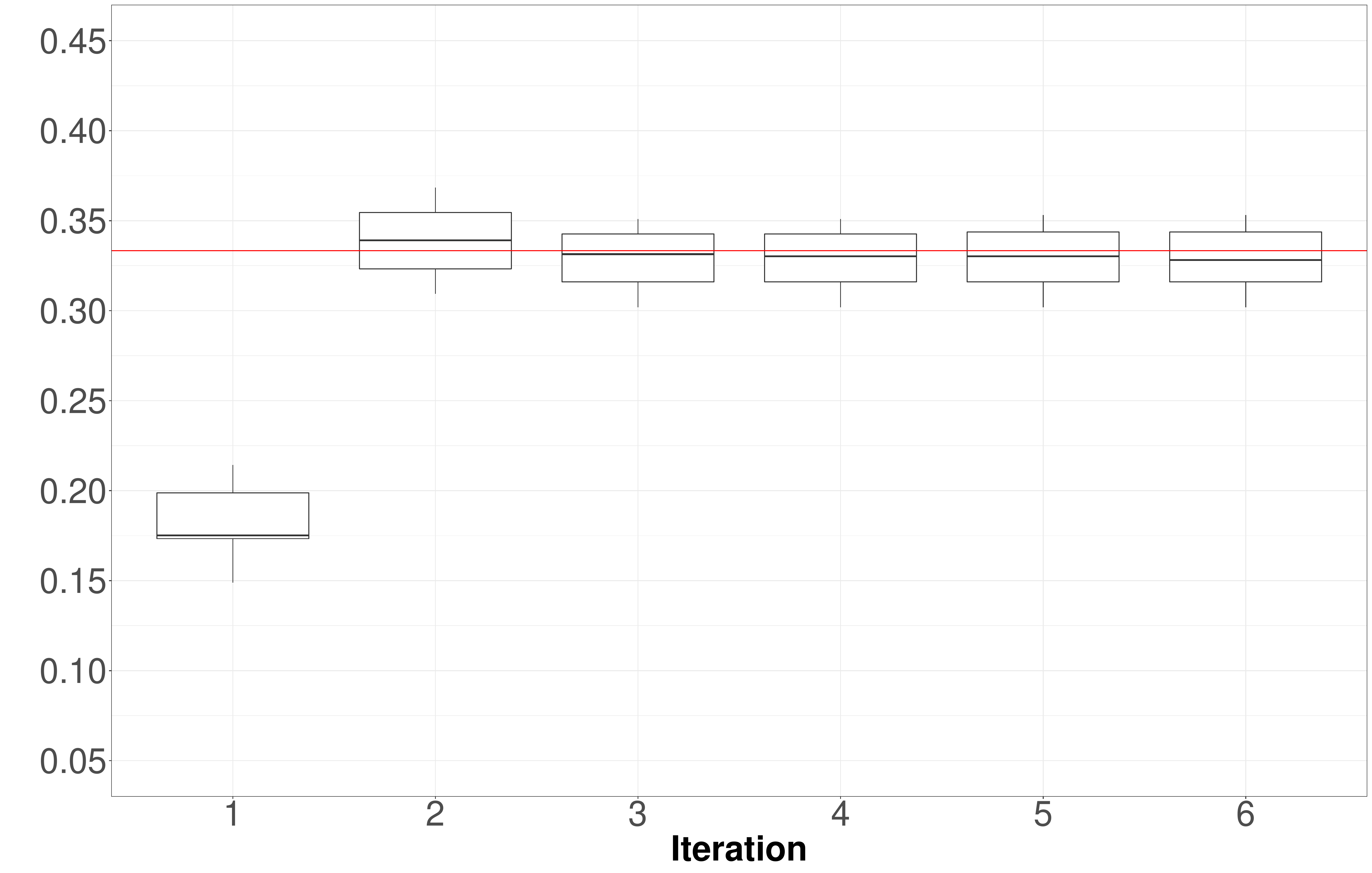} 
  & \includegraphics[width=0.32\textwidth, height=4.5cm]{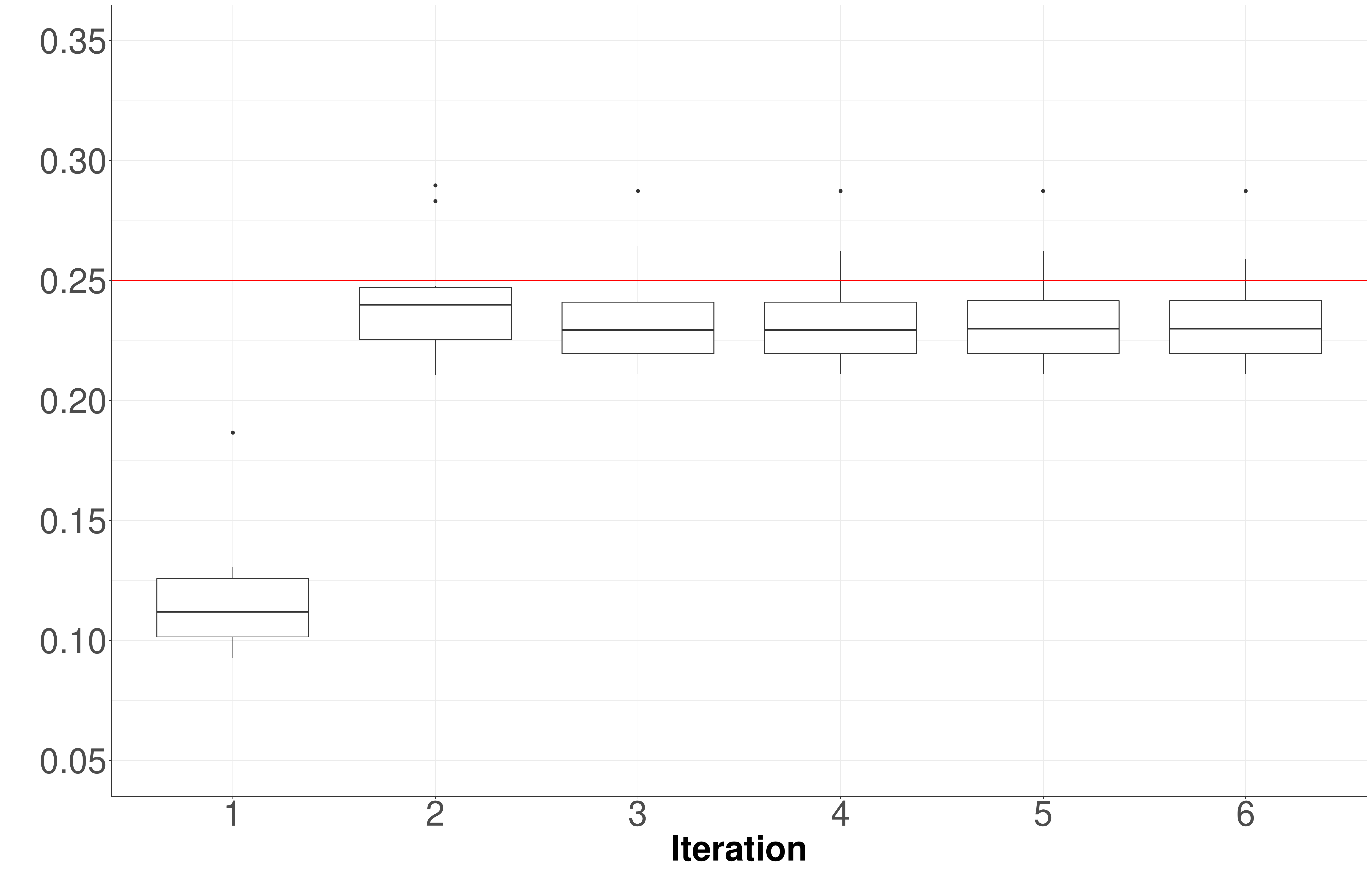}\\
\end{tabular}  
\caption{Boxplots for the estimations of $\boldsymbol{\gamma}^\star$ in Model (\ref{eq:mut_Wt}) with a 10\% sparsity level and $q=1,2,3$ obtained by
  \texttt{ss\_cv}.
Top: $q=1$ and $\gamma_1^\star=0.5$ (left), $q=2$ and $\gamma_1^\star=0.5$ (middle), $q=2$ and $\gamma_2^\star=0.25$ (right). Bottom: $q=3$ and $\gamma_1^\star=0.5$ (left), $q=3$ and  $\gamma_2^\star=1/3$ (middle), $q=3$ and $\gamma_3^\star=0.25$ (right).
The horizontal lines correspond to the values of the $\gamma_i^\star$'s. \label{fig:gamma:10:cv}}
 \end{center}
\end{figure}

\begin{figure}[!h]
  \begin{center}
\begin{tabular}{ccc}
  \includegraphics[width=0.32\textwidth, height=4.5cm]{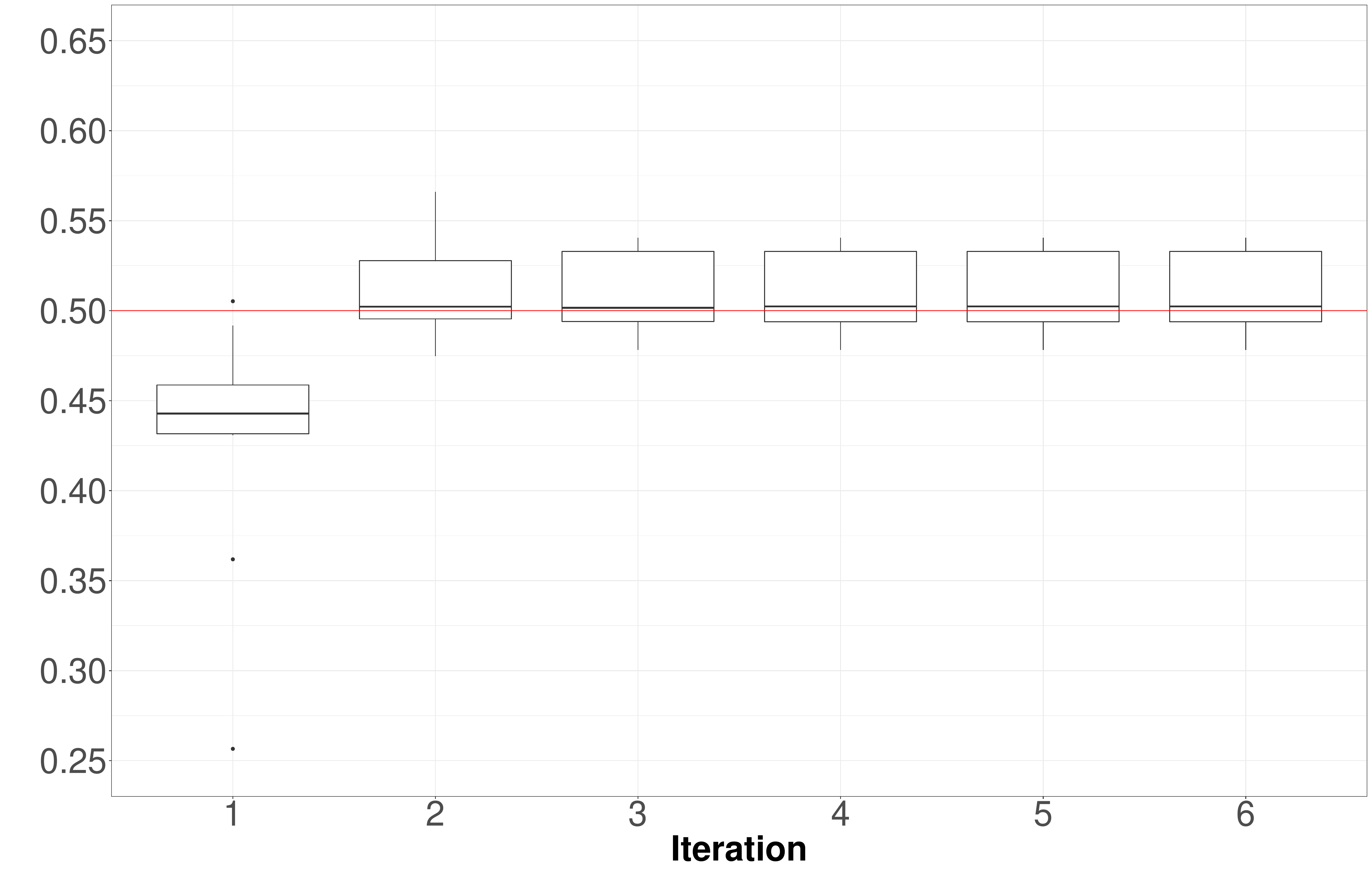}
  &  \includegraphics[width=0.32\textwidth, height=4.5cm]{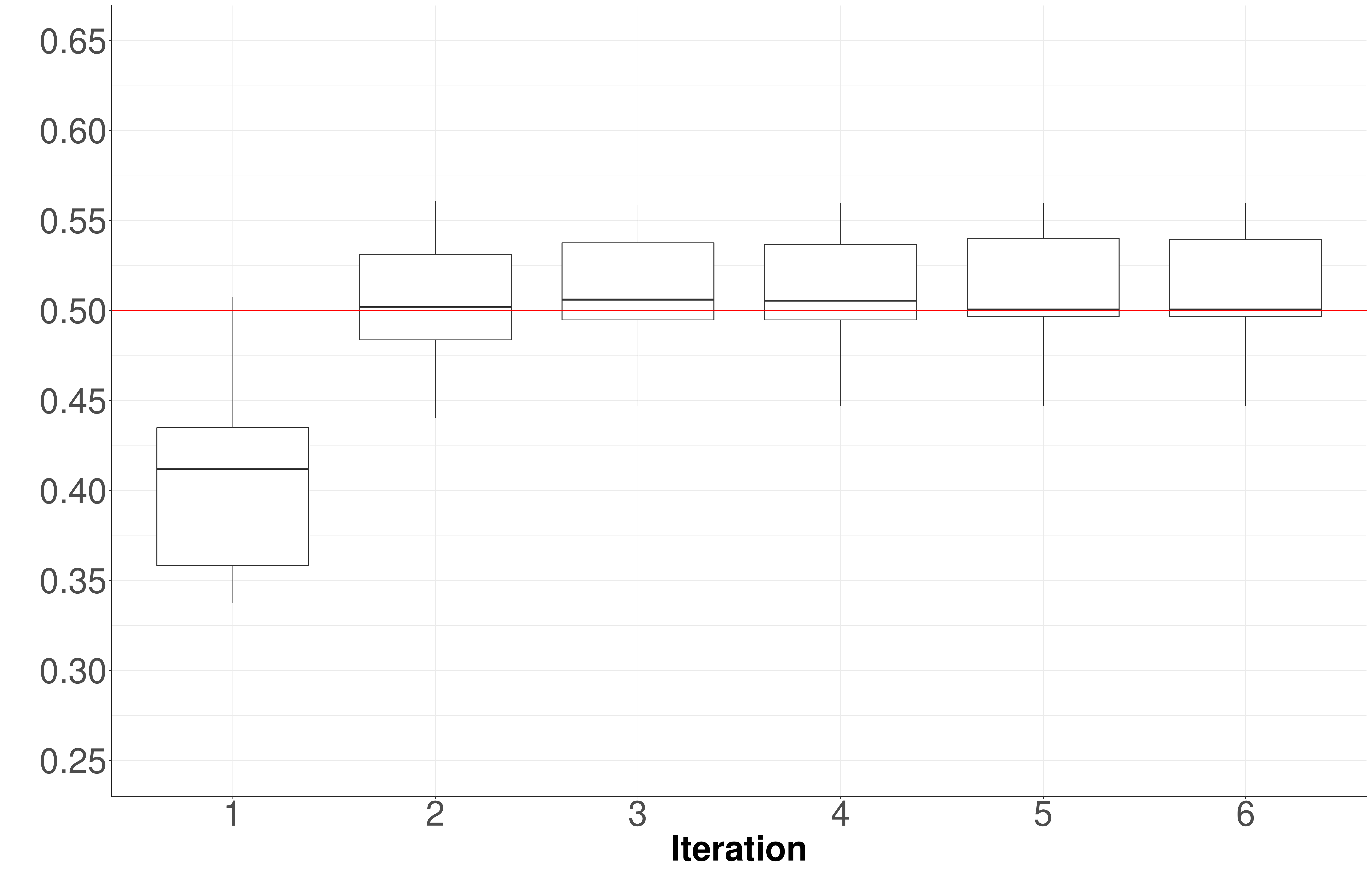} 
  & \includegraphics[width=0.32\textwidth, height=4.5cm]{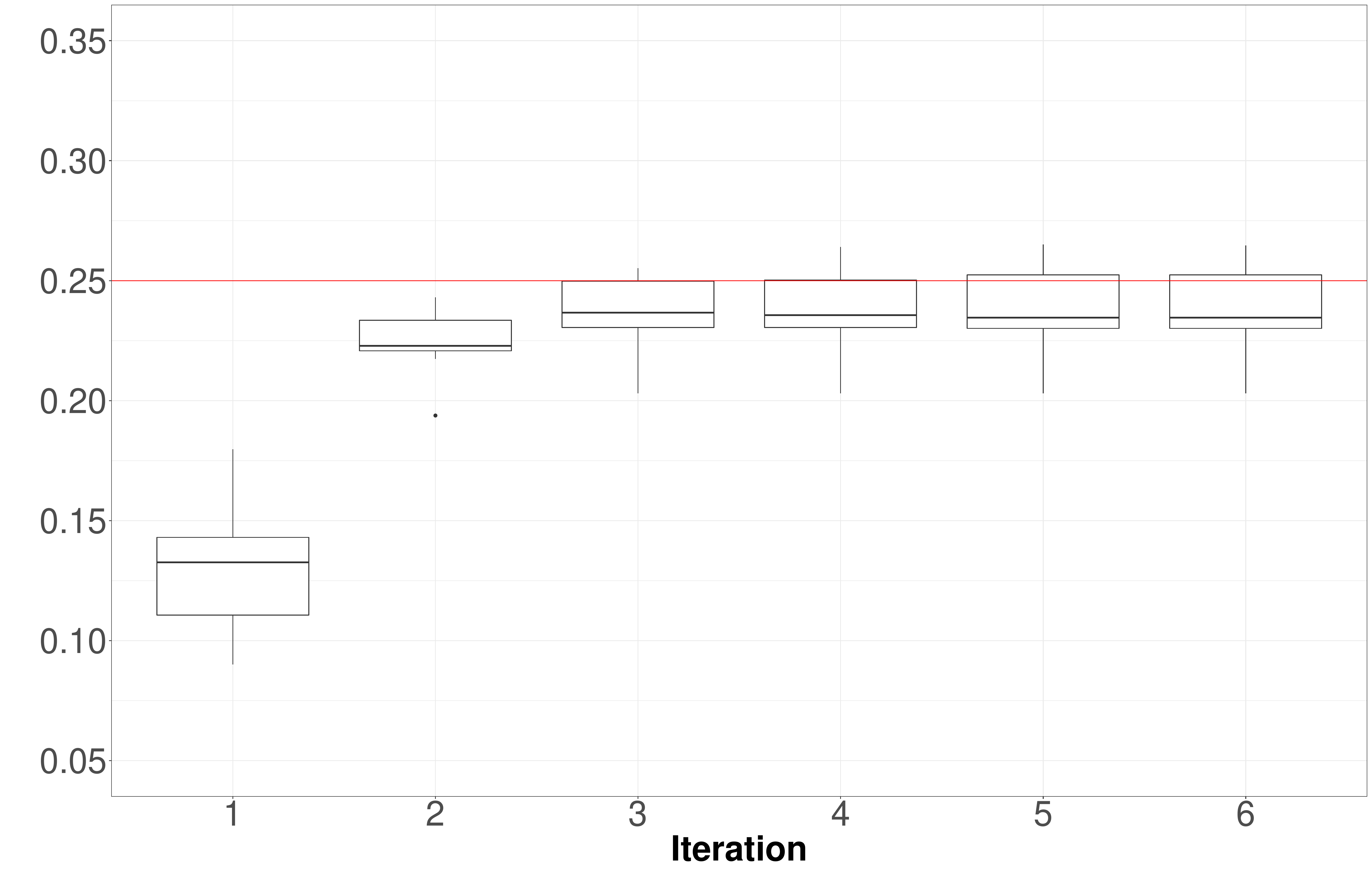}\\
\includegraphics[width=0.32\textwidth, height=4.5cm]{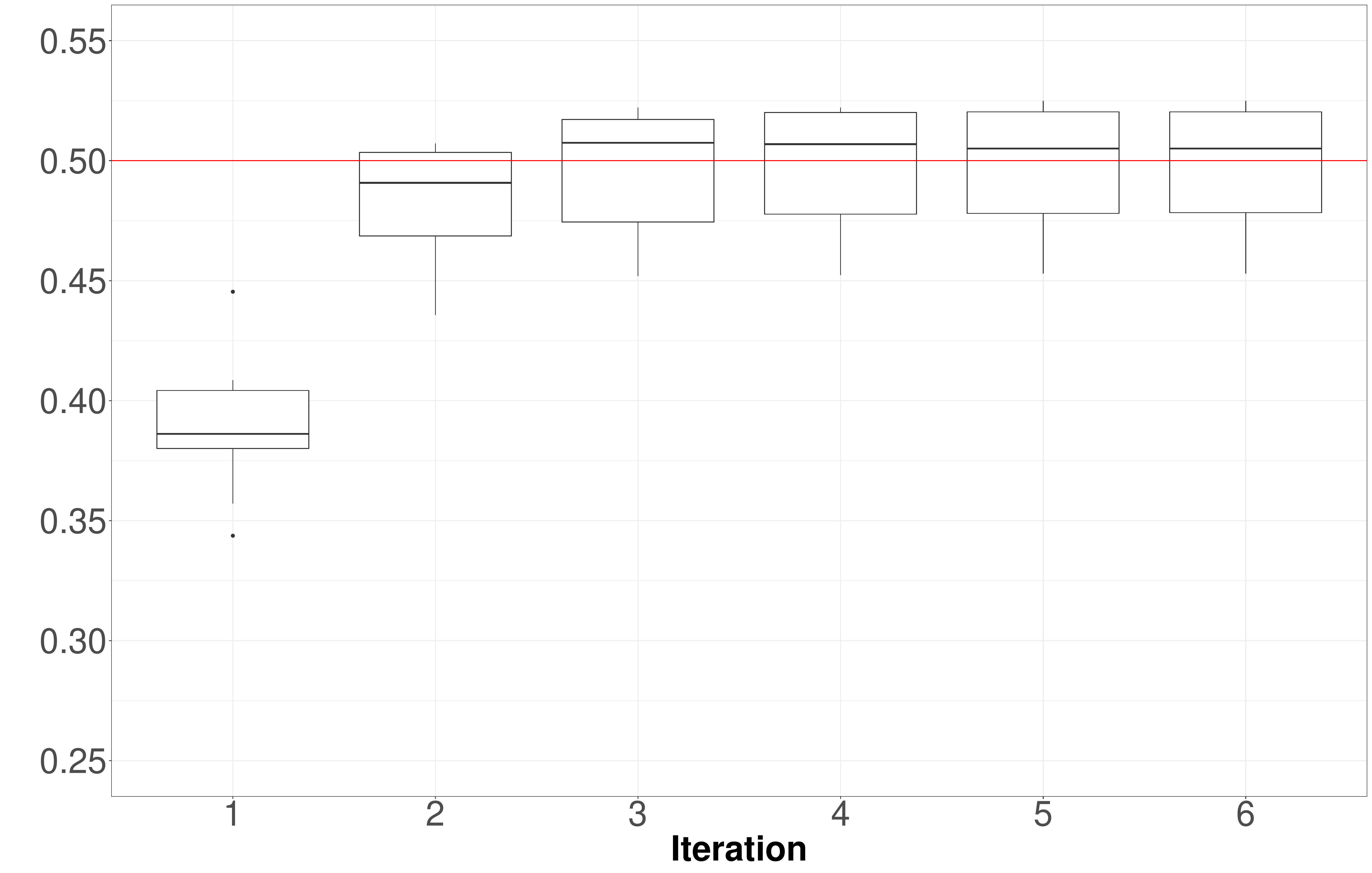}
  &  \includegraphics[width=0.32\textwidth, height=4.5cm]{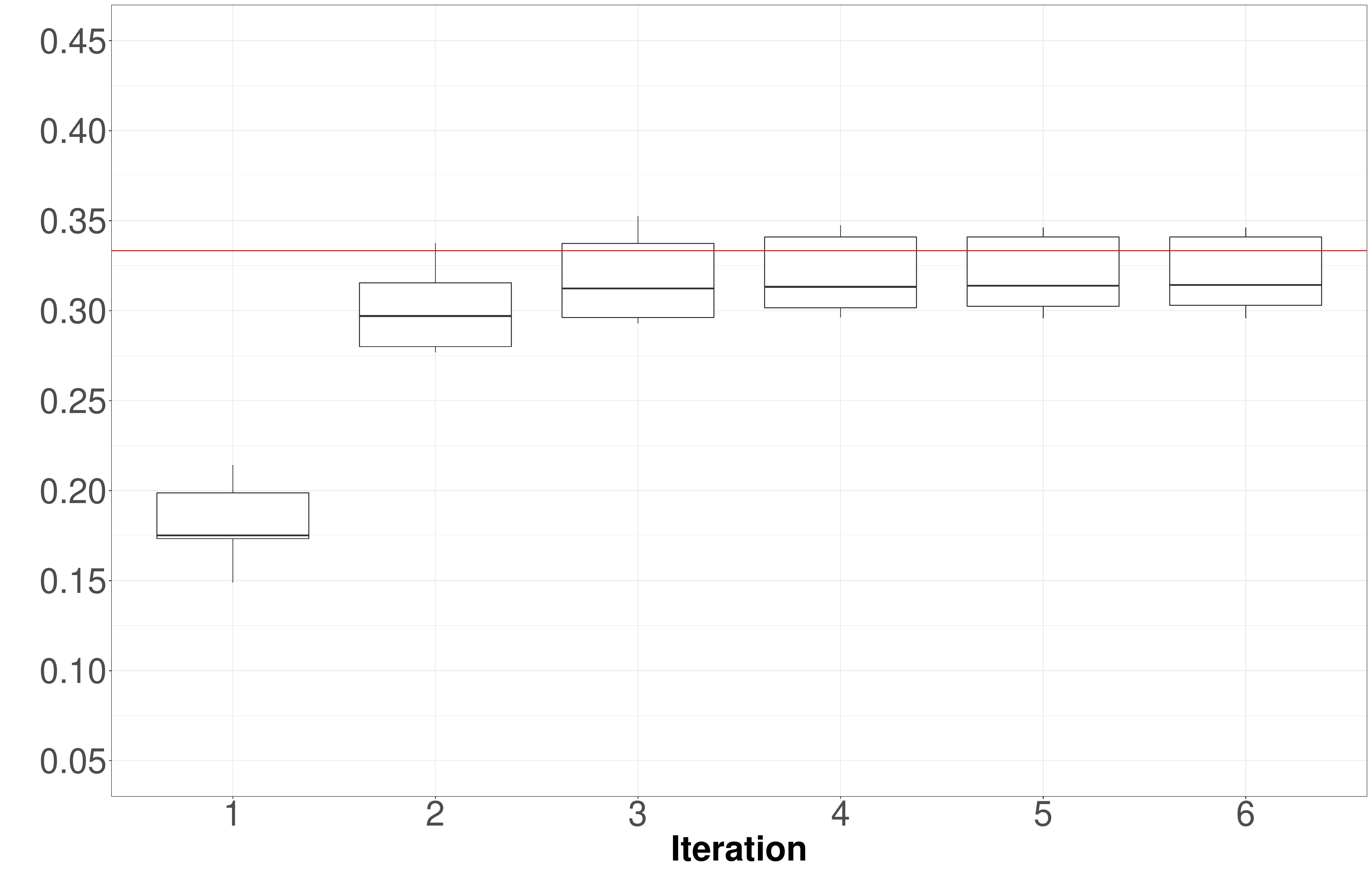} 
  & \includegraphics[width=0.32\textwidth, height=4.5cm]{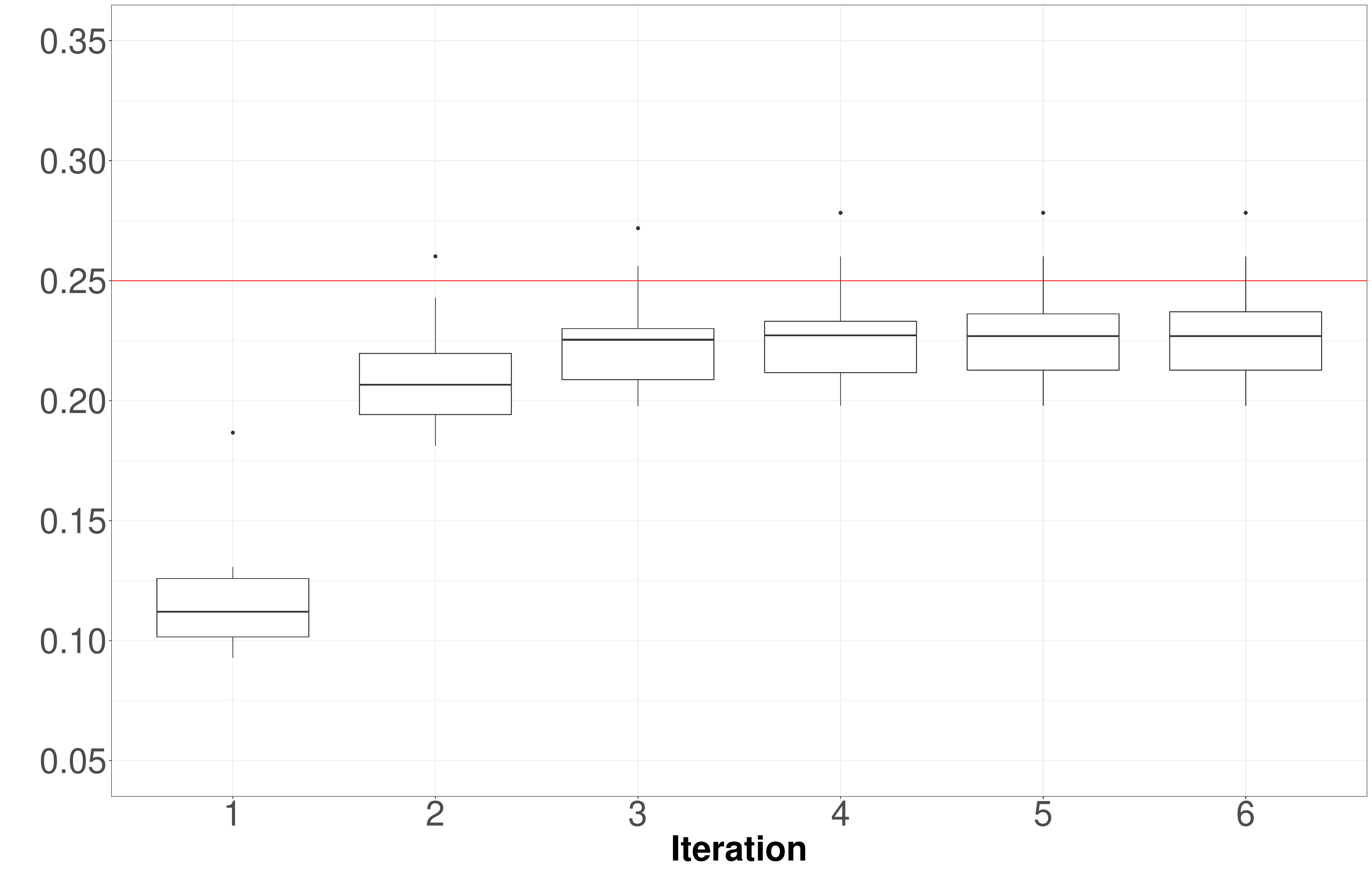}\\
\end{tabular}  
\caption{Boxplots for the estimations of $\boldsymbol{\gamma}^\star$ in Model (\ref{eq:mut_Wt}) with a 10\% sparsity level and $q=1,2,3$ obtained by
  \texttt{fast\_ss}.
 Top: $q=1$ and $\gamma_1^\star=0.5$ (left), $q=2$ and $\gamma_1^\star=0.5$ (middle), $q=2$ and $\gamma_2^\star=0.25$ (right). Bottom: $q=3$ and $\gamma_1^\star=0.5$ (left), $q=3$ and  $\gamma_2^\star=1/3$ (middle), $q=3$ and $\gamma_3^\star=0.25$ (right).
   The horizontal lines correspond to the values of the $\gamma_i^\star$'s.\label{fig:gamma:10:fast}}
 \end{center}
\end{figure}

\begin{figure}[!h]
  \begin{center}
\begin{tabular}{ccc}
  \includegraphics[width=0.32\textwidth, height=4.5cm]{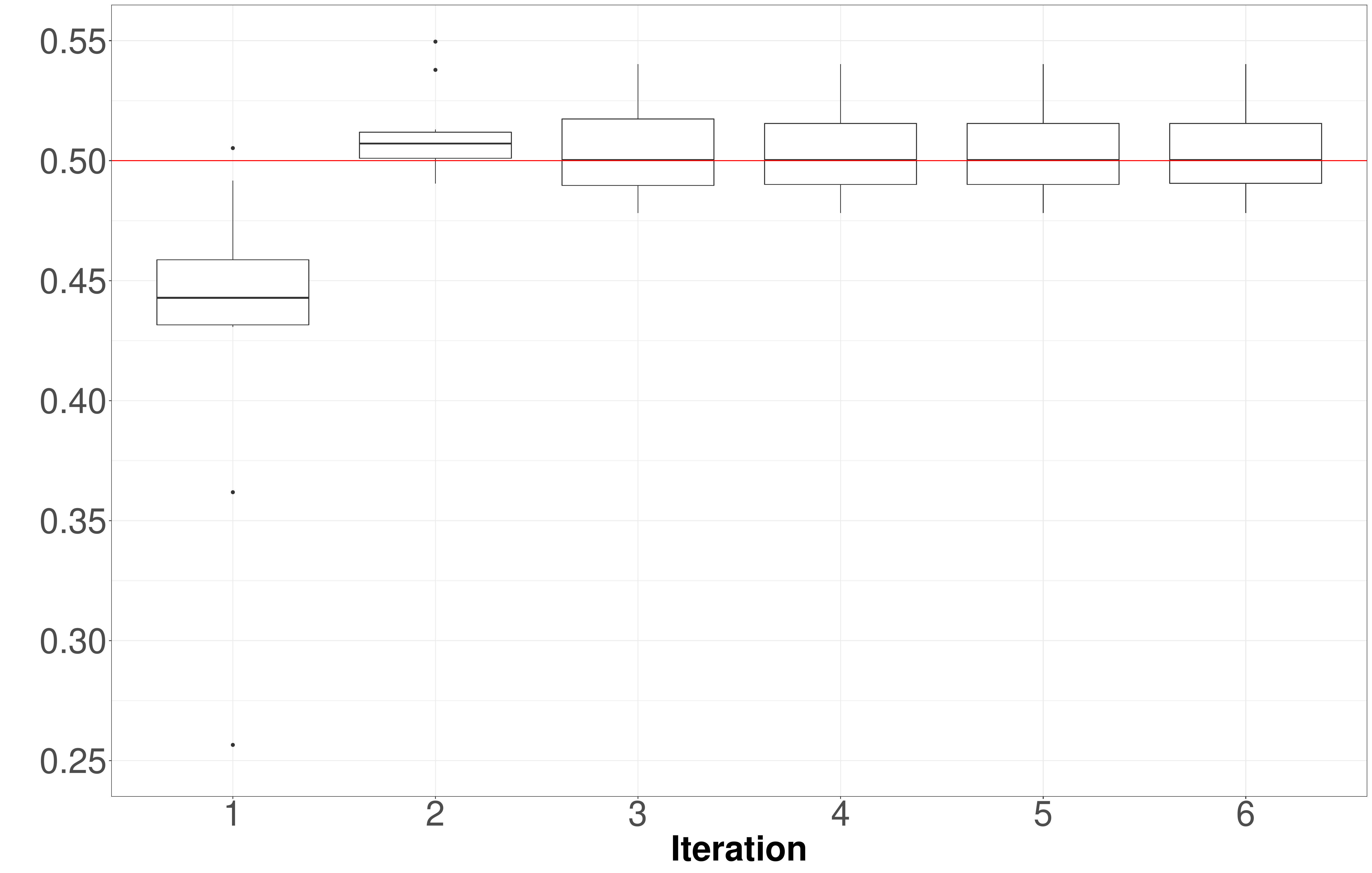}
  &  \includegraphics[width=0.32\textwidth, height=4.5cm]{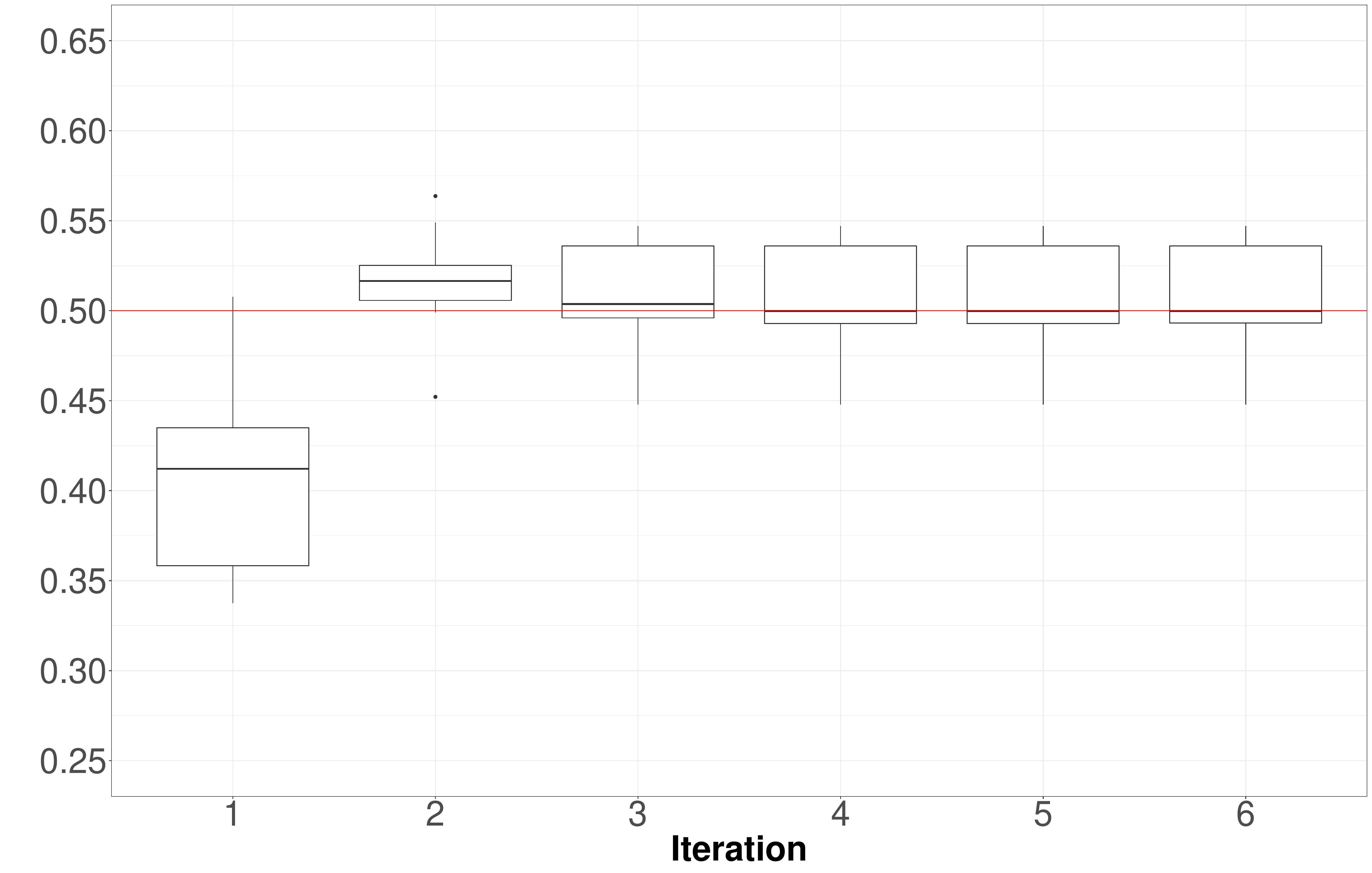} 
  & \includegraphics[width=0.32\textwidth, height=4.5cm]{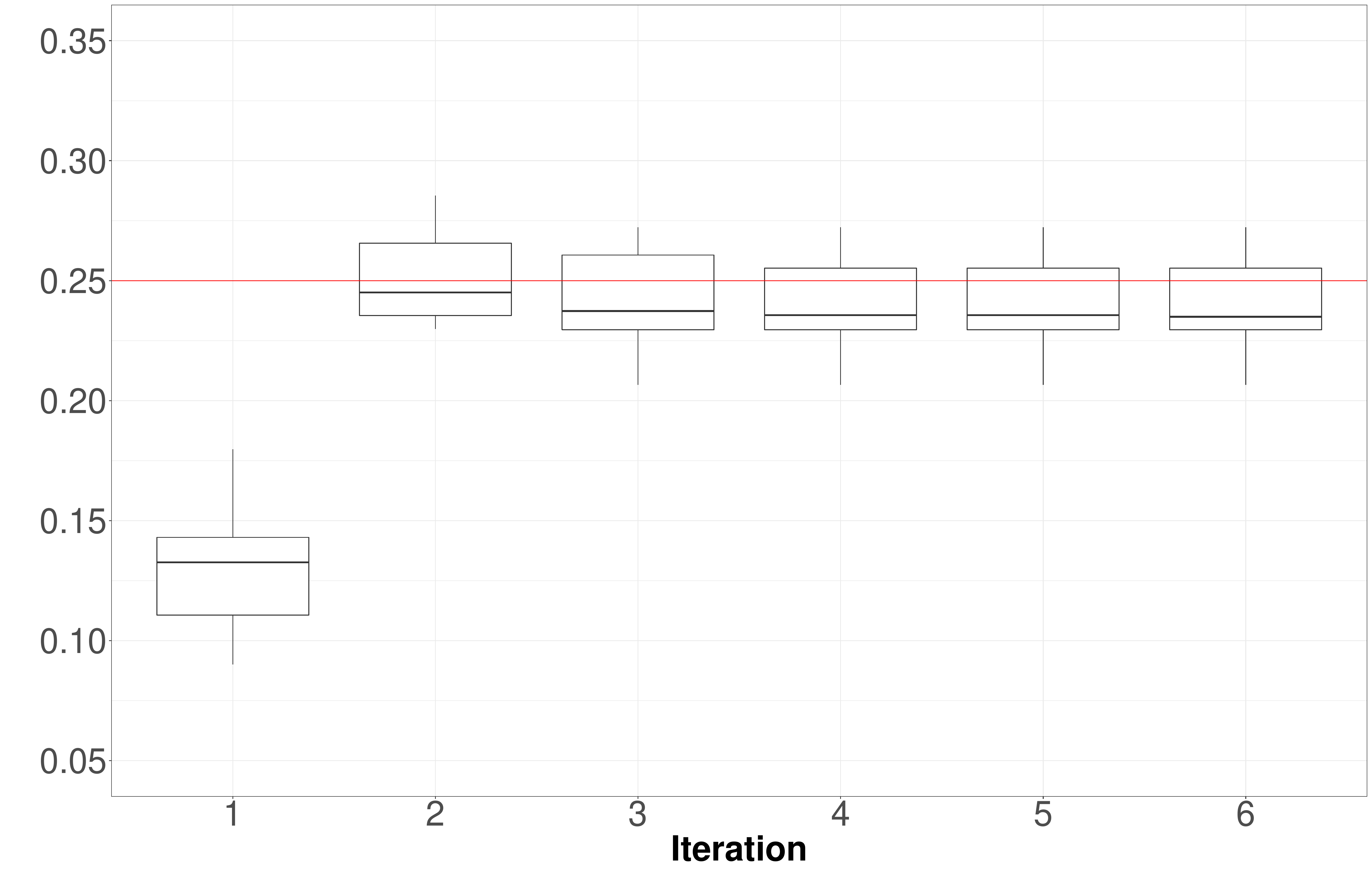}\\
\includegraphics[width=0.32\textwidth, height=4.5cm]{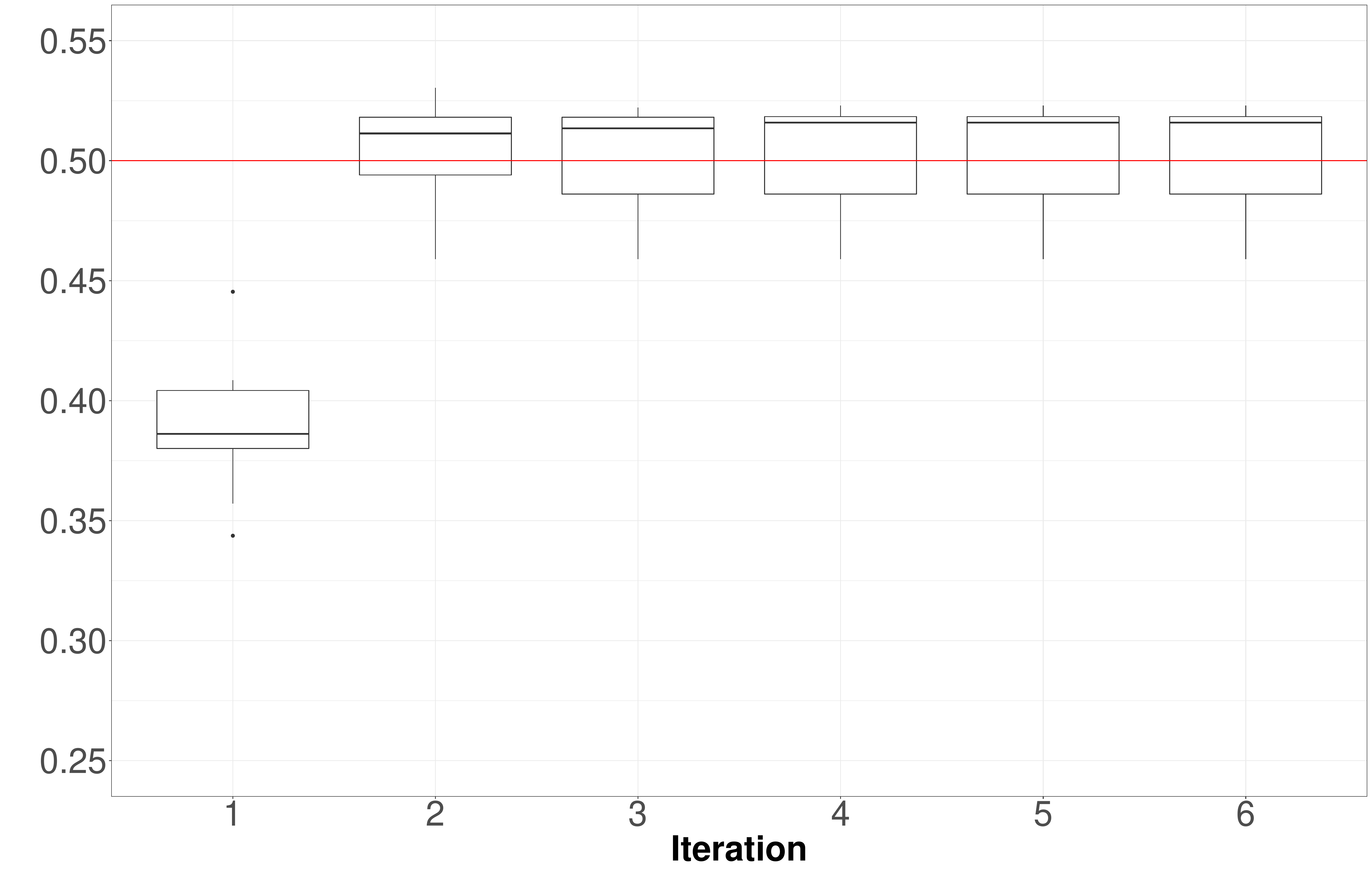}
  &  \includegraphics[width=0.32\textwidth, height=4.5cm]{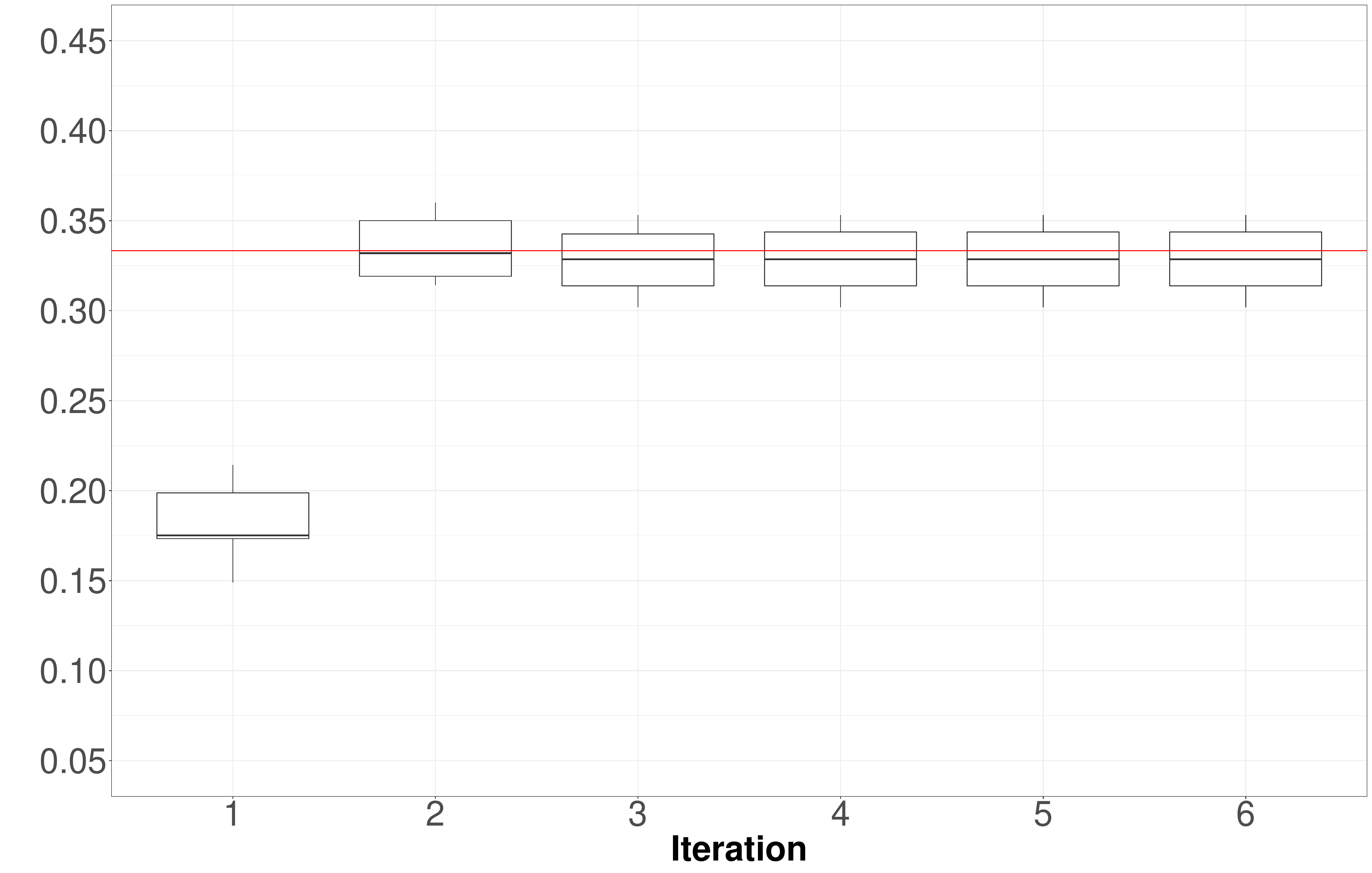} 
  & \includegraphics[width=0.32\textwidth, height=4.5cm]{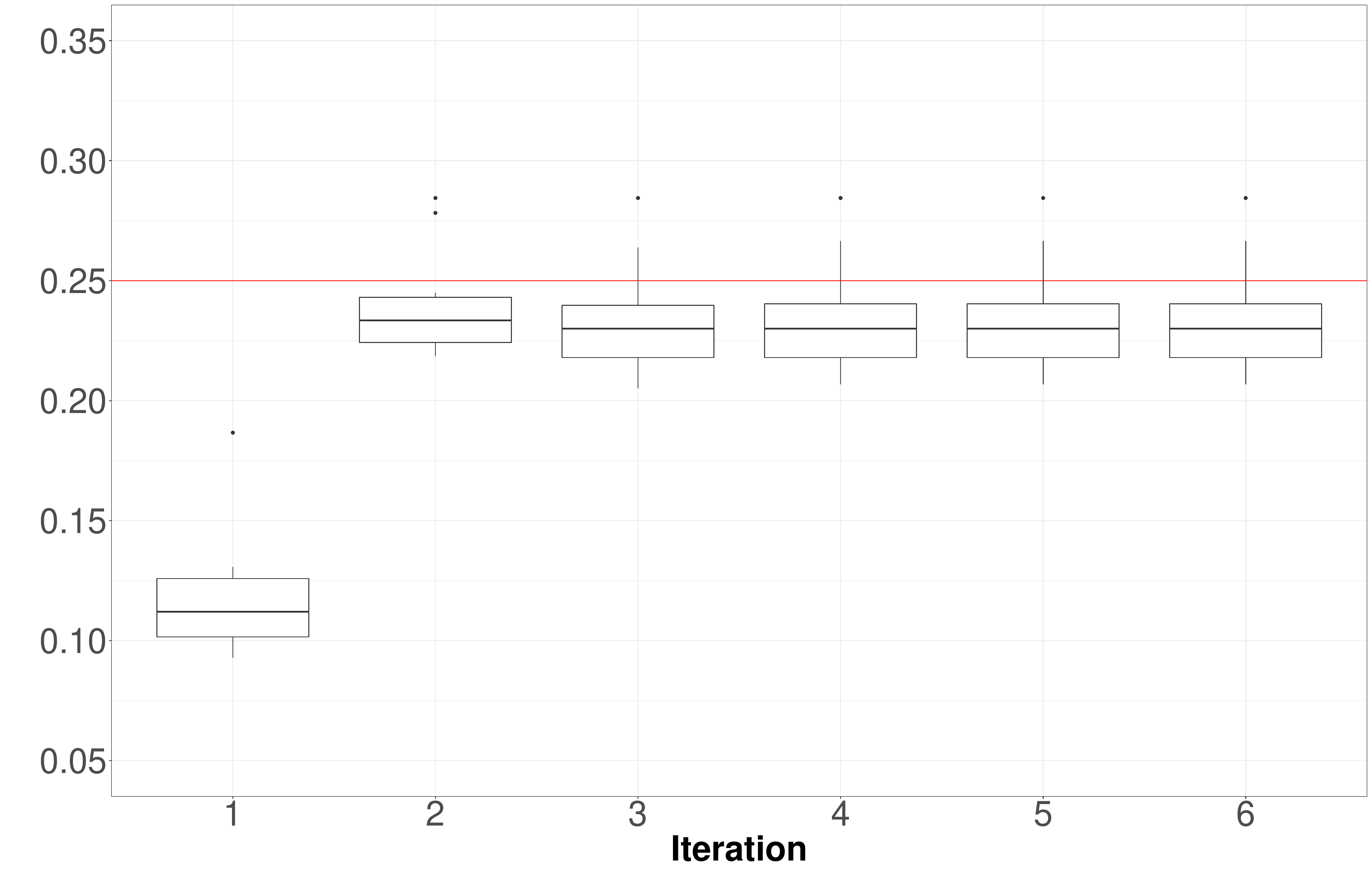}\\
\end{tabular}  
\caption{Boxplots for the estimations of $\boldsymbol{\gamma}^\star$ in Model (\ref{eq:mut_Wt}) with a 10\% sparsity level and $q=1,2,3$ obtained by
  \texttt{ss\_min}.
 Top: $q=1$ and $\gamma_1^\star=0.5$ (left), $q=2$ and $\gamma_1^\star=0.5$ (middle), $q=2$ and $\gamma_2^\star=0.25$ (right). Bottom: $q=3$ and $\gamma_1^\star=0.5$ (left), $q=3$ and  $\gamma_2^\star=1/3$ (middle), $q=3$ and $\gamma_3^\star=0.25$ (right).
   The horizontal lines correspond to the values of the $\gamma_i^\star$'s.\label{fig:gamma:10:min}}
 \end{center}
\end{figure}

\end{document}